%% file: ms.tex
\documentclass{article}
\usepackage[margin=1.12in]{geometry}
\pdfoutput=1

\input{preamble}

\begin{document}

  \maketitle

  \input{00_abstract}

  \input{10_introduction}

  \input{20_preliminary}

  \input{30_K-simulation}

  \input{40_D-simulation}

  \input{50_LO_construction}

  \input{70_conclusion}

  \subsection*{Acknowledgment}
  I am very grateful to my supervisor
  Susumu Nishimura
  for his best support.
  I would also like to thank
  Rihito Takase and Kiichi Hiruma
  for fruitful discussions.

  \bibliography{ref}
  \bibliographystyle{plain}

\end{document}

%% file: preamble.tex
\usepackage{amsmath}

\numberwithin{equation}{section}
\usepackage{amssymb}

\usepackage{amsthm}

\theoremstyle{definition}
\newtheorem{definition}{Definition}[section]
\theoremstyle{plain}
\newtheorem{theorem}[definition]{Theorem}
\newtheorem{proposition}[definition]{Proposition}
\newtheorem{lemma}[definition]{Lemma}
\newtheorem{corollary}[definition]{Corollary}

\theoremstyle{remark}
\newtheorem{remark}[definition]{Remark}
\newtheorem{example}[definition]{Example}

\usepackage{stmaryrd}

\usepackage{url}

\usepackage[inference]{semantic}

\usepackage{graphicx}

\usepackage{tikz}

\usepackage[labelformat=simple]{subcaption}

\input{defs.tex}

\title{
  Determining Existence of Logical Obstructions
  to the Distributed Task Solvability
  \thanks{This article is a revised version of aurthor’s Master’s thesis \cite{Hoshino:Msc22}.}
}
\author{
  Sou Hoshino
  \thanks{Dept.\ of Mathematics, Kyoto University, Japan;
    E-mail: \texttt{hoshino.sou.f76@kyoto-u.jp}}
}

%% file: defs.tex
\usepackage{enumitem}

\newcommand{\C}{\mathcal{C}}

\newcommand{\I}{\mathcal{I}}

\newcommand{\N}{\mathbb{N}}

\newcommand{\facet}[1]{\mathcal{F}({#1})}

\newcommand{\f}[2]{\mathrel{f^{#1}({#2})}}

\newcommand{\IIS}[1]{\mathcal{I}[\mathcal{IS}^{#1}]}
\newcommand{\ISA}[1]{\mathcal{I}[\mathcal{SA}_{#1}]}
\newcommand{\IGr}{\mathcal{I}[\mathcal{G}_{\leq r}]}

\newcommand{\rel}{\facet{M} \times \facet{M'}}
\newcommand{\At}{\mathsf{At}}

\newcommand{\ip}{\mathtt{input}}

\newcommand{\pre}{\mathtt{pre}}

\newcommand{\atom}{\textsf{(Atom)}}
\newcommand{\forth}{\textsf{(Forth)}}
\newcommand{\kforth}{\textsf{(K-Forth)}}
\newcommand{\dforth}{\textsf{(D-Forth)}}
\newcommand{\bforth}{\textsf{($\Box$-Forth)}}

\newcommand{\pfb}{\mathcal{L}_{\Box}^+}
\newcommand{\pf}{\mathcal{L}_K^+}
\newcommand{\pfd}{\mathcal{L}_D^+}

\renewcommand{\S}[1]{S_{#1}^{\Box}}
\newcommand{\SB}[1]{S_{#1}^{\Box}}
\newcommand{\SK}[1]{S_{#1}^{\mathrm{K}}}
\newcommand{\SD}[1]{S_{#1}^{\mathrm{D}}}

\newcommand{\PhiK}{\Phi^{\mathrm{K}}_M}
\newcommand{\PhiD}{\Phi^{\mathrm{D}}_M}
\newcommand{\PhiB}{\Phi^{\Box}_M}

\newcommand{\Views}{\mathsf{Views}}
\newcommand{\car}{carrier}

\usepackage[T1]{fontenc}
\usepackage{slantsc}

\newcommand{\graybullet}{\textcolor{lightgray}{\bullet}}

\newcommand{\Rwg}{R^{\circ, \textcolor{lightgray}{\bullet}}}

%% file: 00_abstract.tex

\begin{abstract}
  To study the distributed task solvability,
  Goubault, Ledent, and Rajsbaum
  devised a model of dynamic epistemic logic
  that is equivalent to the topological model
  for distributed computing.
  In the logical model, the unsolvability of
  a particular distributed task
  can be proven by finding a formula,
  called logical obstruction.
  This logical method is very appealing
  because the concrete formulas that prevent to solve task
  would have implications of intuitive factors for the unsolvability.
  However, it has not been well studied when a logical obstruction exists
  and how to systematically construct a concrete logical obstruction formula,
  if any.
  In addition, it is proved that
  there are some tasks that are solvable
  but do not admit logical obstructions.

  In this paper, we propose a method to prove
  the non-existence of logical obstructions
  to the solvability of distributed tasks,
  based on the technique of simulation.
  Moreover,
  we give a method to determine whether a logical obstruction exists or not
  for a finite protocol and a finite task,
  and if it exists, construct a concrete obstruction.
  Using this method, we demonstrate
  that the language of the standard epistemic logic,
  without distributed knowledge,
  does not admit logical obstruction to
  the solvability of $k$-set agreement tasks.
  We also show that there is no logical obstruction
  for multi-round immediate snapshot even
  in the language of epistemic logic with distributed knowledge.
  In addition, for the \textsc{know-all} model,
  we provide a concrete obstruction formula that shows
  the unsolvability of the $k$-set agreement task.

\end{abstract}

%% file: 10_introduction.tex

\section{Introduction}

For decades, topological models have been used
to study the solvability of tasks in distributed computing~\cite{DBLP:books/mk/Herlihy2013}.
In topological models, the possible configurations
in distributed computation
are represented by a $\Pi$-colored chromatic simplicial complex,
where each vertex is colored by an element taken from $\Pi$,
the set of processes.
This allows the unsolvability of
the $k$-set agreement to be proven
by examining the connectivity of
the simplicial complex that models the protocol~\cite{DBLP:books/mk/Herlihy2013,DBLP:conf/podc/HerlihyR94}.

Recently, Goubault, Ledent, and Rajsbaum
proposed a model of dynamic epistemic logic (DEL), called simplicial model,
as a substitute for the conventional topological model~\cite{DBLP:journals/iandc/GoubaultLR21}.
In the logical model, the unsolvability of a task $T$ by a protocol $P$
is shown as follows:
1.\ Find a formula $\varphi$ that expresses
some property that does not hold
in $\mathcal{I}[P]$, a simplicial model for the protocol,
but holds in $\mathcal{I}[T]$, a simplicial model for the task.
2.\ Show that $\varphi$ is true in $\mathcal{I}[T]$
but is false in $\mathcal{I}[P]$.
3.\ Apply the knowledge gain theorem to
conclude that the task cannot be solved by the protocol.
Thus, to prove the unsolvability,
we only need to find an appropriate formula $\varphi$.
We call such a formula \emph{logical obstruction}.
This logical method is very attractive
because it provides a concrete formula that prevents to solve a task,
which would contribute to better understanding of unsolvability.
In addition, the above procedures 1-3 could be completed
without resorting to sophisticated topological tools,
such as homotopy group, Nerve lemma, and so on.
However, there is no concrete way to obtain logical obstructions,
and furthermore, there may be no logical obstruction.
Goubault et al.~\cite{DBLP:conf/tap/GoubaultLLR19} showed that
a certain unsolvable distributed task
has no logical obstruction,
applying the technique of bisimulation.

So far, concrete instances of logical obstruction
have been devised for a limited variations of distributed tasks.
Goubault et al.\ provided instances of logical obstructions
for simple tasks such as consensus (or 1-set agreement)
and approximate agreement tasks
but they left logical obstruction to
general $k$-set agreement task as open problem~\cite{DBLP:journals/iandc/GoubaultLR21}.
Nishida devised a logical obstruction to
general $k$-set agreement
by extending the logic with modality of distributed knowledge~\cite{Nishida:Msc20}.
Yagi and Nishimura~\cite{DBLP:journals/corr/abs-2011-13630}
applied Nishida's obstruction to show the unsolvability of
$k$-set agreement tasks in the superset closed adversary model.
However, their logical obstructions to the solvability of $k$-set agreement tasks
work for only single-round protocols
but not for multiple-round protocols.
This is in contrast to the topological method,
which can handle multi-round protocols~\cite{DBLP:conf/podc/HerlihyR10,DBLP:journals/jacm/HerlihyS99}.
The concrete instances mentioned above seem to indicate
that the epistemic logic is so weak that
the unsolvability is in general hard to
be expressed as an obstruction.
Indeed there exists an unsolvable task that has
no logical obstruction~\cite{DBLP:conf/tap/GoubaultLLR19}.

In this study, we are concerned with
the ability of epistemic logic in expressing the unsolvability.
Particularly, we study two languages of epistemic logic,
$\mathcal{L}_K$ and $\mathcal{L}_D$, and
how the additional expressibility of the latter
differentiates the ability.
The language $\mathcal{L}_K$ is the epistemic logic
with the standard epistemic modality $\mathrm{K}_a$,
a.k.a.\ knowledge operator;
The language $\mathcal{L}_D$ extends $\mathcal{L}_K$
with additional modality $\mathrm{D}_A$,
called distributed knowledge.
The logical obstructions given in~\cite{Nishida:Msc20,DBLP:journals/corr/abs-2011-13630}
suggest that the distributed knowledge provides
more logical obstructions to unsolvable tasks,
but it does not preclude that
$\mathcal{L}_K$ and $\mathcal{L}_D$ are equally able to define
logical obstructions for unsolvable tasks.

In order to answer this question and
find a systematic way for constructing logical obstructions,
we apply the technique of \emph{simulation}
over simplicial models to argue the existence of logical obstructions.
Simulation is a weaker notion of bisimulation between simplicial models.
As we will show in a proceeding section,
the knowledge gain theorem can be extended for simulations.
Moreover, simulation gives a method to determine
whether a logical obstruction exists or not
for a finite protocol and a finite task,
and if it exists, we provide a procedure for
constructing a logical obstruction.

The contributions of this paper are the followings.
\begin{itemize}
  \item Simulation between simplicial models is introduced
    as a means to show that
    the usual epistemic language $\mathcal{L}_K$ admits no logical obstruction.
  \item Simulation is extended to incorporate distributed knowledge
    and provides us a method for proving
    the non-existence of logical obstruction
    in the epistemic language $\mathcal{L}_D$
    with distributed knowledge operator.
  \item There exists a procedure that determines
    the existence of logical obstruction
    for a given pair of a finite protocol and a finite task.
    Moreover, in case such a logical obstruction exists,
    the procedure further constructs
    a concrete obstruction formula systematically.
\end{itemize}
We will show the following facts to demonstrate the merit of our approach.
\begin{itemize}
  \item The language $\mathcal{L}_K$, without distributed knowledge,
    admits no logical obstruction for $k$-set agreement tasks and
    immediate snapshot protocol, where $k \geq 2$.
    Therefore, $\mathcal{L}_D$ admits properly larger class of
    logical obstructions than $\mathcal{L}_K$.
  \item The language $\mathcal{L}_D$, with distributed knowledge,
    admits no logical obstruction for $2$-set agreement task and
    multi-round immediate snapshot protocol, where $|\Pi| = 3$.
    Therefore, although $\mathcal{L}_D$ admits more logical obstructions
    than $\mathcal{L}_K$,
    it still fails to establish the unsolvability,
    which can be topologically proven.
  \item We generate a concrete logical obstruction for
    $k$-set agreement tasks and a protocol in the \textsc{know-all} model~\cite{DBLP:journals/tcs/CastanedaFPRRT21}.
\end{itemize}

\begin{table}
  \caption{Existence of logical obstructions
    for varying tasks, protocols and languages of epistemic logic}
  \label{table_related_works_result}
  \centering
  \begin{tabular}{p{0.23\linewidth}p{0.23\linewidth}p{0.43\linewidth}}
    \hline
      Task & Protocol & Existence of logical obstruction \\
    \hline
    \hline
      Consensus & Immediate snapshot &
        exist in $\mathcal{L}_{CK}$~\cite{DBLP:journals/iandc/GoubaultLR21} \\[1ex]
      Approximate agreement & Immediate snapshot &
        exist in $\mathcal{L}_K$ unless solvable~\cite{DBLP:journals/iandc/GoubaultLR21} \\[1ex]
      $k$-set agreement where $2 \leq k < |\Pi|$ &
        Single-round immediate snapshot &
          exist in $\mathcal{L}_D$~\cite{Nishida:Msc20},
          but not exist in $\mathcal{L}_K$ [This paper] \\[1ex]
      & Multi-round immediate snapshot &
        not exist in $\mathcal{L}_K$ and
        not exist when $|\Pi| = 3$ in $\mathcal{L}_D$ [This paper] \\[1ex]
      & Instance of \textsc{know-all} model &
        exist in $\mathcal{L}_D$ unless solvable [This paper] \\[1ex]
      Equality negation with two agents and three input values &
      Layered message passing &
      not exist in $\mathcal{L}_{CK}$~\cite{DBLP:conf/tap/GoubaultLLR19} \\[1ex]
    \hline
      \multicolumn{3}{l}{$\mathcal{L}_K$:
          the language of epistemic logic
          with modality of knowledge but without distributed knowledge} \\
      \multicolumn{3}{l}{$\mathcal{L}_D$:
          the language of epistemic logic with distributed knowledge} \\
      \multicolumn{3}{l}{$\mathcal{L}_{CK}$:
          the language of epistemic logic
          with modality of knowledge and common knowledge} \\
  \end{tabular}
\end{table}
The Table~\ref{table_related_works_result} summarizes
the results of prior works and the present paper.

We notice that the present paper concerns
epistemic logic whose atomic propositions can mention
solely input values.
Epistemic logic with an augmentedset of atomic proposition
such as the one proposed in~\cite{DBLP:journals/jlap/DitmarschGLLR21}
which is more expressive in defining logical obstruction,
is out of the scope of the paper.

The rest of the paper is organized as follows.
Section~\ref{sec_preliminary} defines
distributed tasks, protocols, and its solvability
with simplicial models of DEL.
In Section~\ref{sec_k-simulation},
we introduce the notion of simulation
and show that there is no logical obstruction
for $k$-set agreement tasks and
immediate snapshot protocol in $\mathcal{L}_K$.
In Section~\ref{sec_d-simulation},
we reinforce the simulation with distributed knowledge
and show that there is no logical obstruction
for $2$-set agreement task and
multi- round immediate snapshot protocol in $\mathcal{L}_D$
when $|\Pi| = 3$.
In Section~\ref{sec_lo_construction},
we give a procedure that determines whether a simulation exists or not
and further, in case there is no simulation, produce a logical obstruction.
Section~\ref{sec_conclusion} concludes the paper.

%% file: 20_preliminary.tex

\section{Preliminaries} \label{sec_preliminary}

Throughout the paper, we fix a non empty finite set $\Pi$
for colors, or agents.
We write $\N$ to denote the set of non-negative integers.

\subsection{Simplicial models}

In the topological theory of distributed computing,
distributed systems are modeled by simplicial complexes.
In~\cite{DBLP:journals/iandc/GoubaultLR21},
it is shown that simplicial complexes provide
Kripke models that is suitable for analysis by epistemic logic.

\begin{definition}[Chromatic simplicial complex]
  A \emph{chromatic simplicial complex} $\C = (V, S, \chi)$
  is a triple consisting of
  a set $V$ of vertices,
  a set $S$ of non-empty finite subsets of $V$ and
  a coloring map $\chi \colon V \to \Pi$ such that
  $S$ and $\chi$ satisfy the following conditions:
  \begin{itemize}
    \item $S$ is closed under inclusion,
      i.e.\ $X \in S$ and $\emptyset \subsetneq Y \subseteq X \implies Y \in S$.
    \item For any vertex $v \in V$, $\{ v \} \in S$.
    \item For any simplex $X \in S$, $\chi|_X \colon X \to \Pi$ is injective.
  \end{itemize}

  Each element of $S$ is called a \emph{simplex}.
  We write $X \in \C$ to mean $X$ is a simplex of $C$,
  i.e.\ $X \in S$.
  The \emph{dimension} of a simplex $X \in S$
  is defined by $|X| - 1$.
  A simplex that is maximal with respect to
  inclusion is called a \emph{facet}.
  The set of facets is denoted by $\facet{\C}$.
  A complex is called \emph{pure} if all its facets have the same dimension.

\end{definition}

In the following, Let
$\At = \{\ip_a^v \mid \text{$a \in \Pi$, $v \in V^{in}$} \}$
be a set of atomic propositions,
where $V^{in}$ is an arbitrary set of input values.
For each $a \in \Pi$,
$\At_a = \{\ip_a^v \in \At \mid v \in V^{in} \}$.

\begin{definition}[Simplicial model]
  A \emph{simplicial model} $M = (V, S, \chi, l)$ is
  a pure ($|\Pi| - 1$)-dimensional
  chromatic simplicial complex $\C = (V, S, \chi)$
  equipped with a labeling map $l \colon V \to \mathcal{P}(\At)$
  such that $l(v) \subseteq \At_{\chi(v)}$.

  For each $a \in \Pi$, we define an equivalence relation
  $\sim_a$ over $\facet{\C}$,
  called an \emph{indistinguishability relation},
  by $X \sim_a Y$ if and only if $a \in \chi(X \cap Y)$.
  Then $(\facet{C}, (\sim_a)_{a \in \Pi}, l)$ gives
  a multi-agent Kripke model of facets being the set of possible worlds.

  In abuse of notation, we may write
  $\facet{M}$ to mean a set of facets $\facet{C}$
  in its underlying simplicial complex.

\end{definition}

\begin{definition}[Morphism between simplicial models]
  Let $M = (V, S, \chi, l)$, $M' = (V', S', \chi', l')$ be simplicial models.
  A \emph{morphism} $f \colon M \to M'$ is
  a map $f$ from $V$ to $V'$ that satisfies the following conditions:
  \begin{itemize}
    \item For all $X \in S$, $f(X) \in S'$,
    \item $f$ preserves the coloring,
      i.e.\ $\chi'(f(v))= \chi(v)$ for any $v \in V$ and
    \item $f$ preserves the labeling,
      i.e.\ $l'(f(v))= l(v)$ for any $v \in V$.
  \end{itemize}

\end{definition}

\begin{definition}[Input simplicial model]
  An \emph{input simplicial model} is
  a simplicial model $\I = (V, S, \chi, l)$
  where
  \begin{itemize}
    \item $V \subseteq \Pi \times V^{in}$,
    \item $X \in S$ if $X$ have at most one $a$-colored vertex
      for each color $a \in \Pi$,
    \item $\chi(a,v) = a$ and
    \item $l(a,v) = \{ \ip_a^v \}$.
  \end{itemize}

\end{definition}

In the sequel, we assume $V = \Pi \times V^{in}$.
Figure~\ref{fig_input} shows an example of an input simplicial model,
where each color of vertex is illustrated as a color of node.

\begin{figure}
  \centering
  \begin{minipage}[b]{0.45\linewidth}
    \centering
    \input{figures/10_example_input_complex}
  \end{minipage}
  ~~~
  \begin{minipage}[b]{0.45\linewidth}
    \centering
    \input{figures/20_1-set_agreement_task}
  \end{minipage}
\end{figure}

\subsection{Epistemic logic for simplicial models}

\begin{definition}[Epistemic formula]
  We define the language of epistemic formulas $\mathcal{L}_K$ as follows:
  \[
    \varphi ::= p \mid \lnot \varphi
      \mid (\varphi \land \varphi) \mid \mathrm{K}_a \varphi
  \]
  where $p \in \At$ is an atomic proposition and
  $a \in \Pi$ is an agent.
\end{definition}

\begin{definition}
  Given a formula $\varphi \in \mathcal{L}_K$ and
  a facet $X \in \facet{M}$ of a simlicial model $M$,
  we define the \emph{truth} of $\varphi$ at $X$,
  written $M, X \models \varphi$, as below by induction on $\varphi$.
  \begin{align*}
    &M, X \models p
      && \text{if $p \in l(X)$}. \\
    &M, X \models \lnot\varphi
      && \text{if $M, X \not\models \varphi$}. \\
    &M, X \models \varphi_1 \land \varphi_2
      && \text{if $M, X \models \varphi_1$ and $M, X \models \varphi_2$}. \\
    &M, X \models \mathrm{K}_a\varphi
      && \text{if for all Y $\in \facet{M}$,
        $Y \sim_a X \implies M, Y \models \varphi$}.
  \end{align*}

\end{definition}

\subsection{Protocols, tasks and solvability with DEL}

Protocols and tasks in distributed computing are regarded as
transformation of the input configuration to the output.
In simplicial models, this can be defined as product update.

\begin{definition}[Simplicial action model and product update]
  A \emph{simplicial action model},
  (or an \emph{action model} for short)
  $A = (V, S, \chi, \pre)$ is
  a pure ($|\Pi| - 1$)-dimensional
  chromatic simplicial complex $\C = (V, S, \chi)$
  equipped with a map $\pre \colon \facet{\C} \to \mathcal{L}_K$.
  Each facet in $\facet{\C}$ is called an \emph{action} and
  $\pre$ assigns a \emph{precondition} formula to each action.
  In abuse of notation, we may write
  $\facet{A}$ to mean a set of facets $\facet{C}$
  in its underlying simplicial complex.

  Let $M = (V_M, S_M, \chi_M, l_M)$ be a simplicial model and
  $A = (V_A, S_A, \chi_A, \pre)$ be a simplicial action model.
  A \emph{product update model} $M[A] = (V, S, \chi, l)$ is
  a simplicial model defined as follows:
  \begin{itemize}
    \item $\facet{M[A]} = \{X \times_{\Pi} T \mid \text{$X \in \facet{M}$,
      $T \in \facet{A}$, $M, X \models \pre(T)$}\}$,
    \item $\chi(m, a) = \chi_M(m) = \chi_A(a)$ and
    \item $l(m, a) = l_M(m)$,
  \end{itemize}
  where $X \times_{\Pi} T =
  \{(m, a) \in X \times T \mid \chi_M(m) = \chi_A(a)\}$
  is a subset of the product $X \times T$
  whose each element is a tuple of
  a vertex in $X$ and a vertex in $T$ with the same color.

\end{definition}

\begin{definition}[Task]
  A simplicial action model $A = (V, S, \chi, \pre)$ is called a \emph{task}
  if there exists an injection
  $\iota \colon \facet{A} \to [\Pi, V^{out}]$ such that
  $T \sim_a T'$ holds if and only if $\iota(T)(a) = \iota(T')(a)$
  for any $T, T' \in \facet{A}$ and $a \in \Pi$,
  where $V^{out}$ is an arbitrary set of output values,
  $[\Pi, V^{out}]$ is the set of functions from $\Pi$ to $V^{out}$.

\end{definition}

\begin{example}[$k$-set agreement task]
  We define a task $\mathcal{SA}_k$, where $k = 1, \dots , |\Pi|$,
  called $k$-set agreement as follows:
  \begin{itemize}
    \item $V^{out} = V^{in}$,
    \item $V = \Pi \times V^{out}$,
    \item $\facet{\mathcal{SA}_k} =
      \{\{(a, decide(a)) \mid a \in \Pi\} \mid
        decide \in [\Pi, V^{out}], |decide(\Pi)| \leq k$\},
    \item $\chi(a, d) = a$ and
    \item $\pre(\{(a, decide(a)) \mid a \in \Pi\})
      = \bigwedge_{a \in \Pi} \bigvee_{b \in \Pi} \ip_b^{decide(a)}$.
  \end{itemize}

  This is a well-defined task because there is an injection
  $\{(a, decide(a)) \mid a \in \Pi\} \mapsto decide$.
  In what follows, we write $(a, i_a, d_a)$ to mean a vertex
  $((a, i_a), (a, d_a))$ of $\ISA{k}$.
  we assume the set of input values for the $k$-set agreement
  is given by $V^{in} = \Pi$ unless otherwise stated.
  The product update model $\ISA{1}$ is as shown in Figure~\ref{fig_isa1},
  where each color of vertex is illustrated as a color of node
  and the input $i$ and the decision $d$ is indicated as $i \mapsto d$.

\end{example}

Using product updates, solvability of distributed tasks is defined as follows.

\begin{definition}[Solvability]
  Let $\mathcal{I}$ be an input simplicial model.
  A task $T$ is \emph{solvable} by an action model $P$,
  if there exists a morphism $f \colon \I[P] \to \I[T]$.

\end{definition}

%% file: figures/10_example_input_complex.tex

\begin{tikzpicture}
  \node[circle,scale=0.6,draw,fill=white,label=below:{$0$}] (c0v0) at (0, 0) {};
  \node[circle,scale=0.6,draw,fill=white,label=above:{$1$}] (c0v1) at (1, 1) {};
  \node[circle,scale=0.6,draw,fill=black,label=below:{$0$}] (c1v0) at (1, 0) {};
  \node[circle,scale=0.6,draw,fill=black,label=above:{$1$}] (c1v1) at (0, 1) {};
  \draw (c0v0) -- (c1v0) -- (c0v1) -- (c1v1) -- (c0v0);
\end{tikzpicture}
\caption{Example of an input simplicial model for
  $\Pi = \{\circ, \bullet\}$,
  $V^{in} = \{0, 1\}$.}
\label{fig_input}

%% file: figures/20_1-set_agreement_task.tex

\begin{tikzpicture}
  \node[circle,scale=0.6,draw,fill=white,label=below:{$0 \mapsto 0$}] (c0v0d0) at (0, 0) {};
  \node[circle,scale=0.6,draw,fill=white,label=above:{$1 \mapsto 0$}] (c0v1d0) at (1.5, 1) {};
  \node[circle,scale=0.6,draw,fill=black,label=below:{$0 \mapsto 0$}] (c1v0d0) at (1.5, 0) {};
  \node[circle,scale=0.6,draw,fill=black,label=above:{$1 \mapsto 0$}] (c1v1d0) at (0, 1) {};
  \node[circle,scale=0.6,draw,fill=white,label=below:{$0 \mapsto 1$}] (c0v0d1) at (3, 0) {};
  \node[circle,scale=0.6,draw,fill=white,label=above:{$1 \mapsto 1$}] (c0v1d1) at (4.5, 1) {};
  \node[circle,scale=0.6,draw,fill=black,label=below:{$0 \mapsto 1$}] (c1v0d1) at (4.5, 0) {};
  \node[circle,scale=0.6,draw,fill=black,label=above:{$1 \mapsto 1$}] (c1v1d1) at (3, 1) {};
  \draw (c1v1d0) -- (c0v0d0) -- (c1v0d0) -- (c0v1d0);
  \draw (c0v0d1) -- (c1v1d1) -- (c0v1d1) -- (c1v0d1);
\end{tikzpicture}
\caption{$\I[\mathcal{SA}_1]$ for
  $\Pi = \{\circ, \bullet\}$,
  $V^{in} = \{0, 1\}$, $V^{out} = \{0, 1\}$.}
\label{fig_isa1}

%% file: 30_K-simulation.tex

\section{Non-existence of logical obstructions} \label{sec_k-simulation}

In this section, we consider the following language $\pf$
of positive epistemic formulas:
\[
  \varphi ::= p \mid \lnot p \mid (\varphi \lor \varphi)
    \mid (\varphi \land \varphi) \mid \mathrm{K}_a \varphi
\]
where $p \in \At$ is an atomic proposition and
$a \in \Pi$ is an agent.

\subsection{Knowledge gain and logical obstruction}

\begin{theorem}[Knowledge gain~\cite{DBLP:journals/iandc/GoubaultLR21}]
  Suppose $f \colon M \to M'$ be a morphism between simplicial models.
  Then, for any facet $X \in \facet{M}$
  and positive formula $\varphi \in \pf$,
  $M', f(X) \models \varphi$ implies $M, X \models \varphi$.

\end{theorem}

\begin{corollary}
  If there exists a morphism $f \colon M \to M'$,
  for any positive formula $\varphi \in \pf$,
  $M' \models \varphi$ implies $M \models \varphi$.

\end{corollary}

\begin{corollary}
  Let $\mathcal{I}$ be an input simplicial model,
  $T$ be a task and $A$ be an action model.
  If there exists a positive formula $\varphi \in \pf$
  such that $\mathcal{I}[T] \models \varphi$ and
  $\mathcal{I}[A] \not \models \varphi$,
  the task $T$ is not solvable by $A$.

\end{corollary}

A positive formula $\varphi$ such as the one in the above corollary is called
a \emph{logical obstruction to the solvability of $T$ by $A$}.
We simply call $\varphi$ a \emph{logical obstruction}
when $T$ and $A$ are clear from the context.

\subsection{Knowledge gain via simulation}

\begin{definition}[Simulation]  \label{def_K-simulation}
  Suppose $M$ and $M'$ are simplicial models.
  A \emph{simulation of $M$ by $M'$} is a binary relation
  $R \subseteq \facet{M} \times \facet{M'}$
  over the facets of simplicial models
  that satisfies the following properties.
  \begin{description}
    \item[(Atom)]
      For all $X \in \facet{M}$ and $X' \in \facet{M'}$,
      $X \mathrel{R} X'$ implies $l(X) = l'(X')$.
    \item[(Forth)]
      For all $a \in \Pi$, $X, Y \in \facet{M}$ and $X' \in \facet{M'}$,
      if $X \mathrel{R} X'$ and $X \sim_a Y$,
      there exists $Y' \in \facet{M'}$ such that
      $Y \mathrel{R} Y'$ and $X' \sim_a Y'$.
  \end{description}

  A simulation $S$ is called \emph{total}
  if for all $X \in \facet{M}$, there exists $X' \in \facet{M'}$
  such that $X \mathrel{S} X'$.
  For a facet $X \in \facet{M}$,
  we say that $S$ is \emph{not total at $X$} if
  there is no $X' \in \facet{M'}$ such that $X \mathrel{S} X'$.

\end{definition}

\begin{proposition}
  Suppose $f \colon M \to M'$ is a morphism between simplicial models,
  and define a binary relation $graph(f) \subseteq \rel$ by
  $graph(f) = \{(X, f(X)) \mid X \in \facet{M}\}$.
  Then, the $graph(f)$ is a total simulation.

\end{proposition}
\begin{proof}
  \atom: By the definition of morphism, $f$ preserves labeling,
  i.e. $l(X) = l'(f(X))$.

  \forth: Let $X, Y \in \facet{M}$ be facets such that $X \sim_a Y$.
  Then there exists $a$-colored vertex $v \in X \cap Y$.
  Since $f$ is a color-preserving map,
  we have $f(v) \in f(X) \cap f(Y)$ and $\chi'(f(v)) = \chi(v) = a$.
  This implies $f(X) \sim_a f(Y)$.

  It is clear that $graph(f)$ is total.
\end{proof}

\begin{theorem} \label{th_knowledge_gain_via_K-simulation}
  Suppose $S \subseteq \facet{M} \times \facet{M'}$ be a simulation.
  Then, for any facets $(X, X') \in S$
  and positive formula $\varphi \in \pf$,
  $M', X' \models \varphi$ implies $M, X \models \varphi$.
\end{theorem}
\begin{proof}
  We proceed by induction on $\varphi$.

  For the base case, suppose $\varphi = p$ or $\lnot p$,
  where $p$ is an atomic proposition.
  Since $S$ satisfies the condition \atom, we have
  $M, X \models p \iff p \in l(X) \iff p \in l'(X') \iff M', X' \models p$.
  This implies $M', X' \models \varphi \implies M, X \models \varphi$.

  The cases of conjunction and disjunction are
  easily shown by induction hypothesis.

  For the case of modal operator, suppose $\varphi = \mathrm{K}_a \psi$.
  Assuming $M', X' \models \mathrm{K}_a \psi$,
  we show that $M, Y \models \psi$ holds for any $Y \in \facet{M}$
  such that $X \sim_a Y$.
  By the condition \forth, there exists $Y' \in \facet{M'}$
  such that $Y \mathrel{S} Y'$ and
  $X' \sim_a Y'$.
  This implies $M', Y' \models \psi$
  and also $M, Y \models \psi$ by induction hypothesis.
  Therefore, $M, X \models \mathrm{K}_a \psi$.
\end{proof}

\begin{corollary} \label{cor_knowledge_gain_via_K-simulation}
  If there exists a total simulation of $M$ by $M'$,
  for any positive formula $\varphi \in \pf$,
  $M' \models \varphi$ implies $M \models \varphi$.

\end{corollary}

\subsection{Non-existence of logical obstructions
to the solvability of $\mathcal{SA}_k$ by $\mathcal{IS}^r$}

The (iterated) immediate snapshot is
a protocol that characterizes the wait-free asynchronous
distributed computation in the read-write shared memory systems~\cite{DBLP:conf/podc/BorowskyG92,DBLP:books/mk/Herlihy2013}.
The specification of this protocol is as follows.
Let $\Pi$ be a set of processes and suppose that
processes in $\Pi$ communicate with each other
by iterating the immediate snapshot protocol $r$ times.
When a process $a \in \Pi$ finishes immediate snapshot of $i$-th iteration,
it obtains a \emph{view} $V_a \subseteq \{(p, v_p) \mid p \in \Pi\}$
where $v_p$ is the view of $p$'s $(i-1)$-th iteration,
assuming the view of $p$'s iteration is the initial input value to $p$.
The views $(V_p)_{p \in \Pi}$ of each round satisfy
the following properties.
\begin{description}
  \item[(Self-inclusion)]
    For all $a \in \Pi$, $(a, v_a) \in V_a$.
  \item[(Containment)]
    For all $a, b \in \Pi$, $V_a \subseteq V_b$ or $V_b \subseteq V_a$.
  \item[(Immediacy)]
    For all $a, b \in \Pi$, $(b, v_b) \in V_a$ implies $V_b \subseteq V_a$.
\end{description}
The set of views are equally specified by means of
\emph{ordered partition}.
The $r$-iterated immediate snapshot is thus modeled by an action model
whose underlying simplicial complex is determined by
initial input values and a sequence of $r$ ordered partitions.

\begin{definition}
  An \emph{ordered partition} of $\Pi$ is a finite sequence
  $\gamma = (C_1, \dots , C_l)$ of subsets of $\Pi$ such that
  $\emptyset \neq C_1 \subsetneq C_2 \subsetneq \dots \subsetneq C_l = \Pi$.
  We write $\mathrm{OP}_{\Pi}$ for the set of all ordered partitions of $\Pi$.

\end{definition}

\begin{definition}
  Let $\I$ be an input simplicial model.
  We define a set $\Views^r$ of possible views
  obtained by communicating $r$-times,
  by induction on $r$ as
  \begin{itemize}
    \item $\Views^0 = V^{in}$ and
    \item $\Views^{r+1} =
      \{\{(p, v_p) \mid p \in P\}
      \mid \text{$P \subseteq \Pi$, $v_p \in \Views^r$}\}$.
  \end{itemize}
  Let $a \in \Pi$ be an agent.
  We also define a map
  $view^r_a \colon \facet{\I} \times \mathrm{OP}^r_{\Pi} \to \Views^r$,
  by induction on $r$ as
  \begin{itemize}
    \item $view^0_a(\{(b, i_b) \mid b \in \Pi\}) = i_a$ and
    \item $view^{r+1}_a(I, \gamma_1, \dots , \gamma_r, (C_1, \dots , C_l))
      = \{(p, view^r_p(I, \gamma_1, \dots , \gamma_r))
        \mid p \in C_{\min\{j \mid a \in C_j\}}\}$.
  \end{itemize}

\end{definition}

\begin{definition}
  Let $r \geq 1$ be a natural number.
  We define an action model $\mathcal{IS}^r = (V, S, \chi, \pre)$,
  for \emph{$r$-iterated immediate snapshot protocol} as follows:
  \begin{itemize}
    \item $V \subseteq \Pi \times V^{in} \times \Views^r$,
    \item $X \in \facet{\mathcal{IS}^r}$ if there exists
      a facet $I \in \facet{\I}$ and
      ordered partitions $\gamma_1, \dots , \gamma_r$ such that
      $X = \{(a, i_a, view^r_a(I, \gamma_1, \dots, \gamma_r))
        \mid (a, i_a) \in I\}$,
    \item $\chi(a, i, v) = a$ and
    \item $\pre(\{(a, i_a, v_a) \mid a \in \Pi\})
      = \bigwedge_{a \in \Pi} \ip_a^{i_a}$.
  \end{itemize}

\end{definition}

\begin{figure}
  \centering
  \input{figures/30_one_round_is}
\end{figure}

The action model $\mathcal{IS}^1$ is illustrated
in Figure~\ref{fig_iis1}, for the case
$\Pi = \{\circ, \graybullet, \bullet\}$
and all process have $0$ as input.
For instance, let $X$ be a facet labeled with
$\gamma = (\{\bullet\},\{\circ, \bullet\},
  \{\circ, \graybullet, \bullet\})$.
Then, the views of each processes are as follows:
$view^1_{\circ} = \{(\circ, 0), (\bullet, 0)\}$ and
$view^1_{\graybullet} = \{(\circ, 0), (\graybullet, 0), (\bullet, 0)\}$,
$view^1_{\bullet} = \{(\bullet, 0)\}$.
This means
$X = \{(\circ, 0, \{(\circ, 0), (\bullet, 0)\}),
  \{(\graybullet, 0, \{(\circ, 0), (\graybullet, 0), (\bullet, 0)\}),
  \{(\bullet, 0, \{(\bullet, 0)\})\}$.

Let $\mathcal{I}$ be an input simplicial model.
The underlying simplicial complex of the product update model $\IIS{r}$
is equivalent to the one of the action model $\mathcal{IS}^r$
up to a simplicial map
$((a, i_a), (a, i_a, v_a)) \mapsto (a, i_a, v_a)$.

In what follows,
we write $(a, i_a, v_a)$, in abuse of notation, to mean a vertex
$((a, i_a), (a, i_a, v_a))$ of $\IIS{r}$.
We also write $I[\gamma_1, \dots, \gamma_r]$ for a facet
$\{(a, i_a, view^r_a(I, \gamma_1, \dots, \gamma_r))
\mid (a, i_a) \in I\}$.

\begin{definition}
  We define a map
  $\car \colon \bigcup_{r \geq 1}\Views^r \to \Views^1$ as
  \[
    \car(w) =
    \begin{cases}
      w, & \text{(if $w \in \Views^1$)} \\
      \bigcup_{p \in P} \car(v_p), & \text{(otherwise)}
    \end{cases}
  \]
  where $P \subseteq \Pi$ is a set of agents and
  $(v_p)_{p \in P}$ is a family of views
  such that $w = \{(p, v_p) \mid p \in P\}$.
\end{definition}

\begin{lemma} \label{lem_flatview}
  Let $I \in \facet{\mathcal{I}}$
  be a facet of an input simplicial model and
  $\gamma_1, \gamma_2, \ldots \in \mathrm{OP}_{\Pi}$
  be ordered partitions of $\Pi$,
  where $\gamma_1 = (C_1, \dots, C_l)$.
  For any $r \geq 1$ and $a \in \Pi$,
  the following holds:
  \[
    \{(p, i_p) \in I \mid p \in C_1\} \subseteq
      \car(view^r_a(I, \gamma_1, \dots, \gamma_r)) \subseteq I.
  \]

\end{lemma}
\begin{proof}
  We proceed by induction on $r$.
  Given an agent $a \in \Pi$,
  let us write $j_a$ to denote $\min\{j \mid a \in C_j\}$.

  For the base case $r = 1$, we have
  $\car(view^1_a(I, \gamma_1))
    = \{(q, i_q) \in I \mid q \in C_{j_a}\} \subseteq I$
  and $\{(p, i_p) \mid p \in C_1\}
    \subseteq \{(q, i_q) \mid q \in C_{j_a}\}$
  since $C_1 \subseteq C_j$ for any $1 \leq j \leq l$.

  For the case $r + 1$, we have
  \begin{align*}
    &\car(view^{r+1}_a(I, \gamma_1, \dots, \gamma_{r+1})) \\
      &= \car(\{(q, i_q,
        view^r_q(I, \gamma_1, \dots, \gamma_r)) \mid q \in C_{j_a}\}) \\
      &= \bigcup_{q \in C_{j_a}}
        \car(view^r_q(I, \gamma_1, \dots, \gamma_r)).
  \end{align*}
  By induction hypothesis, we obtain
  \begin{align*}
    \{(p, i_p) \mid p \in C_1\}
    &= \bigcup_{q \in C_{j_a}} \{(p, i_p) \mid p \in C_1\} \\
    &\subseteq \bigcup_{q \in C_{j_a}}
      \car(view^r_q(I, \gamma_1, \dots, \gamma_r)) \\
    &\subseteq \bigcup_{q \in C_{j_a}} I = I.
  \end{align*}
  Thus, $\{(p, i_p) \mid p \in C_1\} \subseteq
    \car(view^{r+1}_a(I, \gamma_1, \dots, \gamma_r))
    \subseteq I$ holds.
\end{proof}

The topological model of (iterated) immediate snapshot is
studied~\cite{DBLP:journals/jacm/HerlihyS99}
and it has been shown that the $k$-set agreement task is
unsolvable by $r$-round immediate snapshot
for any $r \geq 1$ and $1 \leq k < |\Pi|$.
It is desirable to have a logical obstruction,
but there is no such a formula.

The following theorem shows that $\pf$ admits no logical obstruction
to $k$-set agreement for any $k \geq 2$.

\begin{theorem} \label{th_non-existence_of_lo_SAk_ISr}
  There is no logical obstruction to the solvability of
  $\mathcal{SA}_k$ by $\mathcal{IS}^r$ in $\pf$,
  whenever $k \geq 2$.

\end{theorem}
\begin{proof}
  Let $S \subseteq \facet{\IIS{r}} \times \facet{\ISA{k}}$
  be a binary relation defined as
  \begin{alignat*}{2}
    \{(a, i_a, v_a) \mid a \in \Pi\} \mathrel{S}
    \{(a, i'_a, d_a) \mid a \in \Pi\}
    \iff
    &\text{(1)}& &\text{ for all $a \in \Pi$, $i_a = i'_a$ and} \\
    &\text{(2)}& &\text{ for all $a \in \Pi$, there exists $a' \in \Pi$} \\
    &&&\text{such that $(a', d_a) \in \car(v_a)$.}
  \end{alignat*}
  By Corollary~\ref{cor_knowledge_gain_via_K-simulation},
  it suffices to show that $S$ is a total simulation.

  Let us first show $S$ is a simulation.

  \atom: The labeling of facets is given by
  $l_{\IIS{r}}(\{(a, i_a, v_a) \mid a \in \Pi\})
    = \{\ip^{i_a}_a \mid a \in \Pi\}$
  and $l_{\ISA{k}}(\{(a, i'_a, d_a) \mid a \in \Pi\})
    = \{\ip^{i'_a}_a \mid a \in \Pi\}$.
  Thus, {\atom} follows from the definition of $S$.

  \forth: Suppose $a_0 \in \Pi$ is an agent,
  $X, Y$ are facets of $\IIS{r}$ and
  $X'$ is a facet of $\ISA{k}$ such that
  $X \mathrel{S} X'$ and $X \sim_{a_0} Y$.
  Let $I \in \facet{\I}$ be a facet of the input simplicial complex and
  $\gamma_1, \dots, \gamma_r$ be ordered partitions such that
  $Y = I[\gamma_1, \dots, \gamma_r]$,
  where $\gamma_1 = (C_1, \dots, C_j)$.
  Fix an agent $a_1 \in C_1$.

  For all $a \in \Pi$, let us write $input_X(a)$ for an input value and
  $view_X(a)$ for a view in $\Views^r$
  such that $(a, input_X(a), view_X(a))$ is an $a$-colored vertex in $X$.
  Similarly, let us write $(a, input_Y(a), view_Y(a))$
  for an $a$-colored vertex in $Y$,
  and $(a, input_{X'}(a), decide_{X'}(a))$
  for an $a$-colored vertex in $X'$.

  We define a facet $Y' \in \ISA{k}$ by
  \[
    Y' = \{(a_0, input_Y(a_0), decide_{X'}(a_0))\}
      \cup \{(a, input_Y(a), input_Y(a_1)) \mid a \in \Pi \setminus \{a_0\}\}.
  \]
  Let us write $(a, input_{Y'}(a), decide_{Y'}(a))$ for
  an $a$-colored vertex in $Y'$ for all $a \in \Pi$ as above.
  Then, we have
  $input_{Y'}(a) = input_Y(a)$ for all $a \in \Pi$,
  $decide_{Y'}(a_0) = decide_{X'}(a_0)$ and
  $decide_{Y'}(a) = input_Y(a_1)$ for all $a \in \Pi \setminus \{a_0\}$.

  We prove that
  $Y' \in \facet{\ISA{k}}$, $X' \sim_{a_0} Y'$ and $Y \mathrel{S} Y'$.

  It is clear that $|decide_{Y'}(\Pi)| \leq 2 \leq k$ by definition.
  Suppose $Y$ and $Y'$ satisfies both (1) and (2).
  Then for all $a \in \Pi$,
  $input_Y(a) = input_{Y'}(a)$ and
  there exists $a' \in \Pi$
  such that $(a', decide_{Y'}(a)) \in \car(view_Y(a))$.
  Because
  $\car(view_Y(a)) = \car(view_a^r(I, \gamma_1, \dots, \gamma_r)) \subseteq I$
  by Lemma~\ref{lem_flatview},
  $(a', decide_{Y'}(a)) \in I$,
  where $I = \{(a, input_Y(a)) \mid a \in \Pi\}
    = \{(a, input_{Y'}(a)) \mid a \in \Pi\}$.
  Hence, we have $I \models \bigwedge_{a \in \Pi}
    \bigvee_{b \in \Pi} \ip_b^{decide_{Y'}(a)}$.
  These imply $Y' \in \facet{\ISA{k}}$
  if $Y$ and $Y'$ satisfies both (1) and (2).
  Therefore,
  $Y' \in \facet{\ISA{2}}$ holds because $Y$ and $Y'$ satisfies
  the conditions (1) and (2), as shown below.

  Since $X \mathrel{S} X'$, $X \sim_{a_0} Y$ and
  $input_{Y'}(a) = input_Y(a)$ for all $a \in \Pi$,
  $input_{X'}(a_0) = input_X(a_0) = input_Y(a_0) = input_{Y'}(a_0)$.
  Recall that $decide_{Y'}(a_0) = decide_{X'}(a_0)$ by definition.
  Thus, we have
  $(a_0, input_{X'}(a_0), decide_{X'}(a_0))
    = (a_0, input_{Y'}(a_0), decide_{Y'}(a_0)) \in X' \cap Y'$.
  Hence, $X' \sim_{a_0} Y'$ holds.

  Since $input_{Y'}(a) = input_Y(a)$ for all $a \in \Pi$ by definition,
  $Y$ and $Y'$ satisfies (1).
  We prove $Y$ and $Y'$ satisfies (2).
  For the case $a \neq a_0$, we have
  $(a_1, decide_{Y'}(a)) = (a_1, input_Y(a_1)) \in \car(view_Y(a))$
  by Lemma~\ref{lem_flatview}.
  For the other case $a = a_0$,
  since $decide_{Y'}(a_0) = decide_{X'}(a_0)$,
  $X \mathrel{S} X'$ and $X \sim_{a_0} Y$,
  there exists $a'_0 \in \Pi$ such that
  $(a'_0, decide_{Y'}(a_0)) = (a'_0, decide_{X'}(a_0)) \in
    \car(view_X(a_0)) = \car(view_Y(a_0))$.
  Consequently, we have for all $a \in \Pi$, there exists $a' \in \Pi$
  such that $(a', decide_{Y'}(a)) \in \car(view_Y(a))$, i.e.\
  $Y$ and $Y'$ satisfies (2), and $Y \mathrel{S} Y'$.
  Therefore, {\forth} is satisfied.

  To show the totality,
  suppose $X$ is a facet of $\facet{\ISA{k}}$.
  Let $I = \{(a, input(a)) \mid a \in \Pi\}$
  be a facet of the input model $\I$ and
  $\gamma_1, \dots, \gamma_r$ be ordered partitions
  such that $X = I[\gamma_1, \dots, \gamma_r]$,
  where $\gamma_1 = (C_1, \dots, C_j)$.
  Fix an agent $a_1 \in C_1$.
  We define a facet $X' \in \ISA{k}$
  as $X' = \{(a, input(a), input(a_1)) \mid a \in \Pi \}$.
  By Lemma~\ref{lem_flatview},
  $(a_1, input(a_1)) \in \{(p, input(p)) \mid p \in C_1\}
    \subseteq \car(view_a^r(I, \gamma_1, \dots, \gamma_r))$
  for all $a \in \Pi$.
  Thus, we have $X \mathrel{S} X'$.
  This implies that $S$ is a total simulation.
\end{proof}

We have shown the non-existence of logical obstruction in $\pf$
for cases where $k \geq 2$.
For the case $k = 1$,~\cite{DBLP:journals/iandc/GoubaultLR21}
have shown that there exists a logical obstruction
to the solvability in $\mathcal{L}_{CK}^+$,
an epistemic logic language of positive formulas
that extends $\mathcal{L}_K^+$
with common knowledge operator.
It is not difficult to see that $k$-set agreement task ($k \geq 2$)
admits no logical obstruction in $\mathcal{L}_{CK}^+$
because the simulation technique presented in this section is
compatible with common knowledge operators.

%% file: figures/30_one_round_is.tex

\begin{tikzpicture}
  \node[circle,scale=0.8,draw,fill=white] (c0v0) at ({0/1+sqrt(3)*0/1}, {0/1+sqrt(3)*0/1}) {};
  \node[circle,scale=0.6,draw,fill=lightgray] (c1v01) at ({4/3+sqrt(3)*0/1}, {0/1+sqrt(3)*0/1}) {};
  \node[circle,scale=0.6,draw,fill=black] (c2v012) at ({2/1+sqrt(3)*0/1}, {0/1+sqrt(3)*2/5}) {};
  \node[circle,scale=0.8,draw,fill=white] (c0v0) at ({0/1+sqrt(3)*0/1}, {0/1+sqrt(3)*0/1}) {};
  \node[circle,scale=0.6,draw,fill=lightgray] (c1v012) at ({8/5+sqrt(3)*0/1}, {0/1+sqrt(3)*4/5}) {};
  \node[circle,scale=0.6,draw,fill=black] (c2v012) at ({2/1+sqrt(3)*0/1}, {0/1+sqrt(3)*2/5}) {};
  \node[circle,scale=0.8,draw,fill=white] (c0v0) at ({0/1+sqrt(3)*0/1}, {0/1+sqrt(3)*0/1}) {};
  \node[circle,scale=0.6,draw,fill=lightgray] (c1v012) at ({8/5+sqrt(3)*0/1}, {0/1+sqrt(3)*4/5}) {};
  \node[circle,scale=0.6,draw,fill=black] (c2v02) at ({2/3+sqrt(3)*0/1}, {0/1+sqrt(3)*2/3}) {};
  \node[circle,scale=0.6,draw,fill=white] (c0v01) at ({8/3+sqrt(3)*0/1}, {0/1+sqrt(3)*0/1}) {};
  \node[circle,scale=0.6,draw,fill=lightgray] (c1v01) at ({4/3+sqrt(3)*0/1}, {0/1+sqrt(3)*0/1}) {};
  \node[circle,scale=0.6,draw,fill=black] (c2v012) at ({2/1+sqrt(3)*0/1}, {0/1+sqrt(3)*2/5}) {};
  \node[circle,scale=0.6,draw,fill=white] (c0v01) at ({8/3+sqrt(3)*0/1}, {0/1+sqrt(3)*0/1}) {};
  \node[circle,scale=0.8,draw,fill=lightgray] (c1v1) at ({4/1+sqrt(3)*0/1}, {0/1+sqrt(3)*0/1}) {};
  \node[circle,scale=0.6,draw,fill=black] (c2v012) at ({2/1+sqrt(3)*0/1}, {0/1+sqrt(3)*2/5}) {};
  \node[circle,scale=0.6,draw,fill=white] (c0v012) at ({12/5+sqrt(3)*0/1}, {0/1+sqrt(3)*4/5}) {};
  \node[circle,scale=0.6,draw,fill=lightgray] (c1v012) at ({8/5+sqrt(3)*0/1}, {0/1+sqrt(3)*4/5}) {};
  \node[circle,scale=0.6,draw,fill=black] (c2v012) at ({2/1+sqrt(3)*0/1}, {0/1+sqrt(3)*2/5}) {};
  \node[circle,scale=0.6,draw,fill=white] (c0v012) at ({12/5+sqrt(3)*0/1}, {0/1+sqrt(3)*4/5}) {};
  \node[circle,scale=0.6,draw,fill=lightgray] (c1v012) at ({8/5+sqrt(3)*0/1}, {0/1+sqrt(3)*4/5}) {};
  \node[circle,scale=0.8,draw,fill=black] (c2v2) at ({2/1+sqrt(3)*0/1}, {0/1+sqrt(3)*2/1}) {};
  \node[circle,scale=0.6,draw,fill=white] (c0v012) at ({12/5+sqrt(3)*0/1}, {0/1+sqrt(3)*4/5}) {};
  \node[circle,scale=0.8,draw,fill=lightgray] (c1v1) at ({4/1+sqrt(3)*0/1}, {0/1+sqrt(3)*0/1}) {};
  \node[circle,scale=0.6,draw,fill=black] (c2v012) at ({2/1+sqrt(3)*0/1}, {0/1+sqrt(3)*2/5}) {};
  \node[circle,scale=0.6,draw,fill=white] (c0v012) at ({12/5+sqrt(3)*0/1}, {0/1+sqrt(3)*4/5}) {};
  \node[circle,scale=0.8,draw,fill=lightgray] (c1v1) at ({4/1+sqrt(3)*0/1}, {0/1+sqrt(3)*0/1}) {};
  \node[circle,scale=0.6,draw,fill=black] (c2v12) at ({10/3+sqrt(3)*0/1}, {0/1+sqrt(3)*2/3}) {};
  \node[circle,scale=0.6,draw,fill=white] (c0v012) at ({12/5+sqrt(3)*0/1}, {0/1+sqrt(3)*4/5}) {};
  \node[circle,scale=0.6,draw,fill=lightgray] (c1v12) at ({8/3+sqrt(3)*0/1}, {0/1+sqrt(3)*4/3}) {};
  \node[circle,scale=0.6,draw,fill=black] (c2v12) at ({10/3+sqrt(3)*0/1}, {0/1+sqrt(3)*2/3}) {};
  \node[circle,scale=0.6,draw,fill=white] (c0v012) at ({12/5+sqrt(3)*0/1}, {0/1+sqrt(3)*4/5}) {};
  \node[circle,scale=0.6,draw,fill=lightgray] (c1v12) at ({8/3+sqrt(3)*0/1}, {0/1+sqrt(3)*4/3}) {};
  \node[circle,scale=0.8,draw,fill=black] (c2v2) at ({2/1+sqrt(3)*0/1}, {0/1+sqrt(3)*2/1}) {};
  \node[circle,scale=0.6,draw,fill=white] (c0v02) at ({4/3+sqrt(3)*0/1}, {0/1+sqrt(3)*4/3}) {};
  \node[circle,scale=0.6,draw,fill=lightgray] (c1v012) at ({8/5+sqrt(3)*0/1}, {0/1+sqrt(3)*4/5}) {};
  \node[circle,scale=0.6,draw,fill=black] (c2v02) at ({2/3+sqrt(3)*0/1}, {0/1+sqrt(3)*2/3}) {};
  \node[circle,scale=0.6,draw,fill=white] (c0v02) at ({4/3+sqrt(3)*0/1}, {0/1+sqrt(3)*4/3}) {};
  \node[circle,scale=0.6,draw,fill=lightgray] (c1v012) at ({8/5+sqrt(3)*0/1}, {0/1+sqrt(3)*4/5}) {};
  \node[circle,scale=0.8,draw,fill=black] (c2v2) at ({2/1+sqrt(3)*0/1}, {0/1+sqrt(3)*2/1}) {};
  \draw (c0v0) -- (c1v01);
  \draw (c0v0) -- (c1v012);
  \draw (c0v0) -- (c2v012);
  \draw (c0v0) -- (c2v02);
  \draw (c0v01) -- (c1v01);
  \draw (c0v01) -- (c1v1);
  \draw (c0v01) -- (c2v012);
  \draw (c0v012) -- (c1v012);
  \draw (c0v012) -- (c1v1);
  \draw (c0v012) -- (c1v12);
  \draw (c0v012) -- (c2v012);
  \draw (c0v012) -- (c2v12);
  \draw (c0v012) -- (c2v2);
  \draw (c0v02) -- (c1v012);
  \draw (c0v02) -- (c2v02);
  \draw (c0v02) -- (c2v2);
  \draw (c1v01) -- (c2v012);
  \draw (c1v012) -- (c2v012);
  \draw (c1v012) -- (c2v02);
  \draw (c1v012) -- (c2v2);
  \draw (c1v1) -- (c2v012);
  \draw (c1v1) -- (c2v12);
  \draw (c1v12) -- (c2v12);
  \draw (c1v12) -- (c2v2);
  \coordinate (f1) at (1.2,1.65);
  \node (f1') at (-1.5, 1.75) {$\gamma = (\{\circ, \bullet\}, \{\circ, \textcolor{lightgray}{\bullet}, \bullet\})$};
  \draw[->] (node cs:name=f1',anchor=east) to (f1);
  \coordinate (f2) at (2,{2*sqrt(3)/3});
  \node (f2') at (-1.5, 0.5) {$\gamma = (\{\circ, \textcolor{lightgray}{\bullet}, \bullet\})$};
  \draw[->] (node cs:name=f2',anchor=east) to (f2);
  \coordinate (f3) at (1.6,2.3);
  \node (f3') at (-1.5,3) {$\gamma = (\{\bullet\},\{\circ, \bullet\}, \{\circ, \textcolor{lightgray}{\bullet}, \bullet\})$};
  \draw[->] (node cs:name=f3',anchor=east) to (f3);
  \coordinate (f4) at (2,2);
  \node (f4') at (5.5,1.65) {$\gamma = (\{\bullet\}, \{\circ, \textcolor{lightgray}{\bullet}, \bullet\})$~~~~~~~};
  \draw[->] (node cs:name=f4',anchor=west) to (f4);
\end{tikzpicture}
\caption{Simplices of $\mathcal{IS}^1$, which is isomorphic to $\IIS{1}$,
  for the case $\Pi = \{\circ, \textcolor{lightgray}{\bullet}, \bullet\}$
  and all process have $0$ as input.}
\label{fig_iis1}

%% file: 40_D-simulation.tex

\section{Non-existence of logical obstructions
with distributed knowledge operator}
\label{sec_d-simulation}

  In this section, we consider a language of epistemic logic $\mathcal{L}_D$,
  which is an extension of $\mathcal{L}_K$ with distributed operator
  $\mathrm{D}_A \varphi$,
  and its sublanguage $\mathcal{L}_D^+$ of positive formulas:
  \begin{align*}
    \mathcal{L}_D: \quad
    \varphi &::= p \mid \lnot \varphi
    \mid (\varphi \land \varphi) \mid \mathrm{D}_A \varphi  \\
    \mathcal{L}_D^+:  \quad
    \varphi &::= p \mid \lnot p \mid (\varphi \lor \varphi)
    \mid (\varphi \land \varphi) \mid \mathrm{D}_A \varphi
  \end{align*}
  where $p \in \At$ is an atomic proposition and
  $A \subseteq \Pi$ is a set of agents.

  Given a formula $\varphi \in \mathcal{L}_D$ and
  a facet $X \in \facet{M}$ of a simplicial model $M$,
  we define the \emph{truth} of $\varphi$ at $X$,
  written $M, X \models \varphi$, as below by induction on $\varphi$.
  \begin{align*}
    &M, X \models p
    && \text{if $p \in l(X)$}. \\
    &M, X \models \lnot\varphi
    && \text{if $M, X \not\models \varphi$}. \\
    &M, X \models \varphi_1 \land \varphi_2
    && \text{if $M, X \models \varphi_1$ and $M, X \models \varphi_2$}. \\
    &M, X \models \mathrm{D}_A\varphi
    && \text{if for all Y $\in \facet{M}$,
    $A \subseteq \chi(X \cap Y) \implies M, Y \models \varphi$}.
  \end{align*}

  Notice that $\mathrm{K}_a \varphi$ is
  a special case of $\mathrm{D}_A \varphi$ where
  $A$ is a singleton set $\{a\}$.

\subsection{Knowledge gain and
a logical obstruction to the solvability of
$\mathcal{SA}_k$ by $\mathcal{IS}^1$
with distributed knowledge}

The knowledge gain theorem holds,
even the additional modality of distributed knowledge.

\begin{theorem}[Knowledge gain with distributed knowledge operator]
  Suppose $f \colon M \to M'$ be a morphism between simplicial models.
  Then, for any facet $X \in \facet{M}$
  and positive formula $\varphi \in \pfd$,
  $M', f(X) \models \varphi$ implies $M, X \models \varphi$.

\end{theorem}

\begin{corollary}
  If there exists a morphism $f \colon M \to M'$,
  for any positive formula $\varphi \in \pfd$,
  $M' \models \varphi$ implies $M \models \varphi$.

\end{corollary}

\begin{corollary}
  Let $\mathcal{I}$ be an input simplicial model,
  $T$ be a task and $A$ be an action model.
  If there exists a positive formula $\varphi \in \pfd$
  such that $\mathcal{I}[T] \models \varphi$ and
  $\mathcal{I}[A] \not \models \varphi$,
  the task $T$ is not solvable by $A$.

\end{corollary}

Thus, we also call a positive formula $\varphi$
such as the one in the above corollary,
a \emph{logical obstruction to the solvability of $T$ by $A$},
or a \emph{logical obstruction}, when $T$ and $A$ are clear from the context.

The following theorem on the unsolvability of the set agreement task
by a single round immediate snapshot is due to Nishida~\cite{Nishida:Msc20},
where he provided a concrete logical obstruction
in the language $\mathcal{L}_D^{+}$.
It has also been shown that the same formula applies to show
the unsolvability of the set agreement task
by a single round atomic snapshot protocol~\cite{DBLP:journals/corr/abs-2011-13630}.

\begin{theorem}
  For any $k < |\Pi|$,
  there exists a logical obstruction to the solvability of
  $\mathcal{SA}_k$ by $\mathcal{IS}^1$ in $\pfd$.

\end{theorem}

\subsection{Knowledge gain via D-simulation}

In the following, we refine the definition of simulation
so that Theorem~\ref{th_knowledge_gain_via_K-simulation}
is valid if positive formulas have
distributed knowledge operator.

\begin{definition}[D-simulation]
  A binary relation over facets of simplicial models
  $R \subseteq \facet{M} \times \facet{M'}$ is
  called a \emph{D-simulation of $M$ by $M'$},
  if the following conditions hold.
  \begin{description}
    \item[(Atom)]
    For all $X \in \facet{M}$ and $X' \in \facet{M'}$,
    $X \mathrel{R} X'$ implies $l(X) = l'(X')$.
    \item[(D-Forth)]
    For all $X, Y \in \facet{M}$ and $X' \in \facet{M'}$,
    if $X \mathrel{R} X'$ holds,
    there exists $Y' \in \facet{M'}$ such that
    $Y \mathrel{R} Y'$ and $\chi(X \cap Y) \subseteq \chi'(X' \cap Y')$.
  \end{description}

  A D-simulation $S$ is called \emph{total}
  if for all $X \in \facet{M}$, there exists $X' \in \facet{M'}$
  such that $X \mathrel{S} X'$.
  For a facet $X \in \facet{M}$,
  we say that $S$ is \emph{not total at $X$} if
  there is no $X' \in \facet{M'}$ such that $X \mathrel{S} X'$.
  We occasionally write \emph{K-simulation}
  to refer to the simulation defined in Definition~\ref{def_K-simulation},
  and as well {\kforth} to refer to {\forth}.

\end{definition}

\begin{proposition}
  Suppose $f \colon M \to M'$ is a morphism between simplicial models.
  Then, a binary relation
  $graph(f) = \{(X, f(X)) \mid X \in \facet{M}\} \subseteq \rel$
  is a total D-simulation.

\end{proposition}
\begin{proof}
  \atom: By the definition of morphism, $f$ preserves labeling,
  i.e. $l(X) = l'(f(X))$.

  \dforth: Let $X, Y \in \facet{M}$ be facets.
  It suffices to show that $\chi(X \cap Y) \subseteq \chi'(f(X) \cap f(Y))$.
  Suppose $a \in \chi(X \cap Y)$.
  Then there exists an $a$-colored vertex $v \in X \cap Y$.
  Since $f$ is a color-preserving map,
  we have $f(v) \in f(X) \cap f(Y)$ and $\chi'(f(v)) = \chi(v) = a$.
  This implies $a \in \chi'(f(X) \cap f(Y))$.

  It is clear that $graph(f)$ is total.
\end{proof}

\begin{theorem}
  Suppose $S \subseteq \facet{M} \times \facet{M'}$ is a D-simulation.
  Then, for any facets $(X, X') \in S$
  and positive formula $\varphi \in \pfd$,
  $M', X' \models \varphi$ implies $M, X \models \varphi$.
\end{theorem}
\begin{proof}
  This theorem is similarly proved
  as in Theorem~\ref{th_knowledge_gain_via_K-simulation}
  by induction on $\varphi$.
  Let us only show the case of distributed knowledge operator.

  Supoose $\varphi = \mathrm{D}_A \psi$.
  Assuming $M', X' \models \mathrm{D}_A \psi$,
  we show that $M, Y \models \psi$ holds for any $Y \in \facet{M}$
  such that $A \subseteq \chi(X \cap Y)$.
  By the condition \dforth, there exists $Y' \in \facet{M'}$
  such that $Y \mathrel{S} Y'$ and
  $A \subseteq \chi(X \cap Y) \subseteq \chi'(X' \cap Y')$.
  This implies $M', Y' \models \psi$
  and also $M, Y \models \psi$ by induction hypothesis.
  Therefore, $M, X \models \mathrm{D}_A \psi$.
\end{proof}

\begin{corollary} \label{cor_knowledge_gain_via_D}
  If there exists a total D-simulation of $M$ by $M'$,
  for any positive formula $\varphi \in \pfd$,
  $M' \models \varphi$ implies $M \models \varphi$.
\end{corollary}

\subsection{Non-existence of logical obstructions
to the solvability of $\mathcal{IS}^2$ by $\mathcal{SA}_2$}

Using a D-simulation, let us argue the unsolvability of the set agreement task
for a particular case $\Pi = \{0, 1, 2\}$.
We leave the general case as an open problem
because D-simulation for higher dimensional model
would be much complicated and hard to construct.

The following binary relation $R^{p,q}$ defines
a set of pairs of two vertices that
we will constrain their decision values.
Although it was sufficient to constrain each vertex individually
when we constructed a K-simulation
in Theorem~\ref{th_non-existence_of_lo_SAk_ISr},
we need to consider constraints for each tuple of vertices
in order to construct D-simulation.

\begin{definition}
  Let $V$ be the set of vertices
  and $S$ be the set of simplices for $\IIS{2}$.
  For a vertex $x = (p, i_p, v_p) \in V$ and an agent $q \in \Pi$,
  in abuse of notation,
  we write $q \in view(x)$ if there exists $i_q \in V^{in}$
  such that $(q, i_q) \in \car(v_p)$
  and $q \notin view(x)$ otherwise.
  Given $p,q \in \Pi$ $(p < q)$,
  We define a binary relation $R_n^{p,q} \subseteq V \times V$
  for $n \in \N$,
  by induction on $n$ as follows.
  \begin{itemize}
    \item $x \mathrel{R_0^{p,q}} y$ if
      $\{x, y\}$ is a 1-simplex of $\IIS{2}$ and
      it holds either
      $p \notin view(x)$ or $q \notin view(y)$.
    \item $x \mathrel{R_{n+1}^{p,q}} y$
      if either $x \mathrel{R_n^{p,q}} y$ or
      there exists $z \in V$ such that
      $x \mathrel{R_n^{p,q}} z$, $z \mathrel{R_n^{p,q}} y$
      and $\{x, y, z\}$ is a 2-simplex.
  \end{itemize}
  We further define a binary relation $R^{p,q} \subseteq V \times V$ as
  $R^{p,q} = \bigcup_{n \in \N} R_n^{p,q} \setminus R_0^{p.q}$.

\end{definition}

\begin{figure}
  \begin{minipage}[t]{0.45\linewidth}
    \centering
    \input{figures/40_two_round_is_1}
    \subcaption{The binary relation $\Rwg_0$}
    \label{fig_iis2_0}
  \end{minipage}
  ~~~
  \begin{minipage}[t]{0.45\linewidth}
    \centering
    \input{figures/41_two_round_is_2}
    \subcaption{The binary relation $\bigcup_{n \in \N} \Rwg_n$}
    \label{fig_iis2_1}
  \end{minipage}

  \vspace{2em}
  \begin{minipage}[t]{0.45\linewidth}
    \centering
    \input{figures/42_two_round_is_3}
    \subcaption{The binary relation $\Rwg$
      and the subcomplex $\IIS{2}_{\circ} \subseteq \IIS{2}$
      (lower left facets enclosed by the thick red lines)}
    \label{fig_iis2_2}
  \end{minipage}
  \hfill

  \caption{The binary relations
  $\Rwg_0$, $\bigcup_{n \in \N} \Rwg_n$ and $\Rwg$
  where each arrow $\textcolor{red}{\to}$ indicates
  a pair of vertices related by
  $\bigcup_{n \in \N} R^{\circ, \textcolor{lightgray}{\bullet}}_n$
  and \textcolor{red}{\textbf{---}} indicates
  both \textcolor{red}{$\to$} and \textcolor{red}{$\leftarrow$}}
  \label{fig_iis2}
\end{figure}

We illustrated the relations
$R_0^{p,q}$, $\bigcup_{n \in \N} R_n^{p,q}$ and $R^{p,q}$
in Figure~\ref{fig_iis2}.
Almost all of the vertices related by $R_n^{p,q}$
(in the Figure~\ref{fig_iis2_1}),
especially all of them on the boundary,
are due to $R_0^{p,q}$ (in the Figure~\ref{fig_iis2_0}).
$R_1^{p,q}$ adds some pairs of vertices such that
each one constitute a facet together with the topmost black vertex.
$R_n^{p,q}$ adds one pair to each of the two symmetrical locations
when $2 \leq n \leq 5$,
and no pair are added when $n \geq 6$.

\begin{definition}  \label{def_simulation_IIS2_ISA2}
  A binary relation
  $S \subseteq \facet{\IIS{2}} \times \facet{\ISA{2}}$ is defined by
  $\{(a, i_a, v_a) \mid a \in \Pi\} \mathrel{S}
    \{(a, i'_a, d_a) \mid a \in \Pi\}$
  if all of the following conditions are satisfied.
  \begin{enumerate}
    \item \label{def_simulation_IIS2_ISA2_cond_1}
      For all $a \in \Pi$, $i_a = i'_a$.
    \item \label{def_simulation_IIS2_ISA2_cond_2}
      For all $a \in \Pi$, there exists $a' \in \Pi$
      such that $(a', d_a) \in \car(v_a)$.
    \item \label{def_simulation_IIS2_ISA2_cond_3}
      For all $a, b, p, q \in \Pi$ such that $a \neq b$ and $p < q$,
      if $|\{i_0, i_1, i_2\}| = 3$ and
      $(a, i_a, v_a) \mathrel{R}^{p,q} (b, i_b, v_b)$,
      either $d_a \neq i_p$ or $d_b \neq i_q$ holds.
  \end{enumerate}

\end{definition}

\begin{lemma} \label{lem_location_of_Rpq}
  Let $a \in \Pi$ be an agent and
  $\IIS{2}_a$ be a subcomplex of $\IIS{2}$ such that
  $\facet{\IIS{2}_a} =
    \{I[\gamma_1, \gamma_2] \in \facet{\IIS{2}}
    \mid \gamma_1 = (C_1, \dots, C_l), a \in C_1\}$.
  Then, for any $p, q \in \Pi$ $(p < q)$,
  $x \mathrel{R^{p,q}} y$ implies
  $\{x, y\} \notin \IIS{2}_p$ and $\{x, y\} \notin \IIS{2}_q$.

\end{lemma}
\begin{proof}
  The relation $R^{p,q}$ and the subcomplex $\IIS{2}_p$
  are shown in the Figure~\ref{fig_iis2_2}
  where $\circ$ and $\bullet$ are read as $p$ and $q$, respectively.
  Hence, as illustrated in Figure~\ref{fig_iis2_2},
  no pair $(x, y)$ related by $R^{p,q}$ is contained in $\IIS{2}_p$
  as a $1$-simplex.
  Similarly for $\IIS{2}_q$.
\end{proof}

\begin{proposition} \label{prop_existence_D-simulation_SA2_IS2}
  The binary relation $S \subseteq \facet{\IIS{2}} \times \facet{\ISA{2}}$
  defined above is a total D-simulation of $\IIS{2}$ by $\ISA{2}$.

\end{proposition}
\begin{proof}

  \atom:
    This case follows from
    Definition~\ref{def_simulation_IIS2_ISA2}~\ref{def_simulation_IIS2_ISA2_cond_1}.

  \dforth:
    Suppose $X, Y$ are facets of $\IIS{2}$ and
    $X'$ is a facet of $\ISA{k}$ such that $X \mathrel{S} X'$.
    Without loss of generality, we may assume $Y \in \IIS{2}_0$ by symmetry.
    Let $A = \chi_{\IIS{2}}(X \cap Y)$.
    We will show that there exists $Y' \in \facet{\ISA{2}}$
    such that $A \subseteq \chi_{\ISA{2}} (X' \cap Y')$ and
    $Y \mathrel{S} Y'$.

    For all $a \in \Pi$, let us write $input_X(a)$ for an input value and
    $view_X(a)$ for a view in $\Views^2$
    such that $(a, input_X(a), view_X(a))$ is an $a$-colored vertex in $X$.
    Similarly, let us write $(a, input_Y(a), view_Y(a))$
    for an $a$-colored vertex in $Y$,
    and $(a, input_{X'}(a), decide_{X'}(a))$
    for an $a$-colored vertex in $X'$.

    \begin{itemize}
      \item Case $|decide_{X'}(A) \cup \{input_Y(0)\}| \leq 2$.
        Define a facet $Y' \in \ISA{2}$ by
        \[
          Y' = \{(a, input_Y(a), decide_{X'}(a))
            \mid a \in A\}
            \cup \{(a, input_Y(a), input_Y(0))
            \mid a \in \Pi \setminus A\}.
        \]
        Let us write $(a, input_{Y'}(a), decide_{Y'}(a))$ for
        an $a$-colored vertex in $Y'$ for all $a \in \Pi$ as above.
        Then, we have
        $input_{Y'}(a) = input_Y(a)$ for all $a \in \Pi$,
        $decide_{Y'}(a) = decide_{X'}(a)$ for all $a \in A$ and
        $decide_{Y'}(a) = input_Y(0)$ for all $a \in \Pi \setminus A$.

        We show that this $Y'$ satisfies
        $Y' \in \facet{\ISA{2}}$,
        $A \subseteq \chi_{\ISA{2}} (X' \cap Y')$ and
        $Y \mathrel{S} Y'$.
        \begin{itemize}
          \item
            We have
            $|decide_{Y'}(\Pi)| \leq
            |decide_{X'}(A) \cup \{input_Y(0)\}| \leq 2$.
            Suppose $Y$ and $Y'$ satisfies the conditions~\ref{def_simulation_IIS2_ISA2_cond_1}
            and~\ref{def_simulation_IIS2_ISA2_cond_2}
            of Definition~\ref{def_simulation_IIS2_ISA2}.
            Then, by a similar argument as in the proof of
            Theorem~\ref{th_non-existence_of_lo_SAk_ISr},
            $\{(a, input_{Y'}(a)) \mid a \in \Pi\} \models
              \bigwedge_{a \in \Pi} \bigvee_{b \in \Pi}
              \ip_b^{decide_{Y'}(a)}$ holds.
            Therefore,
            $Y' \in \facet{\ISA{2}}$ holds
            because $Y$ and $Y'$ satisfies
            the conditions~\ref{def_simulation_IIS2_ISA2_cond_1}
            and~\ref{def_simulation_IIS2_ISA2_cond_2},
            as shown below.
          \item
            We have $A \subseteq \chi_{\ISA{2}} (X' \cap Y')$
            by definition of $Y'$.
          \item
            To show $Y \mathrel{S} Y'$,
            we examine the three conditions
            of Definition~\ref{def_simulation_IIS2_ISA2}.
            \begin{enumerate}
              \item
                It is clear that $Y'$ satisfies
                the condition~\ref{def_simulation_IIS2_ISA2_cond_1}
                by definition of $input_{Y'}$.
              \item
                Suppose $a \in A$. Then, we have
                $decide_{Y'}(a) = decide_{X'}(a)$ and
                $view_Y(a) = view_X(a)$.
                Since $X \mathrel{S} X'$,
                there exists $a' \in \Pi$ such that
                $(a', decide_{X'}(a)) \in \car(view_X(a))$.
                Hence, we obtain $(a', decide_{Y'}(a)) \in \car(view_Y(a))$.
                Suppose otherwise.
                Then $a \notin A$ implies
                $(0, decide_{Y'}(a)) = (0, input_Y(0)) \in \car(view_Y(a))$
                because $Y \in \IIS{2}_0$ and by Lemma~\ref{lem_flatview},
                all vertices in $\IIS{2}_0$ have
                $(0, input_Y(0))$ in their view.
              \item Suppose $|input_Y(\Pi)| = 3$ and
                $(a, input_Y(a), view_Y(a)) \mathrel{R}^{p,q}
                  (b, input_Y(b), view_Y(b))$,
                  where $a \neq b$ and $p < q$.
                Since $Y \in \IIS{2}_0$,
                it must be $p = 1$ and $q = 2$
                by Lemma~\ref{lem_location_of_Rpq}.

                Suppose $decide_{Y'}(a) = decide_{Y'}(b)$.
                Either
                $decide_{Y'}(a) \neq input_Y(1)$ or
                $decide_{Y'}(b) \neq input_Y(2)$ holds because
                $input_Y(1) \neq input_Y(2)$.

                Suppose otherwise, that is,
                $decide_{Y'}(a) \neq decide_{Y'}(b)$.
                Then, either $decide_{Y'}(a)$ or $decide_{Y'}(b)$
                is equal to $input_Y(0)$ since
                $decide_{Y'}(\Pi) \subseteq
                  decide_{X'}(A) \cup \{input_Y(0)\}$ and
                $|decide_{X'}(A) \cup \{input_Y(0)\}| \leq 2$.
                Thus, either $decide_{Y'}(a) \neq input_Y(1)$ or
                $decide_{Y'}(b) \neq input_Y(2)$ holds, since
                $input_Y(0)$, $input_Y(1)$ and $input_Y(2)$ are distinct.
            \end{enumerate}
        \end{itemize}
      \item Case $|decide_{X'}(A) \cup \{input_Y(0)\}| = 3$.
        If $|A| = 3$, it must be $X = Y$.
        Hence the conditions $Y' \in \facet{\ISA{2}}$,
        $A \subseteq \chi_{\ISA{2}} (X' \cap Y')$ and
        $Y \mathrel{S} Y'$
        are satisfied by taking $Y' = X'$.
        If $|A| \leq 1$, we have
        $|decide_{X'}(A) \cup \{input_Y(0)\}| \leq 2$,
        for which case we have already shown above.

        It remains to prove the case $|A| = 2$.
        By the condition~\ref{def_simulation_IIS2_ISA2_cond_2}
        of $X \mathrel{S} X'$,
        for all $a \in A$ there exists $a' \in \Pi$ such that
        $(a', decide_{X'}(a)) \in \car(view_X(a))$.
        This implies $decide_X(A) \subseteq input_Y(\Pi)$,
        because $\car(view_X(a)) = \car(view_Y(a))
          \subseteq \{(a, input_Y(a)) \mid a \in \Pi\}$.
        Hence, we get
        $\Pi = decide_{X'}(A) \cup \{input_Y(0)\} = input_Y(\Pi)$.
        We have $(0, input_X(0)) \in \car(view_X(a))$
        and $(0, input_Y(0)) \in \car(view_Y(a))$
        for some $a \in A$,
        since $X \cap Y \in \IIS{2}_0$.
        Thus, $view_X(a) = view_Y(a)$ implies $input_X(0) = input_Y(0)$
        and hence $decide_{X'}(A) = input_X(\{1, 2\})$.
        Therefore, $input_X(a) = input_Y(a)$ for all $a \in \Pi$
        because $input_X(\Pi) = \Pi = input_Y(\Pi)$ and $|X \cap Y| = 2$.

        Let us write $a_0$ for the sole agent $a \in \Pi \setminus A$
        and $a_j$ for the agent $a \in A$ such that
        $decide_{X'}(a) = input_X(j)$.
        We also write $y_j$ for a vertex
        $(a_j, input_Y(a_j), view_Y(a_j)) \in Y$.

        We have $y_1, y_2 \in X$, $|input_X(\Pi)| = 3$ and
        $decide_X(a_j) = input_X(j)$ for $j \in \{1, 2\}$.
        Thus, $y_1 \mathrel{R}^{1, 2} y_2$ does not hold
        by the condition~\ref{def_simulation_IIS2_ISA2_cond_3}
        of $X \mathrel{S} X'$.
        Furthermore,
        by the condition~\ref{def_simulation_IIS2_ISA2_cond_2}
        of $X \mathrel{S} X'$, we have
        $1 \in view(y_1)$ and $2 \in view(y_2)$.
        Therefore, $y_1 \mathrel{R}_n^{1, 2} y_2$
        does not hold for any $n \in \N$.

        Suppose by contradiction that both
        $y_1 \mathrel{\bigcup_{n \in \N} R_n^{1, 2}} y_0$ and
        $y_0 \mathrel{\bigcup_{n \in \N} R_n^{1, 2}} y_2$ hold.
        Then, there exists $n_1, n_2 \in \N$ such that
        $y_1 \mathrel{R}_{n_1}^{1, 2} y_0$ and
        $y_0 \mathrel{R}_{n_2}^{1, 2} y_2$.
        By definition of $R_n^{1,2}$, we have
        $y_1 \mathrel{R}_n^{1, 2} y_2$ for any $n \geq 1 + \max\{n_1, n_2\}$,
        which is a contradiction.

        Suppose $y_1 \mathrel{\bigcup_{n \in \N} R_n^{1, 2}} y_0$ does not hold.
        The other case is proven in a similar way.
        Define a facet $Y' \in \ISA{k}$ by
        \[
          Y' = \{(a_0, input_Y(a_0), input_Y(2))\}
            \cup  \{(a, input_Y(a), decide_{X'}(a)) \mid a \in A\}.
        \]
        Then, we have
        $input_{Y'}(a) = input_Y(a)$ for all $a \in \Pi$,
        $decide_{Y'}(a_0) = input_Y(2)$,
        $decide_{Y'}(a_1) = decide_{X'}(a_1)$ and
        $decide_{Y'}(a_2) = decide_{X'}(a_2)$.

        We show that this $Y'$ satisfies
        $Y' \in \facet{\ISA{2}}$,
        $A \subseteq \chi_{\ISA{2}} (X' \cap Y')$ and
        $Y \mathrel{S} Y'$.
        \begin{itemize}
          \item
            Since $decide_{X'}(a_2) = input_X(2) = input_Y(2)$,
            we have $|decide_{Y'}(\Pi)| = |decide_{X'}(A)| = 2$.
            Therefore,
            $Y' \in \facet{\ISA{2}}$ holds,
            similarly to the case $|decide_{X'}(A) \cup \{input_Y(0)\}| \leq 2$,
            because $Y$ and $Y'$ satisfies
            the conditions~\ref{def_simulation_IIS2_ISA2_cond_1}
            and~\ref{def_simulation_IIS2_ISA2_cond_2},
            as shown below.
          \item
            We have $A \subseteq \chi_{\ISA{2}} (X' \cap Y')$
            by definition of $Y'$.
          \item
            Let us show $Y \mathrel{S} Y'$.
            \begin{enumerate}
              \item
                It is clear that $Y'$ satisfies
                the condition~\ref{def_simulation_IIS2_ISA2_cond_1}
                by definition of $input_{Y'}$.
              \item
                Suppose $a \in A$.
                Since $decide_{Y'}(a) = decide_{X'}(a)$,
                this is similarly proved as in the case
                $|decide_{X'}(A) \cup \{input_Y(0)\}| \leq 2$ above.
                Suppose otherwise, i.e.\ $a = a_0$.
                Since $y_1 \mathrel{R}_0^{1, 2} y_0$ does not hold,
                we have $2 \in view(y_0)$.
                Hence, we have
                $(2, decide_{Y'}(a_0)) = (2, input_Y(2)) \in \car(view_Y(a_0))$.
              \item
                Suppose
                $(a, input_Y(a), view_Y(a)) \mathrel{R}^{p,q}
                  (b, input_Y(b), view_Y(b))$ holds,
                  where $a \neq b$ and $p < q$.
                By Lemma~\ref{lem_location_of_Rpq},
                it must be $p = 1$ and $q = 2$.
                If $a \neq a_1$, we have
                $decide_{Y'}(a) = input_Y(2) \neq input_Y(p)$;
                If $a = a_1$ and $b = a_0$,
                $y_1 \mathrel{R}^{1, 2} y_0$ does not hold,
                as we have shown above;
                If $a = a_1$ and $b = a_2$,
                $y_1 \mathrel{R}^{1, 2} y_2$ does not hold,
                since we supposed that
                $y_1 \mathrel{\bigcup_{n \in \N} R_n^{1, 2}} y_0$ does not hold.
            \end{enumerate}
        \end{itemize}
      \end{itemize}

    To show the totality,
    let $X = \{(a, input_X(a), view_X(a)) \mid a \in \Pi\}
      \in \facet{\IIS{2}_p}$.
    Then,
    $X' = \{(a, input_X(a), input_X(p)) \mid a \in \Pi\}$ satisfies that
    $X' \in \facet{\ISA{2}}$ and $X \mathrel{S} X'$.
\end{proof}

\begin{theorem}
  There is no logical obstruction to the solvability of
  $\mathcal{SA}_2$ by $\mathcal{IS}^r$ in $\pfd$,
  whenever $r \geq 2$.

\end{theorem}
\begin{proof}
  For the case $r = 2$, the statement follows from
  Corollary~\ref{cor_knowledge_gain_via_D} and
  Proposition~\ref{prop_existence_D-simulation_SA2_IS2}.

  For the case $r \geq 3$, it suffices to show that
  there exists a morphism $f \colon \IIS{r} \to \IIS{2}$.
  Such a morphism is given by
  $f(I[\gamma_1, \dots, \gamma_r]) = I[\gamma_1, \gamma_2]$.
\end{proof}

We have shown the non-existence of logical obstruction in $\pfd$
for case where $|\Pi| = 3$ and $k = 2$.
For general cases, we conjecture that
for any $\Pi$ and $2 \leq k < |\Pi|$,
there exists $r_{\Pi,k} \geq 2$ such that
the following holds:
There is no logical obstruction to the solvability of
$\mathcal{SA}_k$ by $\mathcal{IS}^r$ in $\pfd$,
whenever $r \geq r_{\Pi,k}$.
We further conjecture that the non-existence
still holds even if common knowledge operator is available.
These are because
$\mathcal{L}_{CD}^+$ has enough expressibility
for path-connectedness, but not enough for simply-connectedness.

%% file: figures/40_two_round_is_1.tex

\begin{tikzpicture}
  \node[circle,scale=0.8,draw,fill=white] (c0vc0v0) at ({0/1+sqrt(3)*0/1}, {0/1+sqrt(3)*0/1}) {};
  \node[circle,scale=0.4,draw,fill=lightgray] (c1vc0v0c1v01) at ({2/3+sqrt(3)*0/1}, {0/1+sqrt(3)*0/1}) {};
  \node[circle,scale=0.4,draw,fill=black] (c2vc0v0c1v01c2v012) at ({7/5+sqrt(3)*0/1}, {0/1+sqrt(3)*3/25}) {};
  \node[circle,scale=0.4,draw,fill=lightgray] (c1vc0v0c1v012) at ({4/5+sqrt(3)*0/1}, {0/1+sqrt(3)*2/5}) {};
  \node[circle,scale=0.4,draw,fill=black] (c2vc0v0c1v012c2v012) at ({39/25+sqrt(3)*0/1}, {0/1+sqrt(3)*3/5}) {};
  \node[circle,scale=0.4,draw,fill=black] (c2vc0v0c1v012c2v02) at ({29/25+sqrt(3)*0/1}, {0/1+sqrt(3)*17/25}) {};
  \node[circle,scale=0.4,draw,fill=lightgray] (c1vc0v0c1v012c2v012) at ({42/25+sqrt(3)*0/1}, {0/1+sqrt(3)*12/25}) {};
  \node[circle,scale=0.4,draw,fill=black] (c2vc0v0c2v012) at ({1/1+sqrt(3)*0/1}, {0/1+sqrt(3)*1/5}) {};
  \node[circle,scale=0.4,draw,fill=lightgray] (c1vc0v0c1v012c2v02) at ({22/25+sqrt(3)*0/1}, {0/1+sqrt(3)*16/25}) {};
  \node[circle,scale=0.4,draw,fill=black] (c2vc0v0c2v02) at ({1/3+sqrt(3)*0/1}, {0/1+sqrt(3)*1/3}) {};
  \node[circle,scale=0.4,draw,fill=lightgray] (c1vc0v0c1v01c2v012) at ({8/5+sqrt(3)*0/1}, {0/1+sqrt(3)*6/25}) {};
  \node[circle,scale=0.6,draw,fill=white] (c0vc0v01) at ({4/1+sqrt(3)*0/1}, {0/1+sqrt(3)*0/1}) {};
  \node[circle,scale=0.4,draw,fill=lightgray] (c1vc0v01c1v01) at ({10/3+sqrt(3)*0/1}, {0/1+sqrt(3)*0/1}) {};
  \node[circle,scale=0.4,draw,fill=black] (c2vc0v01c1v01c2v012) at ({3/1+sqrt(3)*0/1}, {0/1+sqrt(3)*3/25}) {};
  \node[circle,scale=0.4,draw,fill=lightgray] (c1vc0v01c1v01c2v012) at ({16/5+sqrt(3)*0/1}, {0/1+sqrt(3)*6/25}) {};
  \node[circle,scale=0.4,draw,fill=black] (c2vc0v01c2v012) at ({11/3+sqrt(3)*0/1}, {0/1+sqrt(3)*1/5}) {};
  \node[circle,scale=0.4,draw,fill=lightgray] (c1vc0v01c1v1) at ({14/3+sqrt(3)*0/1}, {0/1+sqrt(3)*0/1}) {};
  \node[circle,scale=0.4,draw,fill=black] (c2vc0v01c1v1c2v012) at ({23/5+sqrt(3)*0/1}, {0/1+sqrt(3)*3/25}) {};
  \node[circle,scale=0.4,draw,fill=lightgray] (c1vc0v01c1v1c2v012) at ({4/1+sqrt(3)*0/1}, {0/1+sqrt(3)*6/25}) {};
  \node[circle,scale=0.6,draw,fill=white] (c0vc0v012) at ({18/5+sqrt(3)*0/1}, {0/1+sqrt(3)*6/5}) {};
  \node[circle,scale=0.4,draw,fill=lightgray] (c1vc0v012c1v012) at ({16/5+sqrt(3)*0/1}, {0/1+sqrt(3)*6/5}) {};
  \node[circle,scale=0.4,draw,fill=black] (c2vc0v012c1v012c2v012) at ({3/1+sqrt(3)*0/1}, {0/1+sqrt(3)*27/25}) {};
  \node[circle,scale=0.4,draw,fill=black] (c2vc0v012c1v012c2v2) at ({3/1+sqrt(3)*0/1}, {0/1+sqrt(3)*39/25}) {};
  \node[circle,scale=0.4,draw,fill=lightgray] (c1vc0v012c1v012c2v012) at ({78/25+sqrt(3)*0/1}, {0/1+sqrt(3)*24/25}) {};
  \node[circle,scale=0.4,draw,fill=black] (c2vc0v012c2v012) at ({17/5+sqrt(3)*0/1}, {0/1+sqrt(3)*1/1}) {};
  \node[circle,scale=0.4,draw,fill=lightgray] (c1vc0v012c1v012c2v2) at ({78/25+sqrt(3)*0/1}, {0/1+sqrt(3)*48/25}) {};
  \node[circle,scale=0.4,draw,fill=black] (c2vc0v012c2v2) at ({17/5+sqrt(3)*0/1}, {0/1+sqrt(3)*9/5}) {};
  \node[circle,scale=0.4,draw,fill=lightgray] (c1vc0v012c1v1) at ({22/5+sqrt(3)*0/1}, {0/1+sqrt(3)*4/5}) {};
  \node[circle,scale=0.4,draw,fill=black] (c2vc0v012c1v1c2v012) at ({111/25+sqrt(3)*0/1}, {0/1+sqrt(3)*3/5}) {};
  \node[circle,scale=0.4,draw,fill=black] (c2vc0v012c1v1c2v12) at ({121/25+sqrt(3)*0/1}, {0/1+sqrt(3)*17/25}) {};
  \node[circle,scale=0.4,draw,fill=lightgray] (c1vc0v012c1v12) at ({56/15+sqrt(3)*0/1}, {0/1+sqrt(3)*22/15}) {};
  \node[circle,scale=0.4,draw,fill=black] (c2vc0v012c1v12c2v12) at ({101/25+sqrt(3)*0/1}, {0/1+sqrt(3)*37/25}) {};
  \node[circle,scale=0.4,draw,fill=black] (c2vc0v012c1v12c2v2) at ({91/25+sqrt(3)*0/1}, {0/1+sqrt(3)*47/25}) {};
  \node[circle,scale=0.4,draw,fill=lightgray] (c1vc0v012c1v12c2v12) at ({106/25+sqrt(3)*0/1}, {0/1+sqrt(3)*32/25}) {};
  \node[circle,scale=0.4,draw,fill=black] (c2vc0v012c2v12) at ({61/15+sqrt(3)*0/1}, {0/1+sqrt(3)*17/15}) {};
  \node[circle,scale=0.4,draw,fill=lightgray] (c1vc0v012c1v12c2v2) at ({86/25+sqrt(3)*0/1}, {0/1+sqrt(3)*52/25}) {};
  \node[circle,scale=0.4,draw,fill=lightgray] (c1vc0v012c1v1c2v012) at ({96/25+sqrt(3)*0/1}, {0/1+sqrt(3)*18/25}) {};
  \node[circle,scale=0.4,draw,fill=lightgray] (c1vc0v012c1v1c2v12) at ({116/25+sqrt(3)*0/1}, {0/1+sqrt(3)*22/25}) {};
  \node[circle,scale=0.4,draw,fill=white] (c0vc0v012c1v012) at ({14/5+sqrt(3)*0/1}, {0/1+sqrt(3)*6/5}) {};
  \node[circle,scale=0.6,draw,fill=lightgray] (c1vc1v012) at ({12/5+sqrt(3)*0/1}, {0/1+sqrt(3)*6/5}) {};
  \node[circle,scale=0.4,draw,fill=white] (c0vc0v012c1v012c2v012) at ({72/25+sqrt(3)*0/1}, {0/1+sqrt(3)*24/25}) {};
  \node[circle,scale=0.6,draw,fill=black] (c2vc2v012) at ({3/1+sqrt(3)*0/1}, {0/1+sqrt(3)*3/5}) {};
  \node[circle,scale=0.4,draw,fill=black] (c2vc1v012c2v012) at ({13/5+sqrt(3)*0/1}, {0/1+sqrt(3)*1/1}) {};
  \node[circle,scale=0.4,draw,fill=lightgray] (c1vc1v012c2v012) at ({14/5+sqrt(3)*0/1}, {0/1+sqrt(3)*4/5}) {};
  \node[circle,scale=0.4,draw,fill=white] (c0vc0v012c1v012c2v2) at ({72/25+sqrt(3)*0/1}, {0/1+sqrt(3)*48/25}) {};
  \node[circle,scale=0.8,draw,fill=black] (c2vc2v2) at ({3/1+sqrt(3)*0/1}, {0/1+sqrt(3)*3/1}) {};
  \node[circle,scale=0.4,draw,fill=black] (c2vc1v012c2v2) at ({13/5+sqrt(3)*0/1}, {0/1+sqrt(3)*9/5}) {};
  \node[circle,scale=0.4,draw,fill=lightgray] (c1vc1v012c2v2) at ({14/5+sqrt(3)*0/1}, {0/1+sqrt(3)*12/5}) {};
  \node[circle,scale=0.4,draw,fill=white] (c0vc0v012c1v1) at ({26/5+sqrt(3)*0/1}, {0/1+sqrt(3)*2/5}) {};
  \node[circle,scale=0.8,draw,fill=lightgray] (c1vc1v1) at ({6/1+sqrt(3)*0/1}, {0/1+sqrt(3)*0/1}) {};
  \node[circle,scale=0.4,draw,fill=white] (c0vc0v012c1v12) at ({58/15+sqrt(3)*0/1}, {0/1+sqrt(3)*26/15}) {};
  \node[circle,scale=0.6,draw,fill=lightgray] (c1vc1v12) at ({4/1+sqrt(3)*0/1}, {0/1+sqrt(3)*2/1}) {};
  \node[circle,scale=0.4,draw,fill=white] (c0vc0v012c1v12c2v12) at ({108/25+sqrt(3)*0/1}, {0/1+sqrt(3)*36/25}) {};
  \node[circle,scale=0.6,draw,fill=black] (c2vc2v12) at ({5/1+sqrt(3)*0/1}, {0/1+sqrt(3)*1/1}) {};
  \node[circle,scale=0.4,draw,fill=black] (c2vc1v12c2v12) at ({13/3+sqrt(3)*0/1}, {0/1+sqrt(3)*5/3}) {};
  \node[circle,scale=0.4,draw,fill=lightgray] (c1vc1v12c2v12) at ({14/3+sqrt(3)*0/1}, {0/1+sqrt(3)*4/3}) {};
  \node[circle,scale=0.4,draw,fill=white] (c0vc0v012c1v12c2v2) at ({88/25+sqrt(3)*0/1}, {0/1+sqrt(3)*56/25}) {};
  \node[circle,scale=0.4,draw,fill=black] (c2vc1v12c2v2) at ({11/3+sqrt(3)*0/1}, {0/1+sqrt(3)*7/3}) {};
  \node[circle,scale=0.4,draw,fill=lightgray] (c1vc1v12c2v2) at ({10/3+sqrt(3)*0/1}, {0/1+sqrt(3)*8/3}) {};
  \node[circle,scale=0.4,draw,fill=white] (c0vc0v012c1v1c2v012) at ({108/25+sqrt(3)*0/1}, {0/1+sqrt(3)*12/25}) {};
  \node[circle,scale=0.4,draw,fill=black] (c2vc1v1c2v012) at ({5/1+sqrt(3)*0/1}, {0/1+sqrt(3)*1/5}) {};
  \node[circle,scale=0.4,draw,fill=lightgray] (c1vc1v1c2v012) at ({4/1+sqrt(3)*0/1}, {0/1+sqrt(3)*2/5}) {};
  \node[circle,scale=0.4,draw,fill=white] (c0vc0v012c1v1c2v12) at ({128/25+sqrt(3)*0/1}, {0/1+sqrt(3)*16/25}) {};
  \node[circle,scale=0.4,draw,fill=black] (c2vc1v1c2v12) at ({17/3+sqrt(3)*0/1}, {0/1+sqrt(3)*1/3}) {};
  \node[circle,scale=0.4,draw,fill=lightgray] (c1vc1v1c2v12) at ({16/3+sqrt(3)*0/1}, {0/1+sqrt(3)*2/3}) {};
  \node[circle,scale=0.4,draw,fill=white] (c0vc0v012c2v012) at ({16/5+sqrt(3)*0/1}, {0/1+sqrt(3)*4/5}) {};
  \node[circle,scale=0.4,draw,fill=white] (c0vc0v012c2v12) at ({68/15+sqrt(3)*0/1}, {0/1+sqrt(3)*16/15}) {};
  \node[circle,scale=0.4,draw,fill=white] (c0vc0v012c2v2) at ({16/5+sqrt(3)*0/1}, {0/1+sqrt(3)*12/5}) {};
  \node[circle,scale=0.4,draw,fill=white] (c0vc0v01c1v01) at ({8/3+sqrt(3)*0/1}, {0/1+sqrt(3)*0/1}) {};
  \node[circle,scale=0.6,draw,fill=lightgray] (c1vc1v01) at ({2/1+sqrt(3)*0/1}, {0/1+sqrt(3)*0/1}) {};
  \node[circle,scale=0.4,draw,fill=white] (c0vc0v01c1v01c2v012) at ({14/5+sqrt(3)*0/1}, {0/1+sqrt(3)*6/25}) {};
  \node[circle,scale=0.4,draw,fill=black] (c2vc1v01c2v012) at ({7/3+sqrt(3)*0/1}, {0/1+sqrt(3)*1/5}) {};
  \node[circle,scale=0.4,draw,fill=lightgray] (c1vc1v01c2v012) at ({8/3+sqrt(3)*0/1}, {0/1+sqrt(3)*2/5}) {};
  \node[circle,scale=0.4,draw,fill=white] (c0vc0v01c1v1) at ({16/3+sqrt(3)*0/1}, {0/1+sqrt(3)*0/1}) {};
  \node[circle,scale=0.4,draw,fill=white] (c0vc0v01c1v1c2v012) at ({22/5+sqrt(3)*0/1}, {0/1+sqrt(3)*6/25}) {};
  \node[circle,scale=0.4,draw,fill=white] (c0vc0v01c2v012) at ({10/3+sqrt(3)*0/1}, {0/1+sqrt(3)*2/5}) {};
  \node[circle,scale=0.6,draw,fill=white] (c0vc0v02) at ({2/1+sqrt(3)*0/1}, {0/1+sqrt(3)*2/1}) {};
  \node[circle,scale=0.4,draw,fill=lightgray] (c1vc0v02c1v012) at ({32/15+sqrt(3)*0/1}, {0/1+sqrt(3)*26/15}) {};
  \node[circle,scale=0.4,draw,fill=black] (c2vc0v02c1v012c2v02) at ({49/25+sqrt(3)*0/1}, {0/1+sqrt(3)*37/25}) {};
  \node[circle,scale=0.4,draw,fill=black] (c2vc0v02c1v012c2v2) at ({59/25+sqrt(3)*0/1}, {0/1+sqrt(3)*47/25}) {};
  \node[circle,scale=0.4,draw,fill=lightgray] (c1vc0v02c1v012c2v02) at ({42/25+sqrt(3)*0/1}, {0/1+sqrt(3)*36/25}) {};
  \node[circle,scale=0.4,draw,fill=black] (c2vc0v02c2v02) at ({5/3+sqrt(3)*0/1}, {0/1+sqrt(3)*5/3}) {};
  \node[circle,scale=0.4,draw,fill=lightgray] (c1vc0v02c1v012c2v2) at ({62/25+sqrt(3)*0/1}, {0/1+sqrt(3)*56/25}) {};
  \node[circle,scale=0.4,draw,fill=black] (c2vc0v02c2v2) at ({7/3+sqrt(3)*0/1}, {0/1+sqrt(3)*7/3}) {};
  \node[circle,scale=0.4,draw,fill=white] (c0vc0v02c1v012) at ({34/15+sqrt(3)*0/1}, {0/1+sqrt(3)*22/15}) {};
  \node[circle,scale=0.4,draw,fill=white] (c0vc0v02c1v012c2v02) at ({44/25+sqrt(3)*0/1}, {0/1+sqrt(3)*32/25}) {};
  \node[circle,scale=0.6,draw,fill=black] (c2vc2v02) at ({1/1+sqrt(3)*0/1}, {0/1+sqrt(3)*1/1}) {};
  \node[circle,scale=0.4,draw,fill=black] (c2vc1v012c2v02) at ({29/15+sqrt(3)*0/1}, {0/1+sqrt(3)*17/15}) {};
  \node[circle,scale=0.4,draw,fill=lightgray] (c1vc1v012c2v02) at ({22/15+sqrt(3)*0/1}, {0/1+sqrt(3)*16/15}) {};
  \node[circle,scale=0.4,draw,fill=white] (c0vc0v02c1v012c2v2) at ({64/25+sqrt(3)*0/1}, {0/1+sqrt(3)*52/25}) {};
  \node[circle,scale=0.4,draw,fill=white] (c0vc0v02c2v02) at ({4/3+sqrt(3)*0/1}, {0/1+sqrt(3)*4/3}) {};
  \node[circle,scale=0.4,draw,fill=white] (c0vc0v02c2v2) at ({8/3+sqrt(3)*0/1}, {0/1+sqrt(3)*8/3}) {};
  \node[circle,scale=0.4,draw,fill=white] (c0vc0v0c1v01) at ({4/3+sqrt(3)*0/1}, {0/1+sqrt(3)*0/1}) {};
  \node[circle,scale=0.4,draw,fill=white] (c0vc0v0c1v012) at ({8/5+sqrt(3)*0/1}, {0/1+sqrt(3)*4/5}) {};
  \node[circle,scale=0.4,draw,fill=white] (c0vc0v0c1v012c2v012) at ({54/25+sqrt(3)*0/1}, {0/1+sqrt(3)*18/25}) {};
  \node[circle,scale=0.4,draw,fill=white] (c0vc0v0c1v012c2v02) at ({34/25+sqrt(3)*0/1}, {0/1+sqrt(3)*22/25}) {};
  \node[circle,scale=0.4,draw,fill=white] (c0vc0v0c1v01c2v012) at ({2/1+sqrt(3)*0/1}, {0/1+sqrt(3)*6/25}) {};
  \node[circle,scale=0.4,draw,fill=white] (c0vc0v0c2v012) at ({2/1+sqrt(3)*0/1}, {0/1+sqrt(3)*2/5}) {};
  \node[circle,scale=0.4,draw,fill=white] (c0vc0v0c2v02) at ({2/3+sqrt(3)*0/1}, {0/1+sqrt(3)*2/3}) {};
  \draw[] (c1vc0v0c1v01) -- (c2vc0v0c1v01c2v012);
  \draw[red,latex-] (c0vc0v0) -- (c1vc0v0c1v01);
  \draw[red,latex-] (c0vc0v0) -- (c2vc0v0c1v01c2v012);
  \draw[] (c1vc0v0c1v012) -- (c2vc0v0c1v012c2v012);
  \draw[red,latex-] (c0vc0v0) -- (c1vc0v0c1v012);
  \draw[red,latex-] (c0vc0v0) -- (c2vc0v0c1v012c2v012);
  \draw[] (c1vc0v0c1v012) -- (c2vc0v0c1v012c2v02);
  \draw[red,latex-] (c0vc0v0) -- (c2vc0v0c1v012c2v02);
  \draw[] (c1vc0v0c1v012c2v012) -- (c2vc0v0c1v012c2v012);
  \draw[red,latex-] (c0vc0v0) -- (c1vc0v0c1v012c2v012);
  \draw[] (c1vc0v0c1v012c2v012) -- (c2vc0v0c2v012);
  \draw[red,latex-] (c0vc0v0) -- (c2vc0v0c2v012);
  \draw[] (c1vc0v0c1v012c2v02) -- (c2vc0v0c1v012c2v02);
  \draw[red,latex-] (c0vc0v0) -- (c1vc0v0c1v012c2v02);
  \draw[red,-latex] (c1vc0v0c1v012c2v02) -- (c2vc0v0c2v02);
  \draw[red,semithick] (c0vc0v0) -- (c2vc0v0c2v02);
  \draw[] (c1vc0v0c1v01c2v012) -- (c2vc0v0c1v01c2v012);
  \draw[red,latex-] (c0vc0v0) -- (c1vc0v0c1v01c2v012);
  \draw[] (c1vc0v0c1v01c2v012) -- (c2vc0v0c2v012);
  \draw[] (c0vc0v01) -- (c1vc0v01c1v01);
  \draw[] (c0vc0v01) -- (c2vc0v01c1v01c2v012);
  \draw[] (c1vc0v01c1v01) -- (c2vc0v01c1v01c2v012);
  \draw[] (c0vc0v01) -- (c1vc0v01c1v01c2v012);
  \draw[] (c1vc0v01c1v01c2v012) -- (c2vc0v01c1v01c2v012);
  \draw[] (c0vc0v01) -- (c2vc0v01c2v012);
  \draw[] (c1vc0v01c1v01c2v012) -- (c2vc0v01c2v012);
  \draw[] (c0vc0v01) -- (c1vc0v01c1v1);
  \draw[] (c0vc0v01) -- (c2vc0v01c1v1c2v012);
  \draw[] (c1vc0v01c1v1) -- (c2vc0v01c1v1c2v012);
  \draw[] (c0vc0v01) -- (c1vc0v01c1v1c2v012);
  \draw[] (c1vc0v01c1v1c2v012) -- (c2vc0v01c1v1c2v012);
  \draw[] (c1vc0v01c1v1c2v012) -- (c2vc0v01c2v012);
  \draw[] (c0vc0v012) -- (c1vc0v012c1v012);
  \draw[] (c0vc0v012) -- (c2vc0v012c1v012c2v012);
  \draw[] (c1vc0v012c1v012) -- (c2vc0v012c1v012c2v012);
  \draw[] (c0vc0v012) -- (c2vc0v012c1v012c2v2);
  \draw[] (c1vc0v012c1v012) -- (c2vc0v012c1v012c2v2);
  \draw[] (c0vc0v012) -- (c1vc0v012c1v012c2v012);
  \draw[] (c1vc0v012c1v012c2v012) -- (c2vc0v012c1v012c2v012);
  \draw[] (c0vc0v012) -- (c2vc0v012c2v012);
  \draw[] (c1vc0v012c1v012c2v012) -- (c2vc0v012c2v012);
  \draw[] (c0vc0v012) -- (c1vc0v012c1v012c2v2);
  \draw[] (c1vc0v012c1v012c2v2) -- (c2vc0v012c1v012c2v2);
  \draw[] (c0vc0v012) -- (c2vc0v012c2v2);
  \draw[] (c1vc0v012c1v012c2v2) -- (c2vc0v012c2v2);
  \draw[] (c0vc0v012) -- (c1vc0v012c1v1);
  \draw[] (c0vc0v012) -- (c2vc0v012c1v1c2v012);
  \draw[] (c1vc0v012c1v1) -- (c2vc0v012c1v1c2v012);
  \draw[] (c0vc0v012) -- (c2vc0v012c1v1c2v12);
  \draw[] (c1vc0v012c1v1) -- (c2vc0v012c1v1c2v12);
  \draw[] (c0vc0v012) -- (c1vc0v012c1v12);
  \draw[] (c0vc0v012) -- (c2vc0v012c1v12c2v12);
  \draw[] (c1vc0v012c1v12) -- (c2vc0v012c1v12c2v12);
  \draw[] (c0vc0v012) -- (c2vc0v012c1v12c2v2);
  \draw[] (c1vc0v012c1v12) -- (c2vc0v012c1v12c2v2);
  \draw[] (c0vc0v012) -- (c1vc0v012c1v12c2v12);
  \draw[] (c1vc0v012c1v12c2v12) -- (c2vc0v012c1v12c2v12);
  \draw[] (c0vc0v012) -- (c2vc0v012c2v12);
  \draw[] (c1vc0v012c1v12c2v12) -- (c2vc0v012c2v12);
  \draw[] (c0vc0v012) -- (c1vc0v012c1v12c2v2);
  \draw[] (c1vc0v012c1v12c2v2) -- (c2vc0v012c1v12c2v2);
  \draw[] (c1vc0v012c1v12c2v2) -- (c2vc0v012c2v2);
  \draw[] (c0vc0v012) -- (c1vc0v012c1v1c2v012);
  \draw[] (c1vc0v012c1v1c2v012) -- (c2vc0v012c1v1c2v012);
  \draw[] (c1vc0v012c1v1c2v012) -- (c2vc0v012c2v012);
  \draw[] (c0vc0v012) -- (c1vc0v012c1v1c2v12);
  \draw[] (c1vc0v012c1v1c2v12) -- (c2vc0v012c1v1c2v12);
  \draw[] (c1vc0v012c1v1c2v12) -- (c2vc0v012c2v12);
  \draw[] (c0vc0v012c1v012) -- (c1vc0v012c1v012);
  \draw[] (c0vc0v012c1v012) -- (c2vc0v012c1v012c2v012);
  \draw[] (c0vc0v012c1v012) -- (c2vc0v012c1v012c2v2);
  \draw[] (c0vc0v012c1v012) -- (c1vc1v012);
  \draw[] (c1vc1v012) -- (c2vc0v012c1v012c2v012);
  \draw[] (c1vc1v012) -- (c2vc0v012c1v012c2v2);
  \draw[] (c0vc0v012c1v012c2v012) -- (c1vc0v012c1v012c2v012);
  \draw[] (c0vc0v012c1v012c2v012) -- (c2vc0v012c1v012c2v012);
  \draw[] (c0vc0v012c1v012c2v012) -- (c2vc2v012);
  \draw[] (c1vc0v012c1v012c2v012) -- (c2vc2v012);
  \draw[] (c0vc0v012c1v012c2v012) -- (c1vc1v012);
  \draw[] (c0vc0v012c1v012c2v012) -- (c2vc1v012c2v012);
  \draw[] (c1vc1v012) -- (c2vc1v012c2v012);
  \draw[] (c0vc0v012c1v012c2v012) -- (c1vc1v012c2v012);
  \draw[] (c1vc1v012c2v012) -- (c2vc1v012c2v012);
  \draw[] (c1vc1v012c2v012) -- (c2vc2v012);
  \draw[] (c0vc0v012c1v012c2v2) -- (c1vc0v012c1v012c2v2);
  \draw[] (c0vc0v012c1v012c2v2) -- (c2vc0v012c1v012c2v2);
  \draw[red,semithick] (c0vc0v012c1v012c2v2) -- (c2vc2v2);
  \draw[red,semithick] (c1vc0v012c1v012c2v2) -- (c2vc2v2);
  \draw[] (c0vc0v012c1v012c2v2) -- (c1vc1v012);
  \draw[] (c0vc0v012c1v012c2v2) -- (c2vc1v012c2v2);
  \draw[] (c1vc1v012) -- (c2vc1v012c2v2);
  \draw[] (c0vc0v012c1v012c2v2) -- (c1vc1v012c2v2);
  \draw[] (c1vc1v012c2v2) -- (c2vc1v012c2v2);
  \draw[red,semithick] (c1vc1v012c2v2) -- (c2vc2v2);
  \draw[] (c0vc0v012c1v1) -- (c1vc0v012c1v1);
  \draw[] (c0vc0v012c1v1) -- (c2vc0v012c1v1c2v012);
  \draw[] (c0vc0v012c1v1) -- (c2vc0v012c1v1c2v12);
  \draw[red,-latex] (c1vc1v1) -- (c2vc0v012c1v1c2v012);
  \draw[red,latex-] (c0vc0v012c1v1) -- (c1vc1v1);
  \draw[red,-latex] (c1vc1v1) -- (c2vc0v012c1v1c2v12);
  \draw[] (c0vc0v012c1v12) -- (c1vc0v012c1v12);
  \draw[] (c0vc0v012c1v12) -- (c2vc0v012c1v12c2v12);
  \draw[] (c0vc0v012c1v12) -- (c2vc0v012c1v12c2v2);
  \draw[red,-latex] (c1vc1v12) -- (c2vc0v012c1v12c2v12);
  \draw[red,latex-] (c0vc0v012c1v12) -- (c1vc1v12);
  \draw[red,-latex] (c1vc1v12) -- (c2vc0v012c1v12c2v2);
  \draw[] (c0vc0v012c1v12c2v12) -- (c1vc0v012c1v12c2v12);
  \draw[] (c0vc0v012c1v12c2v12) -- (c2vc0v012c1v12c2v12);
  \draw[red,latex-] (c0vc0v012c1v12c2v12) -- (c2vc2v12);
  \draw[red,latex-] (c1vc0v012c1v12c2v12) -- (c2vc2v12);
  \draw[red,latex-] (c0vc0v012c1v12c2v12) -- (c1vc1v12);
  \draw[red,latex-] (c0vc0v012c1v12c2v12) -- (c2vc1v12c2v12);
  \draw[red,semithick] (c1vc1v12) -- (c2vc1v12c2v12);
  \draw[red,latex-] (c0vc0v012c1v12c2v12) -- (c1vc1v12c2v12);
  \draw[red,semithick] (c1vc1v12c2v12) -- (c2vc1v12c2v12);
  \draw[red,semithick] (c1vc1v12c2v12) -- (c2vc2v12);
  \draw[] (c0vc0v012c1v12c2v2) -- (c1vc0v012c1v12c2v2);
  \draw[] (c0vc0v012c1v12c2v2) -- (c2vc0v012c1v12c2v2);
  \draw[red,semithick] (c0vc0v012c1v12c2v2) -- (c2vc2v2);
  \draw[red,semithick] (c1vc0v012c1v12c2v2) -- (c2vc2v2);
  \draw[red,latex-] (c0vc0v012c1v12c2v2) -- (c1vc1v12);
  \draw[red,latex-] (c0vc0v012c1v12c2v2) -- (c2vc1v12c2v2);
  \draw[red,semithick] (c1vc1v12) -- (c2vc1v12c2v2);
  \draw[red,latex-] (c0vc0v012c1v12c2v2) -- (c1vc1v12c2v2);
  \draw[red,semithick] (c1vc1v12c2v2) -- (c2vc1v12c2v2);
  \draw[red,semithick] (c1vc1v12c2v2) -- (c2vc2v2);
  \draw[] (c0vc0v012c1v1c2v012) -- (c1vc0v012c1v1c2v012);
  \draw[] (c0vc0v012c1v1c2v012) -- (c2vc0v012c1v1c2v012);
  \draw[] (c0vc0v012c1v1c2v012) -- (c2vc2v012);
  \draw[] (c1vc0v012c1v1c2v012) -- (c2vc2v012);
  \draw[red,latex-] (c0vc0v012c1v1c2v012) -- (c1vc1v1);
  \draw[] (c0vc0v012c1v1c2v012) -- (c2vc1v1c2v012);
  \draw[red,-latex] (c1vc1v1) -- (c2vc1v1c2v012);
  \draw[] (c0vc0v012c1v1c2v012) -- (c1vc1v1c2v012);
  \draw[] (c1vc1v1c2v012) -- (c2vc1v1c2v012);
  \draw[] (c1vc1v1c2v012) -- (c2vc2v012);
  \draw[] (c0vc0v012c1v1c2v12) -- (c1vc0v012c1v1c2v12);
  \draw[] (c0vc0v012c1v1c2v12) -- (c2vc0v012c1v1c2v12);
  \draw[red,latex-] (c0vc0v012c1v1c2v12) -- (c2vc2v12);
  \draw[red,latex-] (c1vc0v012c1v1c2v12) -- (c2vc2v12);
  \draw[red,latex-] (c0vc0v012c1v1c2v12) -- (c1vc1v1);
  \draw[red,latex-] (c0vc0v012c1v1c2v12) -- (c2vc1v1c2v12);
  \draw[red,semithick] (c1vc1v1) -- (c2vc1v1c2v12);
  \draw[red,latex-] (c0vc0v012c1v1c2v12) -- (c1vc1v1c2v12);
  \draw[red,semithick] (c1vc1v1c2v12) -- (c2vc1v1c2v12);
  \draw[red,semithick] (c1vc1v1c2v12) -- (c2vc2v12);
  \draw[] (c0vc0v012c2v012) -- (c1vc0v012c1v012c2v012);
  \draw[] (c0vc0v012c2v012) -- (c2vc0v012c2v012);
  \draw[] (c0vc0v012c2v012) -- (c2vc2v012);
  \draw[] (c0vc0v012c2v012) -- (c1vc0v012c1v1c2v012);
  \draw[] (c0vc0v012c2v12) -- (c1vc0v012c1v12c2v12);
  \draw[] (c0vc0v012c2v12) -- (c2vc0v012c2v12);
  \draw[red,latex-] (c0vc0v012c2v12) -- (c2vc2v12);
  \draw[] (c0vc0v012c2v12) -- (c1vc0v012c1v1c2v12);
  \draw[] (c0vc0v012c2v2) -- (c1vc0v012c1v012c2v2);
  \draw[] (c0vc0v012c2v2) -- (c2vc0v012c2v2);
  \draw[red,semithick] (c0vc0v012c2v2) -- (c2vc2v2);
  \draw[] (c0vc0v012c2v2) -- (c1vc0v012c1v12c2v2);
  \draw[] (c0vc0v01c1v01) -- (c1vc0v01c1v01);
  \draw[] (c0vc0v01c1v01) -- (c2vc0v01c1v01c2v012);
  \draw[] (c0vc0v01c1v01) -- (c1vc1v01);
  \draw[] (c1vc1v01) -- (c2vc0v01c1v01c2v012);
  \draw[] (c0vc0v01c1v01c2v012) -- (c1vc0v01c1v01c2v012);
  \draw[] (c0vc0v01c1v01c2v012) -- (c2vc0v01c1v01c2v012);
  \draw[] (c0vc0v01c1v01c2v012) -- (c2vc2v012);
  \draw[] (c1vc0v01c1v01c2v012) -- (c2vc2v012);
  \draw[] (c0vc0v01c1v01c2v012) -- (c1vc1v01);
  \draw[] (c0vc0v01c1v01c2v012) -- (c2vc1v01c2v012);
  \draw[] (c1vc1v01) -- (c2vc1v01c2v012);
  \draw[] (c0vc0v01c1v01c2v012) -- (c1vc1v01c2v012);
  \draw[] (c1vc1v01c2v012) -- (c2vc1v01c2v012);
  \draw[] (c1vc1v01c2v012) -- (c2vc2v012);
  \draw[] (c0vc0v01c1v1) -- (c1vc0v01c1v1);
  \draw[] (c0vc0v01c1v1) -- (c2vc0v01c1v1c2v012);
  \draw[red,-latex] (c1vc1v1) -- (c2vc0v01c1v1c2v012);
  \draw[red,latex-] (c0vc0v01c1v1) -- (c1vc1v1);
  \draw[] (c0vc0v01c1v1c2v012) -- (c1vc0v01c1v1c2v012);
  \draw[] (c0vc0v01c1v1c2v012) -- (c2vc0v01c1v1c2v012);
  \draw[] (c0vc0v01c1v1c2v012) -- (c2vc2v012);
  \draw[] (c1vc0v01c1v1c2v012) -- (c2vc2v012);
  \draw[red,latex-] (c0vc0v01c1v1c2v012) -- (c1vc1v1);
  \draw[] (c0vc0v01c1v1c2v012) -- (c2vc1v1c2v012);
  \draw[] (c0vc0v01c1v1c2v012) -- (c1vc1v1c2v012);
  \draw[] (c0vc0v01c2v012) -- (c1vc0v01c1v01c2v012);
  \draw[] (c0vc0v01c2v012) -- (c2vc0v01c2v012);
  \draw[] (c0vc0v01c2v012) -- (c2vc2v012);
  \draw[] (c0vc0v01c2v012) -- (c1vc0v01c1v1c2v012);
  \draw[] (c1vc0v02c1v012) -- (c2vc0v02c1v012c2v02);
  \draw[red,latex-] (c0vc0v02) -- (c1vc0v02c1v012);
  \draw[red,latex-] (c0vc0v02) -- (c2vc0v02c1v012c2v02);
  \draw[] (c1vc0v02c1v012) -- (c2vc0v02c1v012c2v2);
  \draw[red,latex-] (c0vc0v02) -- (c2vc0v02c1v012c2v2);
  \draw[] (c1vc0v02c1v012c2v02) -- (c2vc0v02c1v012c2v02);
  \draw[red,latex-] (c0vc0v02) -- (c1vc0v02c1v012c2v02);
  \draw[red,-latex] (c1vc0v02c1v012c2v02) -- (c2vc0v02c2v02);
  \draw[red,semithick] (c0vc0v02) -- (c2vc0v02c2v02);
  \draw[] (c1vc0v02c1v012c2v2) -- (c2vc0v02c1v012c2v2);
  \draw[red,latex-] (c0vc0v02) -- (c1vc0v02c1v012c2v2);
  \draw[red,-latex] (c1vc0v02c1v012c2v2) -- (c2vc0v02c2v2);
  \draw[red,semithick] (c0vc0v02) -- (c2vc0v02c2v2);
  \draw[] (c0vc0v02c1v012) -- (c1vc0v02c1v012);
  \draw[] (c0vc0v02c1v012) -- (c2vc0v02c1v012c2v02);
  \draw[] (c0vc0v02c1v012) -- (c2vc0v02c1v012c2v2);
  \draw[] (c0vc0v02c1v012) -- (c1vc1v012);
  \draw[] (c1vc1v012) -- (c2vc0v02c1v012c2v02);
  \draw[] (c1vc1v012) -- (c2vc0v02c1v012c2v2);
  \draw[] (c0vc0v02c1v012c2v02) -- (c1vc0v02c1v012c2v02);
  \draw[] (c0vc0v02c1v012c2v02) -- (c2vc0v02c1v012c2v02);
  \draw[red,-latex] (c0vc0v02c1v012c2v02) -- (c2vc2v02);
  \draw[red,-latex] (c1vc0v02c1v012c2v02) -- (c2vc2v02);
  \draw[] (c0vc0v02c1v012c2v02) -- (c1vc1v012);
  \draw[] (c0vc0v02c1v012c2v02) -- (c2vc1v012c2v02);
  \draw[] (c1vc1v012) -- (c2vc1v012c2v02);
  \draw[] (c0vc0v02c1v012c2v02) -- (c1vc1v012c2v02);
  \draw[] (c1vc1v012c2v02) -- (c2vc1v012c2v02);
  \draw[red,-latex] (c1vc1v012c2v02) -- (c2vc2v02);
  \draw[] (c0vc0v02c1v012c2v2) -- (c1vc0v02c1v012c2v2);
  \draw[] (c0vc0v02c1v012c2v2) -- (c2vc0v02c1v012c2v2);
  \draw[red,semithick] (c0vc0v02c1v012c2v2) -- (c2vc2v2);
  \draw[red,semithick] (c1vc0v02c1v012c2v2) -- (c2vc2v2);
  \draw[] (c0vc0v02c1v012c2v2) -- (c1vc1v012);
  \draw[] (c0vc0v02c1v012c2v2) -- (c2vc1v012c2v2);
  \draw[] (c0vc0v02c1v012c2v2) -- (c1vc1v012c2v2);
  \draw[red,latex-] (c0vc0v02c2v02) -- (c1vc0v02c1v012c2v02);
  \draw[red,semithick] (c0vc0v02c2v02) -- (c2vc0v02c2v02);
  \draw[red,semithick] (c0vc0v02c2v02) -- (c2vc2v02);
  \draw[red,latex-] (c0vc0v02c2v2) -- (c1vc0v02c1v012c2v2);
  \draw[red,semithick] (c0vc0v02c2v2) -- (c2vc0v02c2v2);
  \draw[red,semithick] (c0vc0v02c2v2) -- (c2vc2v2);
  \draw[] (c0vc0v0c1v01) -- (c1vc0v0c1v01);
  \draw[] (c0vc0v0c1v01) -- (c2vc0v0c1v01c2v012);
  \draw[] (c0vc0v0c1v01) -- (c1vc1v01);
  \draw[] (c1vc1v01) -- (c2vc0v0c1v01c2v012);
  \draw[] (c0vc0v0c1v012) -- (c1vc0v0c1v012);
  \draw[] (c0vc0v0c1v012) -- (c2vc0v0c1v012c2v012);
  \draw[] (c0vc0v0c1v012) -- (c2vc0v0c1v012c2v02);
  \draw[] (c0vc0v0c1v012) -- (c1vc1v012);
  \draw[] (c1vc1v012) -- (c2vc0v0c1v012c2v012);
  \draw[] (c1vc1v012) -- (c2vc0v0c1v012c2v02);
  \draw[] (c0vc0v0c1v012c2v012) -- (c1vc0v0c1v012c2v012);
  \draw[] (c0vc0v0c1v012c2v012) -- (c2vc0v0c1v012c2v012);
  \draw[] (c0vc0v0c1v012c2v012) -- (c2vc2v012);
  \draw[] (c1vc0v0c1v012c2v012) -- (c2vc2v012);
  \draw[] (c0vc0v0c1v012c2v012) -- (c1vc1v012);
  \draw[] (c0vc0v0c1v012c2v012) -- (c2vc1v012c2v012);
  \draw[] (c0vc0v0c1v012c2v012) -- (c1vc1v012c2v012);
  \draw[] (c0vc0v0c1v012c2v02) -- (c1vc0v0c1v012c2v02);
  \draw[] (c0vc0v0c1v012c2v02) -- (c2vc0v0c1v012c2v02);
  \draw[red,-latex] (c0vc0v0c1v012c2v02) -- (c2vc2v02);
  \draw[red,-latex] (c1vc0v0c1v012c2v02) -- (c2vc2v02);
  \draw[] (c0vc0v0c1v012c2v02) -- (c1vc1v012);
  \draw[] (c0vc0v0c1v012c2v02) -- (c2vc1v012c2v02);
  \draw[] (c0vc0v0c1v012c2v02) -- (c1vc1v012c2v02);
  \draw[] (c0vc0v0c1v01c2v012) -- (c1vc0v0c1v01c2v012);
  \draw[] (c0vc0v0c1v01c2v012) -- (c2vc0v0c1v01c2v012);
  \draw[] (c0vc0v0c1v01c2v012) -- (c2vc2v012);
  \draw[] (c1vc0v0c1v01c2v012) -- (c2vc2v012);
  \draw[] (c0vc0v0c1v01c2v012) -- (c1vc1v01);
  \draw[] (c0vc0v0c1v01c2v012) -- (c2vc1v01c2v012);
  \draw[] (c0vc0v0c1v01c2v012) -- (c1vc1v01c2v012);
  \draw[] (c0vc0v0c2v012) -- (c1vc0v0c1v012c2v012);
  \draw[] (c0vc0v0c2v012) -- (c2vc0v0c2v012);
  \draw[] (c0vc0v0c2v012) -- (c2vc2v012);
  \draw[] (c0vc0v0c2v012) -- (c1vc0v0c1v01c2v012);
  \draw[red,latex-] (c0vc0v0c2v02) -- (c1vc0v0c1v012c2v02);
  \draw[red,semithick] (c0vc0v0c2v02) -- (c2vc0v0c2v02);
  \draw[red,semithick] (c0vc0v0c2v02) -- (c2vc2v02);

\end{tikzpicture}

%% file: figures/41_two_round_is_2.tex

\begin{tikzpicture}
\node[circle,scale=0.8,draw,fill=white] (c0vc0v0) at ({0/1+sqrt(3)*0/1}, {0/1+sqrt(3)*0/1}) {};
\node[circle,scale=0.4,draw,fill=lightgray] (c1vc0v0c1v01) at ({2/3+sqrt(3)*0/1}, {0/1+sqrt(3)*0/1}) {};
\node[circle,scale=0.4,draw,fill=black] (c2vc0v0c1v01c2v012) at ({7/5+sqrt(3)*0/1}, {0/1+sqrt(3)*3/25}) {};
\node[circle,scale=0.4,draw,fill=lightgray] (c1vc0v0c1v012) at ({4/5+sqrt(3)*0/1}, {0/1+sqrt(3)*2/5}) {};
\node[circle,scale=0.4,draw,fill=black] (c2vc0v0c1v012c2v012) at ({39/25+sqrt(3)*0/1}, {0/1+sqrt(3)*3/5}) {};
\node[circle,scale=0.4,draw,fill=black] (c2vc0v0c1v012c2v02) at ({29/25+sqrt(3)*0/1}, {0/1+sqrt(3)*17/25}) {};
\node[circle,scale=0.4,draw,fill=lightgray] (c1vc0v0c1v012c2v012) at ({42/25+sqrt(3)*0/1}, {0/1+sqrt(3)*12/25}) {};
\node[circle,scale=0.4,draw,fill=black] (c2vc0v0c2v012) at ({1/1+sqrt(3)*0/1}, {0/1+sqrt(3)*1/5}) {};
\node[circle,scale=0.4,draw,fill=lightgray] (c1vc0v0c1v012c2v02) at ({22/25+sqrt(3)*0/1}, {0/1+sqrt(3)*16/25}) {};
\node[circle,scale=0.4,draw,fill=black] (c2vc0v0c2v02) at ({1/3+sqrt(3)*0/1}, {0/1+sqrt(3)*1/3}) {};
\node[circle,scale=0.4,draw,fill=lightgray] (c1vc0v0c1v01c2v012) at ({8/5+sqrt(3)*0/1}, {0/1+sqrt(3)*6/25}) {};
\node[circle,scale=0.6,draw,fill=white] (c0vc0v01) at ({4/1+sqrt(3)*0/1}, {0/1+sqrt(3)*0/1}) {};
\node[circle,scale=0.4,draw,fill=lightgray] (c1vc0v01c1v01) at ({10/3+sqrt(3)*0/1}, {0/1+sqrt(3)*0/1}) {};
\node[circle,scale=0.4,draw,fill=black] (c2vc0v01c1v01c2v012) at ({3/1+sqrt(3)*0/1}, {0/1+sqrt(3)*3/25}) {};
\node[circle,scale=0.4,draw,fill=lightgray] (c1vc0v01c1v01c2v012) at ({16/5+sqrt(3)*0/1}, {0/1+sqrt(3)*6/25}) {};
\node[circle,scale=0.4,draw,fill=black] (c2vc0v01c2v012) at ({11/3+sqrt(3)*0/1}, {0/1+sqrt(3)*1/5}) {};
\node[circle,scale=0.4,draw,fill=lightgray] (c1vc0v01c1v1) at ({14/3+sqrt(3)*0/1}, {0/1+sqrt(3)*0/1}) {};
\node[circle,scale=0.4,draw,fill=black] (c2vc0v01c1v1c2v012) at ({23/5+sqrt(3)*0/1}, {0/1+sqrt(3)*3/25}) {};
\node[circle,scale=0.4,draw,fill=lightgray] (c1vc0v01c1v1c2v012) at ({4/1+sqrt(3)*0/1}, {0/1+sqrt(3)*6/25}) {};
\node[circle,scale=0.6,draw,fill=white] (c0vc0v012) at ({18/5+sqrt(3)*0/1}, {0/1+sqrt(3)*6/5}) {};
\node[circle,scale=0.4,draw,fill=lightgray] (c1vc0v012c1v012) at ({16/5+sqrt(3)*0/1}, {0/1+sqrt(3)*6/5}) {};
\node[circle,scale=0.4,draw,fill=black] (c2vc0v012c1v012c2v012) at ({3/1+sqrt(3)*0/1}, {0/1+sqrt(3)*27/25}) {};
\node[circle,scale=0.4,draw,fill=black] (c2vc0v012c1v012c2v2) at ({3/1+sqrt(3)*0/1}, {0/1+sqrt(3)*39/25}) {};
\node[circle,scale=0.4,draw,fill=lightgray] (c1vc0v012c1v012c2v012) at ({78/25+sqrt(3)*0/1}, {0/1+sqrt(3)*24/25}) {};
\node[circle,scale=0.4,draw,fill=black] (c2vc0v012c2v012) at ({17/5+sqrt(3)*0/1}, {0/1+sqrt(3)*1/1}) {};
\node[circle,scale=0.4,draw,fill=lightgray] (c1vc0v012c1v012c2v2) at ({78/25+sqrt(3)*0/1}, {0/1+sqrt(3)*48/25}) {};
\node[circle,scale=0.4,draw,fill=black] (c2vc0v012c2v2) at ({17/5+sqrt(3)*0/1}, {0/1+sqrt(3)*9/5}) {};
\node[circle,scale=0.4,draw,fill=lightgray] (c1vc0v012c1v1) at ({22/5+sqrt(3)*0/1}, {0/1+sqrt(3)*4/5}) {};
\node[circle,scale=0.4,draw,fill=black] (c2vc0v012c1v1c2v012) at ({111/25+sqrt(3)*0/1}, {0/1+sqrt(3)*3/5}) {};
\node[circle,scale=0.4,draw,fill=black] (c2vc0v012c1v1c2v12) at ({121/25+sqrt(3)*0/1}, {0/1+sqrt(3)*17/25}) {};
\node[circle,scale=0.4,draw,fill=lightgray] (c1vc0v012c1v12) at ({56/15+sqrt(3)*0/1}, {0/1+sqrt(3)*22/15}) {};
\node[circle,scale=0.4,draw,fill=black] (c2vc0v012c1v12c2v12) at ({101/25+sqrt(3)*0/1}, {0/1+sqrt(3)*37/25}) {};
\node[circle,scale=0.4,draw,fill=black] (c2vc0v012c1v12c2v2) at ({91/25+sqrt(3)*0/1}, {0/1+sqrt(3)*47/25}) {};
\node[circle,scale=0.4,draw,fill=lightgray] (c1vc0v012c1v12c2v12) at ({106/25+sqrt(3)*0/1}, {0/1+sqrt(3)*32/25}) {};
\node[circle,scale=0.4,draw,fill=black] (c2vc0v012c2v12) at ({61/15+sqrt(3)*0/1}, {0/1+sqrt(3)*17/15}) {};
\node[circle,scale=0.4,draw,fill=lightgray] (c1vc0v012c1v12c2v2) at ({86/25+sqrt(3)*0/1}, {0/1+sqrt(3)*52/25}) {};
\node[circle,scale=0.4,draw,fill=lightgray] (c1vc0v012c1v1c2v012) at ({96/25+sqrt(3)*0/1}, {0/1+sqrt(3)*18/25}) {};
\node[circle,scale=0.4,draw,fill=lightgray] (c1vc0v012c1v1c2v12) at ({116/25+sqrt(3)*0/1}, {0/1+sqrt(3)*22/25}) {};
\node[circle,scale=0.4,draw,fill=white] (c0vc0v012c1v012) at ({14/5+sqrt(3)*0/1}, {0/1+sqrt(3)*6/5}) {};
\node[circle,scale=0.6,draw,fill=lightgray] (c1vc1v012) at ({12/5+sqrt(3)*0/1}, {0/1+sqrt(3)*6/5}) {};
\node[circle,scale=0.4,draw,fill=white] (c0vc0v012c1v012c2v012) at ({72/25+sqrt(3)*0/1}, {0/1+sqrt(3)*24/25}) {};
\node[circle,scale=0.6,draw,fill=black] (c2vc2v012) at ({3/1+sqrt(3)*0/1}, {0/1+sqrt(3)*3/5}) {};
\node[circle,scale=0.4,draw,fill=black] (c2vc1v012c2v012) at ({13/5+sqrt(3)*0/1}, {0/1+sqrt(3)*1/1}) {};
\node[circle,scale=0.4,draw,fill=lightgray] (c1vc1v012c2v012) at ({14/5+sqrt(3)*0/1}, {0/1+sqrt(3)*4/5}) {};
\node[circle,scale=0.4,draw,fill=white] (c0vc0v012c1v012c2v2) at ({72/25+sqrt(3)*0/1}, {0/1+sqrt(3)*48/25}) {};
\node[circle,scale=0.8,draw,fill=black] (c2vc2v2) at ({3/1+sqrt(3)*0/1}, {0/1+sqrt(3)*3/1}) {};
\node[circle,scale=0.4,draw,fill=black] (c2vc1v012c2v2) at ({13/5+sqrt(3)*0/1}, {0/1+sqrt(3)*9/5}) {};
\node[circle,scale=0.4,draw,fill=lightgray] (c1vc1v012c2v2) at ({14/5+sqrt(3)*0/1}, {0/1+sqrt(3)*12/5}) {};
\node[circle,scale=0.4,draw,fill=white] (c0vc0v012c1v1) at ({26/5+sqrt(3)*0/1}, {0/1+sqrt(3)*2/5}) {};
\node[circle,scale=0.8,draw,fill=lightgray] (c1vc1v1) at ({6/1+sqrt(3)*0/1}, {0/1+sqrt(3)*0/1}) {};
\node[circle,scale=0.4,draw,fill=white] (c0vc0v012c1v12) at ({58/15+sqrt(3)*0/1}, {0/1+sqrt(3)*26/15}) {};
\node[circle,scale=0.6,draw,fill=lightgray] (c1vc1v12) at ({4/1+sqrt(3)*0/1}, {0/1+sqrt(3)*2/1}) {};
\node[circle,scale=0.4,draw,fill=white] (c0vc0v012c1v12c2v12) at ({108/25+sqrt(3)*0/1}, {0/1+sqrt(3)*36/25}) {};
\node[circle,scale=0.6,draw,fill=black] (c2vc2v12) at ({5/1+sqrt(3)*0/1}, {0/1+sqrt(3)*1/1}) {};
\node[circle,scale=0.4,draw,fill=black] (c2vc1v12c2v12) at ({13/3+sqrt(3)*0/1}, {0/1+sqrt(3)*5/3}) {};
\node[circle,scale=0.4,draw,fill=lightgray] (c1vc1v12c2v12) at ({14/3+sqrt(3)*0/1}, {0/1+sqrt(3)*4/3}) {};
\node[circle,scale=0.4,draw,fill=white] (c0vc0v012c1v12c2v2) at ({88/25+sqrt(3)*0/1}, {0/1+sqrt(3)*56/25}) {};
\node[circle,scale=0.4,draw,fill=black] (c2vc1v12c2v2) at ({11/3+sqrt(3)*0/1}, {0/1+sqrt(3)*7/3}) {};
\node[circle,scale=0.4,draw,fill=lightgray] (c1vc1v12c2v2) at ({10/3+sqrt(3)*0/1}, {0/1+sqrt(3)*8/3}) {};
\node[circle,scale=0.4,draw,fill=white] (c0vc0v012c1v1c2v012) at ({108/25+sqrt(3)*0/1}, {0/1+sqrt(3)*12/25}) {};
\node[circle,scale=0.4,draw,fill=black] (c2vc1v1c2v012) at ({5/1+sqrt(3)*0/1}, {0/1+sqrt(3)*1/5}) {};
\node[circle,scale=0.4,draw,fill=lightgray] (c1vc1v1c2v012) at ({4/1+sqrt(3)*0/1}, {0/1+sqrt(3)*2/5}) {};
\node[circle,scale=0.4,draw,fill=white] (c0vc0v012c1v1c2v12) at ({128/25+sqrt(3)*0/1}, {0/1+sqrt(3)*16/25}) {};
\node[circle,scale=0.4,draw,fill=black] (c2vc1v1c2v12) at ({17/3+sqrt(3)*0/1}, {0/1+sqrt(3)*1/3}) {};
\node[circle,scale=0.4,draw,fill=lightgray] (c1vc1v1c2v12) at ({16/3+sqrt(3)*0/1}, {0/1+sqrt(3)*2/3}) {};
\node[circle,scale=0.4,draw,fill=white] (c0vc0v012c2v012) at ({16/5+sqrt(3)*0/1}, {0/1+sqrt(3)*4/5}) {};
\node[circle,scale=0.4,draw,fill=white] (c0vc0v012c2v12) at ({68/15+sqrt(3)*0/1}, {0/1+sqrt(3)*16/15}) {};
\node[circle,scale=0.4,draw,fill=white] (c0vc0v012c2v2) at ({16/5+sqrt(3)*0/1}, {0/1+sqrt(3)*12/5}) {};
\node[circle,scale=0.4,draw,fill=white] (c0vc0v01c1v01) at ({8/3+sqrt(3)*0/1}, {0/1+sqrt(3)*0/1}) {};
\node[circle,scale=0.6,draw,fill=lightgray] (c1vc1v01) at ({2/1+sqrt(3)*0/1}, {0/1+sqrt(3)*0/1}) {};
\node[circle,scale=0.4,draw,fill=white] (c0vc0v01c1v01c2v012) at ({14/5+sqrt(3)*0/1}, {0/1+sqrt(3)*6/25}) {};
\node[circle,scale=0.4,draw,fill=black] (c2vc1v01c2v012) at ({7/3+sqrt(3)*0/1}, {0/1+sqrt(3)*1/5}) {};
\node[circle,scale=0.4,draw,fill=lightgray] (c1vc1v01c2v012) at ({8/3+sqrt(3)*0/1}, {0/1+sqrt(3)*2/5}) {};
\node[circle,scale=0.4,draw,fill=white] (c0vc0v01c1v1) at ({16/3+sqrt(3)*0/1}, {0/1+sqrt(3)*0/1}) {};
\node[circle,scale=0.4,draw,fill=white] (c0vc0v01c1v1c2v012) at ({22/5+sqrt(3)*0/1}, {0/1+sqrt(3)*6/25}) {};
\node[circle,scale=0.4,draw,fill=white] (c0vc0v01c2v012) at ({10/3+sqrt(3)*0/1}, {0/1+sqrt(3)*2/5}) {};
\node[circle,scale=0.6,draw,fill=white] (c0vc0v02) at ({2/1+sqrt(3)*0/1}, {0/1+sqrt(3)*2/1}) {};
\node[circle,scale=0.4,draw,fill=lightgray] (c1vc0v02c1v012) at ({32/15+sqrt(3)*0/1}, {0/1+sqrt(3)*26/15}) {};
\node[circle,scale=0.4,draw,fill=black] (c2vc0v02c1v012c2v02) at ({49/25+sqrt(3)*0/1}, {0/1+sqrt(3)*37/25}) {};
\node[circle,scale=0.4,draw,fill=black] (c2vc0v02c1v012c2v2) at ({59/25+sqrt(3)*0/1}, {0/1+sqrt(3)*47/25}) {};
\node[circle,scale=0.4,draw,fill=lightgray] (c1vc0v02c1v012c2v02) at ({42/25+sqrt(3)*0/1}, {0/1+sqrt(3)*36/25}) {};
\node[circle,scale=0.4,draw,fill=black] (c2vc0v02c2v02) at ({5/3+sqrt(3)*0/1}, {0/1+sqrt(3)*5/3}) {};
\node[circle,scale=0.4,draw,fill=lightgray] (c1vc0v02c1v012c2v2) at ({62/25+sqrt(3)*0/1}, {0/1+sqrt(3)*56/25}) {};
\node[circle,scale=0.4,draw,fill=black] (c2vc0v02c2v2) at ({7/3+sqrt(3)*0/1}, {0/1+sqrt(3)*7/3}) {};
\node[circle,scale=0.4,draw,fill=white] (c0vc0v02c1v012) at ({34/15+sqrt(3)*0/1}, {0/1+sqrt(3)*22/15}) {};
\node[circle,scale=0.4,draw,fill=white] (c0vc0v02c1v012c2v02) at ({44/25+sqrt(3)*0/1}, {0/1+sqrt(3)*32/25}) {};
\node[circle,scale=0.6,draw,fill=black] (c2vc2v02) at ({1/1+sqrt(3)*0/1}, {0/1+sqrt(3)*1/1}) {};
\node[circle,scale=0.4,draw,fill=black] (c2vc1v012c2v02) at ({29/15+sqrt(3)*0/1}, {0/1+sqrt(3)*17/15}) {};
\node[circle,scale=0.4,draw,fill=lightgray] (c1vc1v012c2v02) at ({22/15+sqrt(3)*0/1}, {0/1+sqrt(3)*16/15}) {};
\node[circle,scale=0.4,draw,fill=white] (c0vc0v02c1v012c2v2) at ({64/25+sqrt(3)*0/1}, {0/1+sqrt(3)*52/25}) {};
\node[circle,scale=0.4,draw,fill=white] (c0vc0v02c2v02) at ({4/3+sqrt(3)*0/1}, {0/1+sqrt(3)*4/3}) {};
\node[circle,scale=0.4,draw,fill=white] (c0vc0v02c2v2) at ({8/3+sqrt(3)*0/1}, {0/1+sqrt(3)*8/3}) {};
\node[circle,scale=0.4,draw,fill=white] (c0vc0v0c1v01) at ({4/3+sqrt(3)*0/1}, {0/1+sqrt(3)*0/1}) {};
\node[circle,scale=0.4,draw,fill=white] (c0vc0v0c1v012) at ({8/5+sqrt(3)*0/1}, {0/1+sqrt(3)*4/5}) {};
\node[circle,scale=0.4,draw,fill=white] (c0vc0v0c1v012c2v012) at ({54/25+sqrt(3)*0/1}, {0/1+sqrt(3)*18/25}) {};
\node[circle,scale=0.4,draw,fill=white] (c0vc0v0c1v012c2v02) at ({34/25+sqrt(3)*0/1}, {0/1+sqrt(3)*22/25}) {};
\node[circle,scale=0.4,draw,fill=white] (c0vc0v0c1v01c2v012) at ({2/1+sqrt(3)*0/1}, {0/1+sqrt(3)*6/25}) {};
\node[circle,scale=0.4,draw,fill=white] (c0vc0v0c2v012) at ({2/1+sqrt(3)*0/1}, {0/1+sqrt(3)*2/5}) {};
\node[circle,scale=0.4,draw,fill=white] (c0vc0v0c2v02) at ({2/3+sqrt(3)*0/1}, {0/1+sqrt(3)*2/3}) {};
\draw[] (c1vc0v0c1v01) -- (c2vc0v0c1v01c2v012);
\draw[red,latex-] (c0vc0v0) -- (c1vc0v0c1v01);
\draw[red,latex-] (c0vc0v0) -- (c2vc0v0c1v01c2v012);
\draw[] (c1vc0v0c1v012) -- (c2vc0v0c1v012c2v012);
\draw[red,latex-] (c0vc0v0) -- (c1vc0v0c1v012);
\draw[red,latex-] (c0vc0v0) -- (c2vc0v0c1v012c2v012);
\draw[] (c1vc0v0c1v012) -- (c2vc0v0c1v012c2v02);
\draw[red,latex-] (c0vc0v0) -- (c2vc0v0c1v012c2v02);
\draw[] (c1vc0v0c1v012c2v012) -- (c2vc0v0c1v012c2v012);
\draw[red,latex-] (c0vc0v0) -- (c1vc0v0c1v012c2v012);
\draw[] (c1vc0v0c1v012c2v012) -- (c2vc0v0c2v012);
\draw[red,latex-] (c0vc0v0) -- (c2vc0v0c2v012);
\draw[] (c1vc0v0c1v012c2v02) -- (c2vc0v0c1v012c2v02);
\draw[red,latex-] (c0vc0v0) -- (c1vc0v0c1v012c2v02);
\draw[red,-latex] (c1vc0v0c1v012c2v02) -- (c2vc0v0c2v02);
\draw[red,semithick] (c0vc0v0) -- (c2vc0v0c2v02);
\draw[] (c1vc0v0c1v01c2v012) -- (c2vc0v0c1v01c2v012);
\draw[red,latex-] (c0vc0v0) -- (c1vc0v0c1v01c2v012);
\draw[] (c1vc0v0c1v01c2v012) -- (c2vc0v0c2v012);
\draw[] (c0vc0v01) -- (c1vc0v01c1v01);
\draw[] (c0vc0v01) -- (c2vc0v01c1v01c2v012);
\draw[] (c1vc0v01c1v01) -- (c2vc0v01c1v01c2v012);
\draw[] (c0vc0v01) -- (c1vc0v01c1v01c2v012);
\draw[] (c1vc0v01c1v01c2v012) -- (c2vc0v01c1v01c2v012);
\draw[] (c0vc0v01) -- (c2vc0v01c2v012);
\draw[] (c1vc0v01c1v01c2v012) -- (c2vc0v01c2v012);
\draw[] (c0vc0v01) -- (c1vc0v01c1v1);
\draw[] (c0vc0v01) -- (c2vc0v01c1v1c2v012);
\draw[] (c1vc0v01c1v1) -- (c2vc0v01c1v1c2v012);
\draw[] (c0vc0v01) -- (c1vc0v01c1v1c2v012);
\draw[] (c1vc0v01c1v1c2v012) -- (c2vc0v01c1v1c2v012);
\draw[] (c1vc0v01c1v1c2v012) -- (c2vc0v01c2v012);
\draw[] (c0vc0v012) -- (c1vc0v012c1v012);
\draw[] (c0vc0v012) -- (c2vc0v012c1v012c2v012);
\draw[] (c1vc0v012c1v012) -- (c2vc0v012c1v012c2v012);
\draw[] (c0vc0v012) -- (c2vc0v012c1v012c2v2);
\draw[] (c1vc0v012c1v012) -- (c2vc0v012c1v012c2v2);
\draw[] (c0vc0v012) -- (c1vc0v012c1v012c2v012);
\draw[] (c1vc0v012c1v012c2v012) -- (c2vc0v012c1v012c2v012);
\draw[] (c0vc0v012) -- (c2vc0v012c2v012);
\draw[] (c1vc0v012c1v012c2v012) -- (c2vc0v012c2v012);
\draw[] (c0vc0v012) -- (c1vc0v012c1v012c2v2);
\draw[] (c1vc0v012c1v012c2v2) -- (c2vc0v012c1v012c2v2);
\draw[] (c0vc0v012) -- (c2vc0v012c2v2);
\draw[] (c1vc0v012c1v012c2v2) -- (c2vc0v012c2v2);
\draw[] (c0vc0v012) -- (c1vc0v012c1v1);
\draw[] (c0vc0v012) -- (c2vc0v012c1v1c2v012);
\draw[] (c1vc0v012c1v1) -- (c2vc0v012c1v1c2v012);
\draw[] (c0vc0v012) -- (c2vc0v012c1v1c2v12);
\draw[] (c1vc0v012c1v1) -- (c2vc0v012c1v1c2v12);
\draw[] (c0vc0v012) -- (c1vc0v012c1v12);
\draw[] (c0vc0v012) -- (c2vc0v012c1v12c2v12);
\draw[] (c1vc0v012c1v12) -- (c2vc0v012c1v12c2v12);
\draw[] (c0vc0v012) -- (c2vc0v012c1v12c2v2);
\draw[] (c1vc0v012c1v12) -- (c2vc0v012c1v12c2v2);
\draw[] (c0vc0v012) -- (c1vc0v012c1v12c2v12);
\draw[] (c1vc0v012c1v12c2v12) -- (c2vc0v012c1v12c2v12);
\draw[] (c0vc0v012) -- (c2vc0v012c2v12);
\draw[] (c1vc0v012c1v12c2v12) -- (c2vc0v012c2v12);
\draw[] (c0vc0v012) -- (c1vc0v012c1v12c2v2);
\draw[red,-latex] (c1vc0v012c1v12c2v2) -- (c2vc0v012c1v12c2v2);
\draw[] (c1vc0v012c1v12c2v2) -- (c2vc0v012c2v2);
\draw[] (c0vc0v012) -- (c1vc0v012c1v1c2v012);
\draw[] (c1vc0v012c1v1c2v012) -- (c2vc0v012c1v1c2v012);
\draw[] (c1vc0v012c1v1c2v012) -- (c2vc0v012c2v012);
\draw[] (c0vc0v012) -- (c1vc0v012c1v1c2v12);
\draw[] (c1vc0v012c1v1c2v12) -- (c2vc0v012c1v1c2v12);
\draw[] (c1vc0v012c1v1c2v12) -- (c2vc0v012c2v12);
\draw[] (c0vc0v012c1v012) -- (c1vc0v012c1v012);
\draw[] (c0vc0v012c1v012) -- (c2vc0v012c1v012c2v012);
\draw[] (c0vc0v012c1v012) -- (c2vc0v012c1v012c2v2);
\draw[] (c0vc0v012c1v012) -- (c1vc1v012);
\draw[] (c1vc1v012) -- (c2vc0v012c1v012c2v012);
\draw[] (c1vc1v012) -- (c2vc0v012c1v012c2v2);
\draw[] (c0vc0v012c1v012c2v012) -- (c1vc0v012c1v012c2v012);
\draw[] (c0vc0v012c1v012c2v012) -- (c2vc0v012c1v012c2v012);
\draw[] (c0vc0v012c1v012c2v012) -- (c2vc2v012);
\draw[] (c1vc0v012c1v012c2v012) -- (c2vc2v012);
\draw[] (c0vc0v012c1v012c2v012) -- (c1vc1v012);
\draw[] (c0vc0v012c1v012c2v012) -- (c2vc1v012c2v012);
\draw[] (c1vc1v012) -- (c2vc1v012c2v012);
\draw[] (c0vc0v012c1v012c2v012) -- (c1vc1v012c2v012);
\draw[] (c1vc1v012c2v012) -- (c2vc1v012c2v012);
\draw[] (c1vc1v012c2v012) -- (c2vc2v012);
\draw[red,semithick] (c0vc0v012c1v012c2v2) -- (c1vc0v012c1v012c2v2);
\draw[] (c0vc0v012c1v012c2v2) -- (c2vc0v012c1v012c2v2);
\draw[red,semithick] (c0vc0v012c1v012c2v2) -- (c2vc2v2);
\draw[red,semithick] (c1vc0v012c1v012c2v2) -- (c2vc2v2);
\draw[] (c0vc0v012c1v012c2v2) -- (c1vc1v012);
\draw[] (c0vc0v012c1v012c2v2) -- (c2vc1v012c2v2);
\draw[] (c1vc1v012) -- (c2vc1v012c2v2);
\draw[red,semithick] (c0vc0v012c1v012c2v2) -- (c1vc1v012c2v2);
\draw[] (c1vc1v012c2v2) -- (c2vc1v012c2v2);
\draw[red,semithick] (c1vc1v012c2v2) -- (c2vc2v2);
\draw[] (c0vc0v012c1v1) -- (c1vc0v012c1v1);
\draw[] (c0vc0v012c1v1) -- (c2vc0v012c1v1c2v012);
\draw[] (c0vc0v012c1v1) -- (c2vc0v012c1v1c2v12);
\draw[red,-latex] (c1vc1v1) -- (c2vc0v012c1v1c2v012);
\draw[red,latex-] (c0vc0v012c1v1) -- (c1vc1v1);
\draw[red,-latex] (c1vc1v1) -- (c2vc0v012c1v1c2v12);
\draw[] (c0vc0v012c1v12) -- (c1vc0v012c1v12);
\draw[] (c0vc0v012c1v12) -- (c2vc0v012c1v12c2v12);
\draw[] (c0vc0v012c1v12) -- (c2vc0v012c1v12c2v2);
\draw[red,-latex] (c1vc1v12) -- (c2vc0v012c1v12c2v12);
\draw[red,latex-] (c0vc0v012c1v12) -- (c1vc1v12);
\draw[red,-latex] (c1vc1v12) -- (c2vc0v012c1v12c2v2);
\draw[] (c0vc0v012c1v12c2v12) -- (c1vc0v012c1v12c2v12);
\draw[] (c0vc0v012c1v12c2v12) -- (c2vc0v012c1v12c2v12);
\draw[red,latex-] (c0vc0v012c1v12c2v12) -- (c2vc2v12);
\draw[red,latex-] (c1vc0v012c1v12c2v12) -- (c2vc2v12);
\draw[red,latex-] (c0vc0v012c1v12c2v12) -- (c1vc1v12);
\draw[red,latex-] (c0vc0v012c1v12c2v12) -- (c2vc1v12c2v12);
\draw[red,semithick] (c1vc1v12) -- (c2vc1v12c2v12);
\draw[red,latex-] (c0vc0v012c1v12c2v12) -- (c1vc1v12c2v12);
\draw[red,semithick] (c1vc1v12c2v12) -- (c2vc1v12c2v12);
\draw[red,semithick] (c1vc1v12c2v12) -- (c2vc2v12);
\draw[red,semithick] (c0vc0v012c1v12c2v2) -- (c1vc0v012c1v12c2v2);
\draw[red,-latex] (c0vc0v012c1v12c2v2) -- (c2vc0v012c1v12c2v2);
\draw[red,semithick] (c0vc0v012c1v12c2v2) -- (c2vc2v2);
\draw[red,semithick] (c1vc0v012c1v12c2v2) -- (c2vc2v2);
\draw[red,semithick] (c0vc0v012c1v12c2v2) -- (c1vc1v12);
\draw[red,semithick] (c0vc0v012c1v12c2v2) -- (c2vc1v12c2v2);
\draw[red,semithick] (c1vc1v12) -- (c2vc1v12c2v2);
\draw[red,semithick] (c0vc0v012c1v12c2v2) -- (c1vc1v12c2v2);
\draw[red,semithick] (c1vc1v12c2v2) -- (c2vc1v12c2v2);
\draw[red,semithick] (c1vc1v12c2v2) -- (c2vc2v2);
\draw[] (c0vc0v012c1v1c2v012) -- (c1vc0v012c1v1c2v012);
\draw[] (c0vc0v012c1v1c2v012) -- (c2vc0v012c1v1c2v012);
\draw[] (c0vc0v012c1v1c2v012) -- (c2vc2v012);
\draw[] (c1vc0v012c1v1c2v012) -- (c2vc2v012);
\draw[red,latex-] (c0vc0v012c1v1c2v012) -- (c1vc1v1);
\draw[] (c0vc0v012c1v1c2v012) -- (c2vc1v1c2v012);
\draw[red,-latex] (c1vc1v1) -- (c2vc1v1c2v012);
\draw[] (c0vc0v012c1v1c2v012) -- (c1vc1v1c2v012);
\draw[] (c1vc1v1c2v012) -- (c2vc1v1c2v012);
\draw[] (c1vc1v1c2v012) -- (c2vc2v012);
\draw[] (c0vc0v012c1v1c2v12) -- (c1vc0v012c1v1c2v12);
\draw[] (c0vc0v012c1v1c2v12) -- (c2vc0v012c1v1c2v12);
\draw[red,latex-] (c0vc0v012c1v1c2v12) -- (c2vc2v12);
\draw[red,latex-] (c1vc0v012c1v1c2v12) -- (c2vc2v12);
\draw[red,latex-] (c0vc0v012c1v1c2v12) -- (c1vc1v1);
\draw[red,latex-] (c0vc0v012c1v1c2v12) -- (c2vc1v1c2v12);
\draw[red,semithick] (c1vc1v1) -- (c2vc1v1c2v12);
\draw[red,latex-] (c0vc0v012c1v1c2v12) -- (c1vc1v1c2v12);
\draw[red,semithick] (c1vc1v1c2v12) -- (c2vc1v1c2v12);
\draw[red,semithick] (c1vc1v1c2v12) -- (c2vc2v12);
\draw[] (c0vc0v012c2v012) -- (c1vc0v012c1v012c2v012);
\draw[] (c0vc0v012c2v012) -- (c2vc0v012c2v012);
\draw[] (c0vc0v012c2v012) -- (c2vc2v012);
\draw[] (c0vc0v012c2v012) -- (c1vc0v012c1v1c2v012);
\draw[] (c0vc0v012c2v12) -- (c1vc0v012c1v12c2v12);
\draw[] (c0vc0v012c2v12) -- (c2vc0v012c2v12);
\draw[red,latex-] (c0vc0v012c2v12) -- (c2vc2v12);
\draw[] (c0vc0v012c2v12) -- (c1vc0v012c1v1c2v12);
\draw[red,semithick] (c0vc0v012c2v2) -- (c1vc0v012c1v012c2v2);
\draw[] (c0vc0v012c2v2) -- (c2vc0v012c2v2);
\draw[red,semithick] (c0vc0v012c2v2) -- (c2vc2v2);
\draw[red,semithick] (c0vc0v012c2v2) -- (c1vc0v012c1v12c2v2);
\draw[] (c0vc0v01c1v01) -- (c1vc0v01c1v01);
\draw[] (c0vc0v01c1v01) -- (c2vc0v01c1v01c2v012);
\draw[] (c0vc0v01c1v01) -- (c1vc1v01);
\draw[] (c1vc1v01) -- (c2vc0v01c1v01c2v012);
\draw[] (c0vc0v01c1v01c2v012) -- (c1vc0v01c1v01c2v012);
\draw[] (c0vc0v01c1v01c2v012) -- (c2vc0v01c1v01c2v012);
\draw[] (c0vc0v01c1v01c2v012) -- (c2vc2v012);
\draw[] (c1vc0v01c1v01c2v012) -- (c2vc2v012);
\draw[] (c0vc0v01c1v01c2v012) -- (c1vc1v01);
\draw[] (c0vc0v01c1v01c2v012) -- (c2vc1v01c2v012);
\draw[] (c1vc1v01) -- (c2vc1v01c2v012);
\draw[] (c0vc0v01c1v01c2v012) -- (c1vc1v01c2v012);
\draw[] (c1vc1v01c2v012) -- (c2vc1v01c2v012);
\draw[] (c1vc1v01c2v012) -- (c2vc2v012);
\draw[] (c0vc0v01c1v1) -- (c1vc0v01c1v1);
\draw[] (c0vc0v01c1v1) -- (c2vc0v01c1v1c2v012);
\draw[red,-latex] (c1vc1v1) -- (c2vc0v01c1v1c2v012);
\draw[red,latex-] (c0vc0v01c1v1) -- (c1vc1v1);
\draw[] (c0vc0v01c1v1c2v012) -- (c1vc0v01c1v1c2v012);
\draw[] (c0vc0v01c1v1c2v012) -- (c2vc0v01c1v1c2v012);
\draw[] (c0vc0v01c1v1c2v012) -- (c2vc2v012);
\draw[] (c1vc0v01c1v1c2v012) -- (c2vc2v012);
\draw[red,latex-] (c0vc0v01c1v1c2v012) -- (c1vc1v1);
\draw[] (c0vc0v01c1v1c2v012) -- (c2vc1v1c2v012);
\draw[] (c0vc0v01c1v1c2v012) -- (c1vc1v1c2v012);
\draw[] (c0vc0v01c2v012) -- (c1vc0v01c1v01c2v012);
\draw[] (c0vc0v01c2v012) -- (c2vc0v01c2v012);
\draw[] (c0vc0v01c2v012) -- (c2vc2v012);
\draw[] (c0vc0v01c2v012) -- (c1vc0v01c1v1c2v012);
\draw[] (c1vc0v02c1v012) -- (c2vc0v02c1v012c2v02);
\draw[red,latex-] (c0vc0v02) -- (c1vc0v02c1v012);
\draw[red,latex-] (c0vc0v02) -- (c2vc0v02c1v012c2v02);
\draw[] (c1vc0v02c1v012) -- (c2vc0v02c1v012c2v2);
\draw[red,latex-] (c0vc0v02) -- (c2vc0v02c1v012c2v2);
\draw[] (c1vc0v02c1v012c2v02) -- (c2vc0v02c1v012c2v02);
\draw[red,latex-] (c0vc0v02) -- (c1vc0v02c1v012c2v02);
\draw[red,-latex] (c1vc0v02c1v012c2v02) -- (c2vc0v02c2v02);
\draw[red,semithick] (c0vc0v02) -- (c2vc0v02c2v02);
\draw[red,latex-] (c1vc0v02c1v012c2v2) -- (c2vc0v02c1v012c2v2);
\draw[red,semithick] (c0vc0v02) -- (c1vc0v02c1v012c2v2);
\draw[red,semithick] (c1vc0v02c1v012c2v2) -- (c2vc0v02c2v2);
\draw[red,semithick] (c0vc0v02) -- (c2vc0v02c2v2);
\draw[] (c0vc0v02c1v012) -- (c1vc0v02c1v012);
\draw[] (c0vc0v02c1v012) -- (c2vc0v02c1v012c2v02);
\draw[] (c0vc0v02c1v012) -- (c2vc0v02c1v012c2v2);
\draw[] (c0vc0v02c1v012) -- (c1vc1v012);
\draw[] (c1vc1v012) -- (c2vc0v02c1v012c2v02);
\draw[] (c1vc1v012) -- (c2vc0v02c1v012c2v2);
\draw[] (c0vc0v02c1v012c2v02) -- (c1vc0v02c1v012c2v02);
\draw[] (c0vc0v02c1v012c2v02) -- (c2vc0v02c1v012c2v02);
\draw[red,-latex] (c0vc0v02c1v012c2v02) -- (c2vc2v02);
\draw[red,-latex] (c1vc0v02c1v012c2v02) -- (c2vc2v02);
\draw[] (c0vc0v02c1v012c2v02) -- (c1vc1v012);
\draw[] (c0vc0v02c1v012c2v02) -- (c2vc1v012c2v02);
\draw[] (c1vc1v012) -- (c2vc1v012c2v02);
\draw[] (c0vc0v02c1v012c2v02) -- (c1vc1v012c2v02);
\draw[] (c1vc1v012c2v02) -- (c2vc1v012c2v02);
\draw[red,-latex] (c1vc1v012c2v02) -- (c2vc2v02);
\draw[red,semithick] (c0vc0v02c1v012c2v2) -- (c1vc0v02c1v012c2v2);
\draw[red,latex-] (c0vc0v02c1v012c2v2) -- (c2vc0v02c1v012c2v2);
\draw[red,semithick] (c0vc0v02c1v012c2v2) -- (c2vc2v2);
\draw[red,semithick] (c1vc0v02c1v012c2v2) -- (c2vc2v2);
\draw[] (c0vc0v02c1v012c2v2) -- (c1vc1v012);
\draw[] (c0vc0v02c1v012c2v2) -- (c2vc1v012c2v2);
\draw[red,semithick] (c0vc0v02c1v012c2v2) -- (c1vc1v012c2v2);
\draw[red,latex-] (c0vc0v02c2v02) -- (c1vc0v02c1v012c2v02);
\draw[red,semithick] (c0vc0v02c2v02) -- (c2vc0v02c2v02);
\draw[red,semithick] (c0vc0v02c2v02) -- (c2vc2v02);
\draw[red,semithick] (c0vc0v02c2v2) -- (c1vc0v02c1v012c2v2);
\draw[red,semithick] (c0vc0v02c2v2) -- (c2vc0v02c2v2);
\draw[red,semithick] (c0vc0v02c2v2) -- (c2vc2v2);
\draw[] (c0vc0v0c1v01) -- (c1vc0v0c1v01);
\draw[] (c0vc0v0c1v01) -- (c2vc0v0c1v01c2v012);
\draw[] (c0vc0v0c1v01) -- (c1vc1v01);
\draw[] (c1vc1v01) -- (c2vc0v0c1v01c2v012);
\draw[] (c0vc0v0c1v012) -- (c1vc0v0c1v012);
\draw[] (c0vc0v0c1v012) -- (c2vc0v0c1v012c2v012);
\draw[] (c0vc0v0c1v012) -- (c2vc0v0c1v012c2v02);
\draw[] (c0vc0v0c1v012) -- (c1vc1v012);
\draw[] (c1vc1v012) -- (c2vc0v0c1v012c2v012);
\draw[] (c1vc1v012) -- (c2vc0v0c1v012c2v02);
\draw[] (c0vc0v0c1v012c2v012) -- (c1vc0v0c1v012c2v012);
\draw[] (c0vc0v0c1v012c2v012) -- (c2vc0v0c1v012c2v012);
\draw[] (c0vc0v0c1v012c2v012) -- (c2vc2v012);
\draw[] (c1vc0v0c1v012c2v012) -- (c2vc2v012);
\draw[] (c0vc0v0c1v012c2v012) -- (c1vc1v012);
\draw[] (c0vc0v0c1v012c2v012) -- (c2vc1v012c2v012);
\draw[] (c0vc0v0c1v012c2v012) -- (c1vc1v012c2v012);
\draw[] (c0vc0v0c1v012c2v02) -- (c1vc0v0c1v012c2v02);
\draw[] (c0vc0v0c1v012c2v02) -- (c2vc0v0c1v012c2v02);
\draw[red,-latex] (c0vc0v0c1v012c2v02) -- (c2vc2v02);
\draw[red,-latex] (c1vc0v0c1v012c2v02) -- (c2vc2v02);
\draw[] (c0vc0v0c1v012c2v02) -- (c1vc1v012);
\draw[] (c0vc0v0c1v012c2v02) -- (c2vc1v012c2v02);
\draw[] (c0vc0v0c1v012c2v02) -- (c1vc1v012c2v02);
\draw[] (c0vc0v0c1v01c2v012) -- (c1vc0v0c1v01c2v012);
\draw[] (c0vc0v0c1v01c2v012) -- (c2vc0v0c1v01c2v012);
\draw[] (c0vc0v0c1v01c2v012) -- (c2vc2v012);
\draw[] (c1vc0v0c1v01c2v012) -- (c2vc2v012);
\draw[] (c0vc0v0c1v01c2v012) -- (c1vc1v01);
\draw[] (c0vc0v0c1v01c2v012) -- (c2vc1v01c2v012);
\draw[] (c0vc0v0c1v01c2v012) -- (c1vc1v01c2v012);
\draw[] (c0vc0v0c2v012) -- (c1vc0v0c1v012c2v012);
\draw[] (c0vc0v0c2v012) -- (c2vc0v0c2v012);
\draw[] (c0vc0v0c2v012) -- (c2vc2v012);
\draw[] (c0vc0v0c2v012) -- (c1vc0v0c1v01c2v012);
\draw[red,latex-] (c0vc0v0c2v02) -- (c1vc0v0c1v012c2v02);
\draw[red,semithick] (c0vc0v0c2v02) -- (c2vc0v0c2v02);
\draw[red,semithick] (c0vc0v0c2v02) -- (c2vc2v02);


\end{tikzpicture}

%% file: figures/42_two_round_is_3.tex

\begin{tikzpicture}
\node[circle,scale=0.8,draw,fill=white] (c0vc0v0) at ({0/1+sqrt(3)*0/1}, {0/1+sqrt(3)*0/1}) {};
\node[circle,scale=0.4,draw,fill=lightgray] (c1vc0v0c1v01) at ({2/3+sqrt(3)*0/1}, {0/1+sqrt(3)*0/1}) {};
\node[circle,scale=0.4,draw,fill=black] (c2vc0v0c1v01c2v012) at ({7/5+sqrt(3)*0/1}, {0/1+sqrt(3)*3/25}) {};
\node[circle,scale=0.4,draw,fill=lightgray] (c1vc0v0c1v012) at ({4/5+sqrt(3)*0/1}, {0/1+sqrt(3)*2/5}) {};
\node[circle,scale=0.4,draw,fill=black] (c2vc0v0c1v012c2v012) at ({39/25+sqrt(3)*0/1}, {0/1+sqrt(3)*3/5}) {};
\node[circle,scale=0.4,draw,fill=black] (c2vc0v0c1v012c2v02) at ({29/25+sqrt(3)*0/1}, {0/1+sqrt(3)*17/25}) {};
\node[circle,scale=0.4,draw,fill=lightgray] (c1vc0v0c1v012c2v012) at ({42/25+sqrt(3)*0/1}, {0/1+sqrt(3)*12/25}) {};
\node[circle,scale=0.4,draw,fill=black] (c2vc0v0c2v012) at ({1/1+sqrt(3)*0/1}, {0/1+sqrt(3)*1/5}) {};
\node[circle,scale=0.4,draw,fill=lightgray] (c1vc0v0c1v012c2v02) at ({22/25+sqrt(3)*0/1}, {0/1+sqrt(3)*16/25}) {};
\node[circle,scale=0.4,draw,fill=black] (c2vc0v0c2v02) at ({1/3+sqrt(3)*0/1}, {0/1+sqrt(3)*1/3}) {};
\node[circle,scale=0.4,draw,fill=lightgray] (c1vc0v0c1v01c2v012) at ({8/5+sqrt(3)*0/1}, {0/1+sqrt(3)*6/25}) {};
\node[circle,scale=0.6,draw,fill=white] (c0vc0v01) at ({4/1+sqrt(3)*0/1}, {0/1+sqrt(3)*0/1}) {};
\node[circle,scale=0.4,draw,fill=lightgray] (c1vc0v01c1v01) at ({10/3+sqrt(3)*0/1}, {0/1+sqrt(3)*0/1}) {};
\node[circle,scale=0.4,draw,fill=black] (c2vc0v01c1v01c2v012) at ({3/1+sqrt(3)*0/1}, {0/1+sqrt(3)*3/25}) {};
\node[circle,scale=0.4,draw,fill=lightgray] (c1vc0v01c1v01c2v012) at ({16/5+sqrt(3)*0/1}, {0/1+sqrt(3)*6/25}) {};
\node[circle,scale=0.4,draw,fill=black] (c2vc0v01c2v012) at ({11/3+sqrt(3)*0/1}, {0/1+sqrt(3)*1/5}) {};
\node[circle,scale=0.4,draw,fill=lightgray] (c1vc0v01c1v1) at ({14/3+sqrt(3)*0/1}, {0/1+sqrt(3)*0/1}) {};
\node[circle,scale=0.4,draw,fill=black] (c2vc0v01c1v1c2v012) at ({23/5+sqrt(3)*0/1}, {0/1+sqrt(3)*3/25}) {};
\node[circle,scale=0.4,draw,fill=lightgray] (c1vc0v01c1v1c2v012) at ({4/1+sqrt(3)*0/1}, {0/1+sqrt(3)*6/25}) {};
\node[circle,scale=0.6,draw,fill=white] (c0vc0v012) at ({18/5+sqrt(3)*0/1}, {0/1+sqrt(3)*6/5}) {};
\node[circle,scale=0.4,draw,fill=lightgray] (c1vc0v012c1v012) at ({16/5+sqrt(3)*0/1}, {0/1+sqrt(3)*6/5}) {};
\node[circle,scale=0.4,draw,fill=black] (c2vc0v012c1v012c2v012) at ({3/1+sqrt(3)*0/1}, {0/1+sqrt(3)*27/25}) {};
\node[circle,scale=0.4,draw,fill=black] (c2vc0v012c1v012c2v2) at ({3/1+sqrt(3)*0/1}, {0/1+sqrt(3)*39/25}) {};
\node[circle,scale=0.4,draw,fill=lightgray] (c1vc0v012c1v012c2v012) at ({78/25+sqrt(3)*0/1}, {0/1+sqrt(3)*24/25}) {};
\node[circle,scale=0.4,draw,fill=black] (c2vc0v012c2v012) at ({17/5+sqrt(3)*0/1}, {0/1+sqrt(3)*1/1}) {};
\node[circle,scale=0.4,draw,fill=lightgray] (c1vc0v012c1v012c2v2) at ({78/25+sqrt(3)*0/1}, {0/1+sqrt(3)*48/25}) {};
\node[circle,scale=0.4,draw,fill=black] (c2vc0v012c2v2) at ({17/5+sqrt(3)*0/1}, {0/1+sqrt(3)*9/5}) {};
\node[circle,scale=0.4,draw,fill=lightgray] (c1vc0v012c1v1) at ({22/5+sqrt(3)*0/1}, {0/1+sqrt(3)*4/5}) {};
\node[circle,scale=0.4,draw,fill=black] (c2vc0v012c1v1c2v012) at ({111/25+sqrt(3)*0/1}, {0/1+sqrt(3)*3/5}) {};
\node[circle,scale=0.4,draw,fill=black] (c2vc0v012c1v1c2v12) at ({121/25+sqrt(3)*0/1}, {0/1+sqrt(3)*17/25}) {};
\node[circle,scale=0.4,draw,fill=lightgray] (c1vc0v012c1v12) at ({56/15+sqrt(3)*0/1}, {0/1+sqrt(3)*22/15}) {};
\node[circle,scale=0.4,draw,fill=black] (c2vc0v012c1v12c2v12) at ({101/25+sqrt(3)*0/1}, {0/1+sqrt(3)*37/25}) {};
\node[circle,scale=0.4,draw,fill=black] (c2vc0v012c1v12c2v2) at ({91/25+sqrt(3)*0/1}, {0/1+sqrt(3)*47/25}) {};
\node[circle,scale=0.4,draw,fill=lightgray] (c1vc0v012c1v12c2v12) at ({106/25+sqrt(3)*0/1}, {0/1+sqrt(3)*32/25}) {};
\node[circle,scale=0.4,draw,fill=black] (c2vc0v012c2v12) at ({61/15+sqrt(3)*0/1}, {0/1+sqrt(3)*17/15}) {};
\node[circle,scale=0.4,draw,fill=lightgray] (c1vc0v012c1v12c2v2) at ({86/25+sqrt(3)*0/1}, {0/1+sqrt(3)*52/25}) {};
\node[circle,scale=0.4,draw,fill=lightgray] (c1vc0v012c1v1c2v012) at ({96/25+sqrt(3)*0/1}, {0/1+sqrt(3)*18/25}) {};
\node[circle,scale=0.4,draw,fill=lightgray] (c1vc0v012c1v1c2v12) at ({116/25+sqrt(3)*0/1}, {0/1+sqrt(3)*22/25}) {};
\node[circle,scale=0.4,draw,fill=white] (c0vc0v012c1v012) at ({14/5+sqrt(3)*0/1}, {0/1+sqrt(3)*6/5}) {};
\node[circle,scale=0.6,draw,fill=lightgray] (c1vc1v012) at ({12/5+sqrt(3)*0/1}, {0/1+sqrt(3)*6/5}) {};
\node[circle,scale=0.4,draw,fill=white] (c0vc0v012c1v012c2v012) at ({72/25+sqrt(3)*0/1}, {0/1+sqrt(3)*24/25}) {};
\node[circle,scale=0.6,draw,fill=black] (c2vc2v012) at ({3/1+sqrt(3)*0/1}, {0/1+sqrt(3)*3/5}) {};
\node[circle,scale=0.4,draw,fill=black] (c2vc1v012c2v012) at ({13/5+sqrt(3)*0/1}, {0/1+sqrt(3)*1/1}) {};
\node[circle,scale=0.4,draw,fill=lightgray] (c1vc1v012c2v012) at ({14/5+sqrt(3)*0/1}, {0/1+sqrt(3)*4/5}) {};
\node[circle,scale=0.4,draw,fill=white] (c0vc0v012c1v012c2v2) at ({72/25+sqrt(3)*0/1}, {0/1+sqrt(3)*48/25}) {};
\node[circle,scale=0.8,draw,fill=black] (c2vc2v2) at ({3/1+sqrt(3)*0/1}, {0/1+sqrt(3)*3/1}) {};
\node[circle,scale=0.4,draw,fill=black] (c2vc1v012c2v2) at ({13/5+sqrt(3)*0/1}, {0/1+sqrt(3)*9/5}) {};
\node[circle,scale=0.4,draw,fill=lightgray] (c1vc1v012c2v2) at ({14/5+sqrt(3)*0/1}, {0/1+sqrt(3)*12/5}) {};
\node[circle,scale=0.4,draw,fill=white] (c0vc0v012c1v1) at ({26/5+sqrt(3)*0/1}, {0/1+sqrt(3)*2/5}) {};
\node[circle,scale=0.8,draw,fill=lightgray] (c1vc1v1) at ({6/1+sqrt(3)*0/1}, {0/1+sqrt(3)*0/1}) {};
\node[circle,scale=0.4,draw,fill=white] (c0vc0v012c1v12) at ({58/15+sqrt(3)*0/1}, {0/1+sqrt(3)*26/15}) {};
\node[circle,scale=0.6,draw,fill=lightgray] (c1vc1v12) at ({4/1+sqrt(3)*0/1}, {0/1+sqrt(3)*2/1}) {};
\node[circle,scale=0.4,draw,fill=white] (c0vc0v012c1v12c2v12) at ({108/25+sqrt(3)*0/1}, {0/1+sqrt(3)*36/25}) {};
\node[circle,scale=0.6,draw,fill=black] (c2vc2v12) at ({5/1+sqrt(3)*0/1}, {0/1+sqrt(3)*1/1}) {};
\node[circle,scale=0.4,draw,fill=black] (c2vc1v12c2v12) at ({13/3+sqrt(3)*0/1}, {0/1+sqrt(3)*5/3}) {};
\node[circle,scale=0.4,draw,fill=lightgray] (c1vc1v12c2v12) at ({14/3+sqrt(3)*0/1}, {0/1+sqrt(3)*4/3}) {};
\node[circle,scale=0.4,draw,fill=white] (c0vc0v012c1v12c2v2) at ({88/25+sqrt(3)*0/1}, {0/1+sqrt(3)*56/25}) {};
\node[circle,scale=0.4,draw,fill=black] (c2vc1v12c2v2) at ({11/3+sqrt(3)*0/1}, {0/1+sqrt(3)*7/3}) {};
\node[circle,scale=0.4,draw,fill=lightgray] (c1vc1v12c2v2) at ({10/3+sqrt(3)*0/1}, {0/1+sqrt(3)*8/3}) {};
\node[circle,scale=0.4,draw,fill=white] (c0vc0v012c1v1c2v012) at ({108/25+sqrt(3)*0/1}, {0/1+sqrt(3)*12/25}) {};
\node[circle,scale=0.4,draw,fill=black] (c2vc1v1c2v012) at ({5/1+sqrt(3)*0/1}, {0/1+sqrt(3)*1/5}) {};
\node[circle,scale=0.4,draw,fill=lightgray] (c1vc1v1c2v012) at ({4/1+sqrt(3)*0/1}, {0/1+sqrt(3)*2/5}) {};
\node[circle,scale=0.4,draw,fill=white] (c0vc0v012c1v1c2v12) at ({128/25+sqrt(3)*0/1}, {0/1+sqrt(3)*16/25}) {};
\node[circle,scale=0.4,draw,fill=black] (c2vc1v1c2v12) at ({17/3+sqrt(3)*0/1}, {0/1+sqrt(3)*1/3}) {};
\node[circle,scale=0.4,draw,fill=lightgray] (c1vc1v1c2v12) at ({16/3+sqrt(3)*0/1}, {0/1+sqrt(3)*2/3}) {};
\node[circle,scale=0.4,draw,fill=white] (c0vc0v012c2v012) at ({16/5+sqrt(3)*0/1}, {0/1+sqrt(3)*4/5}) {};
\node[circle,scale=0.4,draw,fill=white] (c0vc0v012c2v12) at ({68/15+sqrt(3)*0/1}, {0/1+sqrt(3)*16/15}) {};
\node[circle,scale=0.4,draw,fill=white] (c0vc0v012c2v2) at ({16/5+sqrt(3)*0/1}, {0/1+sqrt(3)*12/5}) {};
\node[circle,scale=0.4,draw,fill=white] (c0vc0v01c1v01) at ({8/3+sqrt(3)*0/1}, {0/1+sqrt(3)*0/1}) {};
\node[circle,scale=0.6,draw,fill=lightgray] (c1vc1v01) at ({2/1+sqrt(3)*0/1}, {0/1+sqrt(3)*0/1}) {};
\node[circle,scale=0.4,draw,fill=white] (c0vc0v01c1v01c2v012) at ({14/5+sqrt(3)*0/1}, {0/1+sqrt(3)*6/25}) {};
\node[circle,scale=0.4,draw,fill=black] (c2vc1v01c2v012) at ({7/3+sqrt(3)*0/1}, {0/1+sqrt(3)*1/5}) {};
\node[circle,scale=0.4,draw,fill=lightgray] (c1vc1v01c2v012) at ({8/3+sqrt(3)*0/1}, {0/1+sqrt(3)*2/5}) {};
\node[circle,scale=0.4,draw,fill=white] (c0vc0v01c1v1) at ({16/3+sqrt(3)*0/1}, {0/1+sqrt(3)*0/1}) {};
\node[circle,scale=0.4,draw,fill=white] (c0vc0v01c1v1c2v012) at ({22/5+sqrt(3)*0/1}, {0/1+sqrt(3)*6/25}) {};
\node[circle,scale=0.4,draw,fill=white] (c0vc0v01c2v012) at ({10/3+sqrt(3)*0/1}, {0/1+sqrt(3)*2/5}) {};
\node[circle,scale=0.6,draw,fill=white] (c0vc0v02) at ({2/1+sqrt(3)*0/1}, {0/1+sqrt(3)*2/1}) {};
\node[circle,scale=0.4,draw,fill=lightgray] (c1vc0v02c1v012) at ({32/15+sqrt(3)*0/1}, {0/1+sqrt(3)*26/15}) {};
\node[circle,scale=0.4,draw,fill=black] (c2vc0v02c1v012c2v02) at ({49/25+sqrt(3)*0/1}, {0/1+sqrt(3)*37/25}) {};
\node[circle,scale=0.4,draw,fill=black] (c2vc0v02c1v012c2v2) at ({59/25+sqrt(3)*0/1}, {0/1+sqrt(3)*47/25}) {};
\node[circle,scale=0.4,draw,fill=lightgray] (c1vc0v02c1v012c2v02) at ({42/25+sqrt(3)*0/1}, {0/1+sqrt(3)*36/25}) {};
\node[circle,scale=0.4,draw,fill=black] (c2vc0v02c2v02) at ({5/3+sqrt(3)*0/1}, {0/1+sqrt(3)*5/3}) {};
\node[circle,scale=0.4,draw,fill=lightgray] (c1vc0v02c1v012c2v2) at ({62/25+sqrt(3)*0/1}, {0/1+sqrt(3)*56/25}) {};
\node[circle,scale=0.4,draw,fill=black] (c2vc0v02c2v2) at ({7/3+sqrt(3)*0/1}, {0/1+sqrt(3)*7/3}) {};
\node[circle,scale=0.4,draw,fill=white] (c0vc0v02c1v012) at ({34/15+sqrt(3)*0/1}, {0/1+sqrt(3)*22/15}) {};
\node[circle,scale=0.4,draw,fill=white] (c0vc0v02c1v012c2v02) at ({44/25+sqrt(3)*0/1}, {0/1+sqrt(3)*32/25}) {};
\node[circle,scale=0.6,draw,fill=black] (c2vc2v02) at ({1/1+sqrt(3)*0/1}, {0/1+sqrt(3)*1/1}) {};
\node[circle,scale=0.4,draw,fill=black] (c2vc1v012c2v02) at ({29/15+sqrt(3)*0/1}, {0/1+sqrt(3)*17/15}) {};
\node[circle,scale=0.4,draw,fill=lightgray] (c1vc1v012c2v02) at ({22/15+sqrt(3)*0/1}, {0/1+sqrt(3)*16/15}) {};
\node[circle,scale=0.4,draw,fill=white] (c0vc0v02c1v012c2v2) at ({64/25+sqrt(3)*0/1}, {0/1+sqrt(3)*52/25}) {};
\node[circle,scale=0.4,draw,fill=white] (c0vc0v02c2v02) at ({4/3+sqrt(3)*0/1}, {0/1+sqrt(3)*4/3}) {};
\node[circle,scale=0.4,draw,fill=white] (c0vc0v02c2v2) at ({8/3+sqrt(3)*0/1}, {0/1+sqrt(3)*8/3}) {};
\node[circle,scale=0.4,draw,fill=white] (c0vc0v0c1v01) at ({4/3+sqrt(3)*0/1}, {0/1+sqrt(3)*0/1}) {};
\node[circle,scale=0.4,draw,fill=white] (c0vc0v0c1v012) at ({8/5+sqrt(3)*0/1}, {0/1+sqrt(3)*4/5}) {};
\node[circle,scale=0.4,draw,fill=white] (c0vc0v0c1v012c2v012) at ({54/25+sqrt(3)*0/1}, {0/1+sqrt(3)*18/25}) {};
\node[circle,scale=0.4,draw,fill=white] (c0vc0v0c1v012c2v02) at ({34/25+sqrt(3)*0/1}, {0/1+sqrt(3)*22/25}) {};
\node[circle,scale=0.4,draw,fill=white] (c0vc0v0c1v01c2v012) at ({2/1+sqrt(3)*0/1}, {0/1+sqrt(3)*6/25}) {};
\node[circle,scale=0.4,draw,fill=white] (c0vc0v0c2v012) at ({2/1+sqrt(3)*0/1}, {0/1+sqrt(3)*2/5}) {};
\node[circle,scale=0.4,draw,fill=white] (c0vc0v0c2v02) at ({2/3+sqrt(3)*0/1}, {0/1+sqrt(3)*2/3}) {};
\draw (c0vc0v0) -- (c1vc0v0c1v01);
\draw (c0vc0v0) -- (c2vc0v0c1v01c2v012);
\draw (c1vc0v0c1v01) -- (c2vc0v0c1v01c2v012);
\draw (c0vc0v0) -- (c1vc0v0c1v012);
\draw (c0vc0v0) -- (c2vc0v0c1v012c2v012);
\draw (c1vc0v0c1v012) -- (c2vc0v0c1v012c2v012);
\draw (c0vc0v0) -- (c2vc0v0c1v012c2v02);
\draw (c1vc0v0c1v012) -- (c2vc0v0c1v012c2v02);
\draw (c0vc0v0) -- (c1vc0v0c1v012c2v012);
\draw (c1vc0v0c1v012c2v012) -- (c2vc0v0c1v012c2v012);
\draw (c0vc0v0) -- (c2vc0v0c2v012);
\draw (c1vc0v0c1v012c2v012) -- (c2vc0v0c2v012);
\draw (c0vc0v0) -- (c1vc0v0c1v012c2v02);
\draw (c1vc0v0c1v012c2v02) -- (c2vc0v0c1v012c2v02);
\draw (c0vc0v0) -- (c2vc0v0c2v02);
\draw (c1vc0v0c1v012c2v02) -- (c2vc0v0c2v02);
\draw (c0vc0v0) -- (c1vc0v0c1v01c2v012);
\draw (c1vc0v0c1v01c2v012) -- (c2vc0v0c1v01c2v012);
\draw (c1vc0v0c1v01c2v012) -- (c2vc0v0c2v012);
\draw (c0vc0v01) -- (c1vc0v01c1v01);
\draw (c0vc0v01) -- (c2vc0v01c1v01c2v012);
\draw (c1vc0v01c1v01) -- (c2vc0v01c1v01c2v012);
\draw (c0vc0v01) -- (c1vc0v01c1v01c2v012);
\draw (c1vc0v01c1v01c2v012) -- (c2vc0v01c1v01c2v012);
\draw (c0vc0v01) -- (c2vc0v01c2v012);
\draw (c1vc0v01c1v01c2v012) -- (c2vc0v01c2v012);
\draw (c0vc0v01) -- (c1vc0v01c1v1);
\draw (c0vc0v01) -- (c2vc0v01c1v1c2v012);
\draw (c1vc0v01c1v1) -- (c2vc0v01c1v1c2v012);
\draw (c0vc0v01) -- (c1vc0v01c1v1c2v012);
\draw (c1vc0v01c1v1c2v012) -- (c2vc0v01c1v1c2v012);
\draw (c1vc0v01c1v1c2v012) -- (c2vc0v01c2v012);
\draw (c0vc0v012) -- (c1vc0v012c1v012);
\draw (c0vc0v012) -- (c2vc0v012c1v012c2v012);
\draw (c1vc0v012c1v012) -- (c2vc0v012c1v012c2v012);
\draw (c0vc0v012) -- (c2vc0v012c1v012c2v2);
\draw (c1vc0v012c1v012) -- (c2vc0v012c1v012c2v2);
\draw (c0vc0v012) -- (c1vc0v012c1v012c2v012);
\draw (c1vc0v012c1v012c2v012) -- (c2vc0v012c1v012c2v012);
\draw (c0vc0v012) -- (c2vc0v012c2v012);
\draw (c1vc0v012c1v012c2v012) -- (c2vc0v012c2v012);
\draw (c0vc0v012) -- (c1vc0v012c1v012c2v2);
\draw (c1vc0v012c1v012c2v2) -- (c2vc0v012c1v012c2v2);
\draw (c0vc0v012) -- (c2vc0v012c2v2);
\draw (c1vc0v012c1v012c2v2) -- (c2vc0v012c2v2);
\draw (c0vc0v012) -- (c1vc0v012c1v1);
\draw (c0vc0v012) -- (c2vc0v012c1v1c2v012);
\draw (c1vc0v012c1v1) -- (c2vc0v012c1v1c2v012);
\draw (c0vc0v012) -- (c2vc0v012c1v1c2v12);
\draw (c1vc0v012c1v1) -- (c2vc0v012c1v1c2v12);
\draw (c0vc0v012) -- (c1vc0v012c1v12);
\draw (c0vc0v012) -- (c2vc0v012c1v12c2v12);
\draw (c1vc0v012c1v12) -- (c2vc0v012c1v12c2v12);
\draw (c0vc0v012) -- (c2vc0v012c1v12c2v2);
\draw (c1vc0v012c1v12) -- (c2vc0v012c1v12c2v2);
\draw (c0vc0v012) -- (c1vc0v012c1v12c2v12);
\draw (c1vc0v012c1v12c2v12) -- (c2vc0v012c1v12c2v12);
\draw (c0vc0v012) -- (c2vc0v012c2v12);
\draw (c1vc0v012c1v12c2v12) -- (c2vc0v012c2v12);
\draw (c0vc0v012) -- (c1vc0v012c1v12c2v2);
\draw[red,-latex] (c1vc0v012c1v12c2v2) -- (c2vc0v012c1v12c2v2);
\draw (c1vc0v012c1v12c2v2) -- (c2vc0v012c2v2);
\draw (c0vc0v012) -- (c1vc0v012c1v1c2v012);
\draw (c1vc0v012c1v1c2v012) -- (c2vc0v012c1v1c2v012);
\draw (c1vc0v012c1v1c2v012) -- (c2vc0v012c2v012);
\draw (c0vc0v012) -- (c1vc0v012c1v1c2v12);
\draw (c1vc0v012c1v1c2v12) -- (c2vc0v012c1v1c2v12);
\draw (c1vc0v012c1v1c2v12) -- (c2vc0v012c2v12);
\draw (c0vc0v012c1v012) -- (c1vc0v012c1v012);
\draw (c0vc0v012c1v012) -- (c2vc0v012c1v012c2v012);
\draw (c0vc0v012c1v012) -- (c2vc0v012c1v012c2v2);
\draw (c0vc0v012c1v012) -- (c1vc1v012);
\draw (c1vc1v012) -- (c2vc0v012c1v012c2v012);
\draw (c1vc1v012) -- (c2vc0v012c1v012c2v2);
\draw (c0vc0v012c1v012c2v012) -- (c1vc0v012c1v012c2v012);
\draw (c0vc0v012c1v012c2v012) -- (c2vc0v012c1v012c2v012);
\draw (c0vc0v012c1v012c2v012) -- (c2vc2v012);
\draw (c1vc0v012c1v012c2v012) -- (c2vc2v012);
\draw (c0vc0v012c1v012c2v012) -- (c1vc1v012);
\draw (c0vc0v012c1v012c2v012) -- (c2vc1v012c2v012);
\draw (c1vc1v012) -- (c2vc1v012c2v012);
\draw (c0vc0v012c1v012c2v012) -- (c1vc1v012c2v012);
\draw (c1vc1v012c2v012) -- (c2vc1v012c2v012);
\draw (c1vc1v012c2v012) -- (c2vc2v012);
\draw[red,semithick] (c0vc0v012c1v012c2v2) -- (c1vc0v012c1v012c2v2);
\draw (c0vc0v012c1v012c2v2) -- (c2vc0v012c1v012c2v2);
\draw (c0vc0v012c1v012c2v2) -- (c2vc2v2);
\draw (c1vc0v012c1v012c2v2) -- (c2vc2v2);
\draw (c0vc0v012c1v012c2v2) -- (c1vc1v012);
\draw (c0vc0v012c1v012c2v2) -- (c2vc1v012c2v2);
\draw (c1vc1v012) -- (c2vc1v012c2v2);
\draw[red,semithick] (c0vc0v012c1v012c2v2) -- (c1vc1v012c2v2);
\draw (c1vc1v012c2v2) -- (c2vc1v012c2v2);
\draw (c1vc1v012c2v2) -- (c2vc2v2);
\draw (c0vc0v012c1v1) -- (c1vc0v012c1v1);
\draw (c0vc0v012c1v1) -- (c2vc0v012c1v1c2v012);
\draw (c0vc0v012c1v1) -- (c2vc0v012c1v1c2v12);
\draw (c0vc0v012c1v1) -- (c1vc1v1);
\draw (c1vc1v1) -- (c2vc0v012c1v1c2v012);
\draw (c1vc1v1) -- (c2vc0v012c1v1c2v12);
\draw (c0vc0v012c1v12) -- (c1vc0v012c1v12);
\draw (c0vc0v012c1v12) -- (c2vc0v012c1v12c2v12);
\draw (c0vc0v012c1v12) -- (c2vc0v012c1v12c2v2);
\draw (c0vc0v012c1v12) -- (c1vc1v12);
\draw (c1vc1v12) -- (c2vc0v012c1v12c2v12);
\draw (c1vc1v12) -- (c2vc0v012c1v12c2v2);
\draw (c0vc0v012c1v12c2v12) -- (c1vc0v012c1v12c2v12);
\draw (c0vc0v012c1v12c2v12) -- (c2vc0v012c1v12c2v12);
\draw (c0vc0v012c1v12c2v12) -- (c2vc2v12);
\draw (c1vc0v012c1v12c2v12) -- (c2vc2v12);
\draw (c0vc0v012c1v12c2v12) -- (c1vc1v12);
\draw (c0vc0v012c1v12c2v12) -- (c2vc1v12c2v12);
\draw (c1vc1v12) -- (c2vc1v12c2v12);
\draw (c0vc0v012c1v12c2v12) -- (c1vc1v12c2v12);
\draw (c1vc1v12c2v12) -- (c2vc1v12c2v12);
\draw (c1vc1v12c2v12) -- (c2vc2v12);
\draw[red,semithick] (c0vc0v012c1v12c2v2) -- (c1vc0v012c1v12c2v2);
\draw[red,-latex] (c0vc0v012c1v12c2v2) -- (c2vc0v012c1v12c2v2);
\draw (c0vc0v012c1v12c2v2) -- (c2vc2v2);
\draw (c1vc0v012c1v12c2v2) -- (c2vc2v2);
\draw[red,-latex] (c0vc0v012c1v12c2v2) -- (c1vc1v12);
\draw[red,-latex] (c0vc0v012c1v12c2v2) -- (c2vc1v12c2v2);
\draw (c1vc1v12) -- (c2vc1v12c2v2);
\draw[red,-latex] (c0vc0v012c1v12c2v2) -- (c1vc1v12c2v2);
\draw (c1vc1v12c2v2) -- (c2vc1v12c2v2);
\draw (c1vc1v12c2v2) -- (c2vc2v2);
\draw (c0vc0v012c1v1c2v012) -- (c1vc0v012c1v1c2v012);
\draw (c0vc0v012c1v1c2v012) -- (c2vc0v012c1v1c2v012);
\draw (c0vc0v012c1v1c2v012) -- (c2vc2v012);
\draw (c1vc0v012c1v1c2v012) -- (c2vc2v012);
\draw (c0vc0v012c1v1c2v012) -- (c1vc1v1);
\draw (c0vc0v012c1v1c2v012) -- (c2vc1v1c2v012);
\draw (c1vc1v1) -- (c2vc1v1c2v012);
\draw (c0vc0v012c1v1c2v012) -- (c1vc1v1c2v012);
\draw (c1vc1v1c2v012) -- (c2vc1v1c2v012);
\draw (c1vc1v1c2v012) -- (c2vc2v012);
\draw (c0vc0v012c1v1c2v12) -- (c1vc0v012c1v1c2v12);
\draw (c0vc0v012c1v1c2v12) -- (c2vc0v012c1v1c2v12);
\draw (c0vc0v012c1v1c2v12) -- (c2vc2v12);
\draw (c1vc0v012c1v1c2v12) -- (c2vc2v12);
\draw (c0vc0v012c1v1c2v12) -- (c1vc1v1);
\draw (c0vc0v012c1v1c2v12) -- (c2vc1v1c2v12);
\draw (c1vc1v1) -- (c2vc1v1c2v12);
\draw (c0vc0v012c1v1c2v12) -- (c1vc1v1c2v12);
\draw (c1vc1v1c2v12) -- (c2vc1v1c2v12);
\draw (c1vc1v1c2v12) -- (c2vc2v12);
\draw (c0vc0v012c2v012) -- (c1vc0v012c1v012c2v012);
\draw (c0vc0v012c2v012) -- (c2vc0v012c2v012);
\draw (c0vc0v012c2v012) -- (c2vc2v012);
\draw (c0vc0v012c2v012) -- (c1vc0v012c1v1c2v012);
\draw (c0vc0v012c2v12) -- (c1vc0v012c1v12c2v12);
\draw (c0vc0v012c2v12) -- (c2vc0v012c2v12);
\draw (c0vc0v012c2v12) -- (c2vc2v12);
\draw (c0vc0v012c2v12) -- (c1vc0v012c1v1c2v12);
\draw[red,semithick] (c0vc0v012c2v2) -- (c1vc0v012c1v012c2v2);
\draw (c0vc0v012c2v2) -- (c2vc0v012c2v2);
\draw (c0vc0v012c2v2) -- (c2vc2v2);
\draw[red,semithick] (c0vc0v012c2v2) -- (c1vc0v012c1v12c2v2);
\draw (c0vc0v01c1v01) -- (c1vc0v01c1v01);
\draw (c0vc0v01c1v01) -- (c2vc0v01c1v01c2v012);
\draw (c0vc0v01c1v01) -- (c1vc1v01);
\draw (c1vc1v01) -- (c2vc0v01c1v01c2v012);
\draw (c0vc0v01c1v01c2v012) -- (c1vc0v01c1v01c2v012);
\draw (c0vc0v01c1v01c2v012) -- (c2vc0v01c1v01c2v012);
\draw (c0vc0v01c1v01c2v012) -- (c2vc2v012);
\draw (c1vc0v01c1v01c2v012) -- (c2vc2v012);
\draw (c0vc0v01c1v01c2v012) -- (c1vc1v01);
\draw (c0vc0v01c1v01c2v012) -- (c2vc1v01c2v012);
\draw (c1vc1v01) -- (c2vc1v01c2v012);
\draw (c0vc0v01c1v01c2v012) -- (c1vc1v01c2v012);
\draw (c1vc1v01c2v012) -- (c2vc1v01c2v012);
\draw (c1vc1v01c2v012) -- (c2vc2v012);
\draw (c0vc0v01c1v1) -- (c1vc0v01c1v1);
\draw (c0vc0v01c1v1) -- (c2vc0v01c1v1c2v012);
\draw (c0vc0v01c1v1) -- (c1vc1v1);
\draw (c1vc1v1) -- (c2vc0v01c1v1c2v012);
\draw (c0vc0v01c1v1c2v012) -- (c1vc0v01c1v1c2v012);
\draw (c0vc0v01c1v1c2v012) -- (c2vc0v01c1v1c2v012);
\draw (c0vc0v01c1v1c2v012) -- (c2vc2v012);
\draw (c1vc0v01c1v1c2v012) -- (c2vc2v012);
\draw (c0vc0v01c1v1c2v012) -- (c1vc1v1);
\draw (c0vc0v01c1v1c2v012) -- (c2vc1v1c2v012);
\draw (c0vc0v01c1v1c2v012) -- (c1vc1v1c2v012);
\draw (c0vc0v01c2v012) -- (c1vc0v01c1v01c2v012);
\draw (c0vc0v01c2v012) -- (c2vc0v01c2v012);
\draw (c0vc0v01c2v012) -- (c2vc2v012);
\draw (c0vc0v01c2v012) -- (c1vc0v01c1v1c2v012);
\draw (c0vc0v02) -- (c1vc0v02c1v012);
\draw (c0vc0v02) -- (c2vc0v02c1v012c2v02);
\draw (c1vc0v02c1v012) -- (c2vc0v02c1v012c2v02);
\draw (c0vc0v02) -- (c2vc0v02c1v012c2v2);
\draw (c1vc0v02c1v012) -- (c2vc0v02c1v012c2v2);
\draw (c0vc0v02) -- (c1vc0v02c1v012c2v02);
\draw (c1vc0v02c1v012c2v02) -- (c2vc0v02c1v012c2v02);
\draw (c0vc0v02) -- (c2vc0v02c2v02);
\draw (c1vc0v02c1v012c2v02) -- (c2vc0v02c2v02);
\draw[red,-latex] (c0vc0v02) -- (c1vc0v02c1v012c2v2);
\draw[red,latex-] (c1vc0v02c1v012c2v2) -- (c2vc0v02c1v012c2v2);
\draw (c0vc0v02) -- (c2vc0v02c2v2);
\draw[red,latex-] (c1vc0v02c1v012c2v2) -- (c2vc0v02c2v2);
\draw (c0vc0v02c1v012) -- (c1vc0v02c1v012);
\draw (c0vc0v02c1v012) -- (c2vc0v02c1v012c2v02);
\draw (c0vc0v02c1v012) -- (c2vc0v02c1v012c2v2);
\draw (c0vc0v02c1v012) -- (c1vc1v012);
\draw (c1vc1v012) -- (c2vc0v02c1v012c2v02);
\draw (c1vc1v012) -- (c2vc0v02c1v012c2v2);
\draw (c0vc0v02c1v012c2v02) -- (c1vc0v02c1v012c2v02);
\draw (c0vc0v02c1v012c2v02) -- (c2vc0v02c1v012c2v02);
\draw (c0vc0v02c1v012c2v02) -- (c2vc2v02);
\draw (c1vc0v02c1v012c2v02) -- (c2vc2v02);
\draw (c0vc0v02c1v012c2v02) -- (c1vc1v012);
\draw (c0vc0v02c1v012c2v02) -- (c2vc1v012c2v02);
\draw (c1vc1v012) -- (c2vc1v012c2v02);
\draw (c0vc0v02c1v012c2v02) -- (c1vc1v012c2v02);
\draw (c1vc1v012c2v02) -- (c2vc1v012c2v02);
\draw (c1vc1v012c2v02) -- (c2vc2v02);
\draw[red,semithick] (c0vc0v02c1v012c2v2) -- (c1vc0v02c1v012c2v2);
\draw[red,latex-] (c0vc0v02c1v012c2v2) -- (c2vc0v02c1v012c2v2);
\draw (c0vc0v02c1v012c2v2) -- (c2vc2v2);
\draw (c1vc0v02c1v012c2v2) -- (c2vc2v2);
\draw (c0vc0v02c1v012c2v2) -- (c1vc1v012);
\draw (c0vc0v02c1v012c2v2) -- (c2vc1v012c2v2);
\draw[red,semithick] (c0vc0v02c1v012c2v2) -- (c1vc1v012c2v2);
\draw (c0vc0v02c2v02) -- (c1vc0v02c1v012c2v02);
\draw (c0vc0v02c2v02) -- (c2vc0v02c2v02);
\draw (c0vc0v02c2v02) -- (c2vc2v02);
\draw[red,-latex] (c0vc0v02c2v2) -- (c1vc0v02c1v012c2v2);
\draw (c0vc0v02c2v2) -- (c2vc0v02c2v2);
\draw (c0vc0v02c2v2) -- (c2vc2v2);
\draw (c0vc0v0c1v01) -- (c1vc0v0c1v01);
\draw (c0vc0v0c1v01) -- (c2vc0v0c1v01c2v012);
\draw (c0vc0v0c1v01) -- (c1vc1v01);
\draw (c1vc1v01) -- (c2vc0v0c1v01c2v012);
\draw (c0vc0v0c1v012) -- (c1vc0v0c1v012);
\draw (c0vc0v0c1v012) -- (c2vc0v0c1v012c2v012);
\draw (c0vc0v0c1v012) -- (c2vc0v0c1v012c2v02);
\draw (c0vc0v0c1v012) -- (c1vc1v012);
\draw (c1vc1v012) -- (c2vc0v0c1v012c2v012);
\draw (c1vc1v012) -- (c2vc0v0c1v012c2v02);
\draw (c0vc0v0c1v012c2v012) -- (c1vc0v0c1v012c2v012);
\draw (c0vc0v0c1v012c2v012) -- (c2vc0v0c1v012c2v012);
\draw (c0vc0v0c1v012c2v012) -- (c2vc2v012);
\draw (c1vc0v0c1v012c2v012) -- (c2vc2v012);
\draw (c0vc0v0c1v012c2v012) -- (c1vc1v012);
\draw (c0vc0v0c1v012c2v012) -- (c2vc1v012c2v012);
\draw (c0vc0v0c1v012c2v012) -- (c1vc1v012c2v012);
\draw (c0vc0v0c1v012c2v02) -- (c1vc0v0c1v012c2v02);
\draw (c0vc0v0c1v012c2v02) -- (c2vc0v0c1v012c2v02);
\draw (c0vc0v0c1v012c2v02) -- (c2vc2v02);
\draw (c1vc0v0c1v012c2v02) -- (c2vc2v02);
\draw (c0vc0v0c1v012c2v02) -- (c1vc1v012);
\draw (c0vc0v0c1v012c2v02) -- (c2vc1v012c2v02);
\draw (c0vc0v0c1v012c2v02) -- (c1vc1v012c2v02);
\draw (c0vc0v0c1v01c2v012) -- (c1vc0v0c1v01c2v012);
\draw (c0vc0v0c1v01c2v012) -- (c2vc0v0c1v01c2v012);
\draw (c0vc0v0c1v01c2v012) -- (c2vc2v012);
\draw (c1vc0v0c1v01c2v012) -- (c2vc2v012);
\draw (c0vc0v0c1v01c2v012) -- (c1vc1v01);
\draw (c0vc0v0c1v01c2v012) -- (c2vc1v01c2v012);
\draw (c0vc0v0c1v01c2v012) -- (c1vc1v01c2v012);
\draw (c0vc0v0c2v012) -- (c1vc0v0c1v012c2v012);
\draw (c0vc0v0c2v012) -- (c2vc0v0c2v012);
\draw (c0vc0v0c2v012) -- (c2vc2v012);
\draw (c0vc0v0c2v012) -- (c1vc0v0c1v01c2v012);
\draw (c0vc0v0c2v02) -- (c1vc0v0c1v012c2v02);
\draw (c0vc0v0c2v02) -- (c2vc0v0c2v02);
\draw (c0vc0v0c2v02) -- (c2vc2v02);
\draw[red,very thick] (c0vc0v0) -- (c0vc0v01) -- (c2vc2v012) -- (c0vc0v012) -- (c1vc1v012) -- (c0vc0v02) -- (c0vc0v0);
\end{tikzpicture}

%% file: 50_LO_construction.tex

\section{Constructing a logical obstruction}  \label{sec_lo_construction}

We define the maximum nesting depth of modal operators of
an epistemic formula $\varphi$ of either $\mathcal{L}_K$ or $\mathcal{L}_D$,
written $\deg(\varphi)$, as below by induction on $\varphi$.
\begin{alignat*}{2}
  &\deg(p) &&= 0 \\
  &\deg(\lnot \varphi) &&= \deg(\varphi) \\
  &\deg(\varphi_1 \land \varphi_2) &&= \max(\deg(\varphi_1), \deg(\varphi_2)) \\
  &\deg(\mathrm{K} \varphi) &&= \deg(\varphi) + 1 \\
  &\deg(\mathrm{D} \varphi) &&= \deg(\varphi) + 1 \\
\end{alignat*}

\subsection{$n$-simulation}

An $n$-simulation is a finite approximation of a simulation.
The following is a formal definition of $n$-simulation.

\begin{definition}
  Let $R \subseteq \rel$ be a binary relation.
  We define an \emph{$n$-K-simulation} as below
  by induction on $n$.
  \begin{itemize}
    \item $R$ is a $0$-K-simulation
      if $R$ satisfies the condition {\atom}.
    \item $R$ is an $(n + 1)$-K-simulation
      if $R$ satisfies the condition {\atom} and
      for all $a \in \Pi$, $X, Y \in \facet{M}$ and $X' \in \facet{M'}$,
      $X \mathrel{R} X'$ and $X \sim_a Y$ implies that
      there exists an $n$-K-simulation $R'$ and
      $Y' \in \facet{M'}$ such that
      $Y \mathrel{R'} Y'$ and $X' \sim_a Y'$.
  \end{itemize}
  Similarly, we define an \emph{$n$-D-simulation} as below
  by induction on $n$.
  \begin{itemize}
    \item $R$ is a $0$-D-simulation
      if $R$ satisfies the condition {\atom}.
    \item $R$ is an $(n + 1)$-D-simulation
      if for all $X, Y \in \facet{M}$ and $X' \in \facet{M'}$,
      $X \mathrel{R} X'$ implies that
      there exists an $n$-D-simulation $R'$ and
      $Y' \in \facet{M'}$ such that
      $Y \mathrel{R} Y'$ and $\chi(X \cap Y) \subseteq \chi'(X' \cap Y')$.
  \end{itemize}

\end{definition}

Notice that a K-simulation $S$ is an $n$-K-simulation for all $n \in \N$
and a D-simulation $S$ is an $n$-D-simulation for all $n \in \N$.

In the sequel, we write $\Box$ to mean K or D.
For instance,
``a $\Box$-simulation $S$ is an $n$-$\Box$-simulation for all $n \in \N$''
means both
``a K-simulation $S$ is an $n$-K-simulation for all $n \in \N$'' and
``a D-simulation $S$ is an $n$-D-simulation for all $n \in \N$.''

\begin{remark}  \label{rem_n-D-simulation_satisfy_atom}
  Although not explicitly mentioned,
  all $n$-D-simulations satisfy {\atom}
  even if $n \geq 1$.
  Indeed, if facets $X$ and $X'$ satisfy $X \mathrel{S} X'$
  for some $(n + 1)$-D-simulation,
  by taking $Y = X$ in the above definition,
  there exists an $n$-D-simulation $S'$ and
  $Y' \in \facet{M'}$ such that $X \mathrel{S'} Y'$ and
  $\chi(X) = \Pi \subseteq \chi' (X' \cap Y')$.
  Since the latter condition implies $Y' = X'$,
  we have $X \mathrel{S'} X'$.
  This inductively proves that every $n$-simulation satisfies {\atom}.
\end{remark}

\begin{proposition} \label{prop_knowledge_gain_via_n-simulation}
  Suppose $S \subseteq \facet{M} \times \facet{M'}$
  be an $n$-$\Box$-simulation.
  Then, for any facets $(X, X') \in S$
  and positive formula $\varphi \in \pfb$ with
  $\deg(\varphi) \leq n$,
  $M', X' \models \varphi$
  implies $M, X \models \varphi$.

\end{proposition}
\begin{proof}
  We proceed by induction on the lexicographical order
  of the pairs $(n, \varphi)$.

  For the base case $n = 0$, suppose $X \mathrel{S} X'$ and
  $\varphi$ be a positive formula such that $\deg(\varphi) = 0$,
  which means $\varphi$ has no modal operator.
  \begin{itemize}
    \item For the case of atomic propositions,
      let $\varphi = p \in \At$.
      Since $S$ satisfies the condition {\atom},
      $M, X \models p \iff p \in l(X) \iff p \in l'(X') \iff M', X' \models p$.
      We are done.
    \item The case $\varphi = \lnot p$ is similarly proved.
    \item The cases of conjunction and disjunction are
      easily shown by induction hypothesis.
  \end{itemize}

  For the induction case, suppose $X \mathrel{S} X'$ and
  $\varphi$ be a positive formula such that $\deg(\varphi) = n + 1$.
  \begin{itemize}
    \item The cases of atomic propositions and
      negated atomic propositions are similarly proved
      since $S$ satisfies {\atom}.
      (See Remark~\ref{rem_n-D-simulation_satisfy_atom}.)
    \item The cases of conjunction and disjunction are
      easily shown by induction hypothesis.
    \item For the case of modal operator,
      suppose $\Box = \mathrm{D}$ and $\varphi = \mathrm{D}_A \psi$.
      Then $\deg(\psi) \leq n$.
      Assuming $M', X' \models \mathrm{D}_A \psi$,
      we show that $M, Y \models \psi$ holds for any $Y \in \facet{M}$
      such that $A \subseteq \chi(X \cap Y)$.
      By definition, there exists an $n$-simulation $S'$
      and a facet $Y' \in \facet{M'}$
      such that $Y \mathrel{S'} Y'$ and
      $A \subseteq \chi(X \cap Y) \subseteq \chi'(X' \cap Y')$.
      This implies $M', Y' \models \psi$
      and hence $M, Y \models \psi$ by induction hypothesis.
      Therefore, $M, X \models \mathrm{D}_A \psi$.

      The other case of $\Box = \mathrm{K}$ is similarly proved.
      \qedhere
  \end{itemize}
\end{proof}

\begin{corollary} \label{cor_knowledge_gain_via_n}
  If there exists a total $n$-$\Box$-simulation of $M$ by $M'$,
  for any positive formula $\varphi \in \pfb$ with $\deg(\varphi) \leq n$,
  $M' \models \varphi$ implies $M \models \varphi$.
\end{corollary}

\begin{definition}  \label{def_n-simulation}
  We define a map
  $f^{\Box} \colon \mathcal{P}(\rel) \to \mathcal{P}(\rel)$
  over relations as follows.
  \begin{itemize}
    \item Given a relation $R \subseteq \rel$,
      $X \f{\mathrm{K}}{R} X'$ if
      for all $a \in \Pi$, $Y \in \facet{M}$ with $X \sim_a Y$,
      there exists $Y' \in \facet{M'}$
      such that $Y \mathrel{R} Y'$ and
      $X' \sim_a Y'$.
    \item Given a relation $R \subseteq \rel$,
      $X \f{\mathrm{D}}{R} X'$ if
      for all $Y \in \facet{M}$
      there exists $Y' \in \facet{M'}$
      such that $Y \mathrel{R} Y'$ and
      $\chi(X \cap Y) \subseteq \chi'(X' \cap Y')$.
  \end{itemize}

  We also define a sequence of binary relations $(\SB{n})_{n \in \N}$,
  where $\SB{n} \subseteq \rel$, as follows.
  \begin{itemize}
    \item $\SB{0} = \{(X, X') \mid l(X) = l'(X)\}$.
    \item $\SK{n+1} = \SK{0} \cap f^{\mathrm{K}}(\SK{n})$
    \item $\SD{n+1} = f^{\mathrm{D}}(\SD{n})$.
  \end{itemize}

\end{definition}

\begin{proposition} \label{prop_Sn_is_maximum}
  $\S{n}$ is the maximum $n$-$\Box$-simulation with respect to inclusion.

\end{proposition}
\begin{proof}
  We proceed by induction on $n$.

  The base case $n = 0$ is trivially follows from definition.

  For the induction case, suppose $\SK{n}$ is
  the maximum $n$-K-simulation.
  Then, $\SK{n + 1} = \SK{0} \cap f^{\mathrm{K}}(\SK{n})$ is
  an $(n + 1)$-simulation by definition.
  Let $R$ be any $(n + 1)$-K-simulation and
  suppose $X \mathrel{R} X'$.
  Then $l(X) = l'(X')$ and
  for all $Y \in \facet{M}$ with $X \sim_a Y$,
  there exists an $n$-K-simulation $R'$ and $Y' \in \facet{M'}$
  such that $X' \sim_a Y'$ and $Y \mathrel{R'} Y'$.
  Since $\SK{n}$ is the maximum $n$-K-simulation,
  we have $R' \subseteq \SK{n}$ and
  $R \subseteq \SK{0} \cap f^{\mathrm{K}}(\SK{n}) = \SK{n + 1}$.
  Therefore, $\SK{n + 1}$ is the maximum $(n + 1)$-$\Box$-simulation.
  The case of $\Box = \mathrm{D}$ is similar.
\end{proof}

\begin{corollary}
  $\S{n+1} \subseteq \S{n}$ for all $n$.
\end{corollary}
\begin{proof}
  It follows from that any $(n + 1)$-$\Box$-simulation
  is also $n$-$\Box$-simulation by definition.
\end{proof}

\subsection{A method for determining non-existence of logical obstruction}

\begin{proposition} \label{prop_equivalence_simulation_fixpoint}
  For any binary relation $R \subseteq \rel$,
  $R$ is a $\Box$-simulation if and only if
  $R \subseteq \S{0} \cap f^{\Box}(R)$.

\end{proposition}
\begin{proof}
  By definition of the condition {\atom},
  $R$ satisfies {\atom} if and only if $R \subseteq \S{0}$.
  It suffices to show that
  $R$ satisfies {\bforth}
  if and only if $R \subseteq \f{\Box}{R}$.

  Suppose $R$ satisfies {\kforth}.
  Let $X \in \facet{M}$ and $X' \in \facet{M'}$
  be facets with $X \mathrel{R} Y$.
  Applying the condition {\kforth}, we obtain that
  for any $Y \in \facet{M}$ with $X \sim_a Y$,
  there exists $Y' \in \facet{M'}$
  such that $Y \mathrel{R} Y'$ and $X' \sim_a Y'$.
  This implies $X \f{\mathrm{K}}{R} Y$. Hence $R \subseteq \f{\mathrm{K}}{R}$.
  The converse is similar.

  The case for $\Box = \mathrm{D}$ is similarly proved.
\end{proof}

\begin{theorem} \label{th_maximam_simulation}
  Suppose $M$ and $M'$ are finite.
  Then, there exists $n \in \N$ such that
  $\S{n}$ is the maximum $\Box$-simulation.
  Moreover, the maximum $\Box$-simulation is total
  if $\S{n}$ is total for all $n \in \N$.

\end{theorem}
\begin{proof}
  Because the powerset of $\rel$ is finite, the descending chain
  $\S{0} \supseteq \S{1} \supseteq \cdots \supseteq \S{n} \supseteq \cdots$
  stabilizes for a sufficiently large $n$.
  Then, $\S{n}$ is a simulation by Proposition~\ref{prop_equivalence_simulation_fixpoint}.
  Let $R \subseteq \rel$ be any simulation.
  Again by Proposition~\ref{prop_equivalence_simulation_fixpoint},
  $R \subseteq \S{0}$ and $R \subseteq f^{\Box}(R)$.
  Because $f^{\Box}$ is monotone with respect to inclusion by definition,
  we obtain $R \subseteq f^{\Box}\circ \dots \circ f^{\Box}(R)
    \subseteq f^{\Box}\circ \dots \circ f^{\Box}(\S{0}) = \S{n}$.
\end{proof}

It is worth observing here that
the finiteness assumption of the theorem above cannot be eliminated.

\begin{proposition}
  There exists an action model $\mathcal{P}$
  and a task $\mathcal{T}$ such that
  neither logical obstruction to the solvability of $T$ by $P$
  nor total $\Box$-simulation exists.

\end{proposition}
\begin{proof}
  We consider the case $\Pi = V^{in} = \{0,1\}$.
  Let us define a protocol $P = (V_P, S_P, \chi_P, \pre_P)$ as follows:
  \begin{itemize}
    \item $V_P = \Pi$,
    \item $S_P = \mathcal{P}(V_P)$,
    \item $\chi_P(a) = a$ and
    \item $\pre_P(X) = \top$,
  \end{itemize}
  where $\top \in \mathcal{L}_K$ is defined by
  $p \lor \lnot p$ for some $p \in \At$.
  We also define a task $T = (V_T, S_T, \chi_T, \pre_T)$ as follows:
  \begin{itemize}
    \item $V^{out} = \N$,
    \item $V_T = \Pi \times V^{out}$,
    \item $\facet{T} = \{\{(a, decide(a)) \mid a \in \Pi\}
      \mid decide \in [\Pi, V^{out}], 0 \leq decide(0) - decide(1) \leq 1\}$,
    \item $\chi_T(a,d) = a$ and
    \item $\pre_P(\{(a, decide(a)) \mid a \in \Pi\}) =
      \ip_0^{decide(0) \bmod 2} \land \ip_1^{decide(1) \bmod 2}$.
  \end{itemize}
  This is a well-defined task because there is an injection
  $\{(a, decide(a)) \mid a \in \Pi\} \mapsto decide$.
  Let $\mathcal{I}$ be an input model.
  Then, the product update models $\mathcal{I}[P]$ and $\mathcal{I}[T]$
  are defined as illustrated in Figure~\ref{fig_counterexample},
  where we write $(a, i_a)$ to mean a vertex
  $((a, i_a), a)$ of $\mathcal{I}[P]$ and
  $(a, i_a, d_a)$ to mean a vertex
  $((a, i_a), (a, d_a))$ of $\mathcal{I}[T]$.

  For all $i,j \in \Pi$, let $X_{i,j} = \{(0,i), (1,j)\}$
  be a facet of $\mathcal{I}[P]$.
  Then, an easy induction shows:
  \[
    X_{i,j} \mathrel{\SB{n}} X'
    \iff
    \text{there exists $d_0, d_1 \in \N$ such that
    $d_0 + d_1 \geq n$ and $X' = \{\{(0,i,d_0), (1,j,d_1)\}\}$.}
  \]
  Hence, $\SB{n}$ is total for all $n$ and
  $\bigcap_{n \in \N} \SB{n} = \emptyset$.

  Non-existence of logical obstruction follows from the totality of $\SB{n}$'s
  and Corollary~\ref{cor_knowledge_gain_via_n}.

  Let $R$ be any $\Box$-simulation.
  Then we have $R \subseteq \bigcap_{n \in \N} \SB{n} = \emptyset$
  since any $\Box$-simulation is an $n$-$\Box$-simulation and
  by Proposition~\ref{prop_Sn_is_maximum}, $\SB{n}$ is
  the maximum $n$-$\Box$-simulation for all $n$.
  Thus, $R$ cannot be a total $\Box$-simulation.
\end{proof}

\begin{figure}
  \centering
  \input{figures/50_counterexample_for_finite}
\end{figure}

\subsection{A method for constructing a logical obstruction}

  In the following, we assume simplicial models $M, M'$ are finite.

\begin{definition}

  For any facet $X \in \facet{M}$,
  we define positive formulas $\PhiB(n, X)$ as follows
  by induction on $n$:
  \begin{align*}
    \PhiB(0, X)
      & = \bigvee_{p \in l(X)} \lnot p
        \lor \bigvee_{p \in \At \setminus l(X)} p \\
    \PhiK(n + 1, X)
      & = \PhiK(0, X) \lor
        \bigvee_{a \in \Pi}
        \bigvee_{\substack{Y \in \facet{M} \\ X \sim_a Y}}
        \mathrm{K}_a \PhiK(n, Y) \\
    \PhiD(n + 1, X)
      & = \bigvee_{Y \in \facet{M}}
      \mathrm{D}_{\chi(X \cap Y)} \PhiD(n, Y)
  \end{align*}

\end{definition}

The following theorem implies that
each formula $\PhiB(n, X)$ characterizes facets $X'$ such that
$X$ and $X'$ is not $n$-$\Box$-similar.
It is known that there exists similar characteristic formulas
for $n$-bisimulation~\cite{DBLP:books/el/07/BBW2007}.

\begin{theorem} \label{theorem_universal_logical_obstruction}
  For all $n \in \N$, $X \in \facet{M}$ and $X' \in \facet{M'}$,
  \[
    X \mathrel{\S{n}} X' \iff M', X' \not\models \Phi^{\Box}_M(n, X)
  \]
\end{theorem}
\begin{proof}
  We proceed by induction on $n$.

  For the base case $n = 0$,
  we have $X \mathrel{\S{0}} X' \iff l(X) = l'(X')$ and
  $M', X' \models \PhiB(0, X) \iff l(X) \neq l'(X')$
  by definition of $\PhiB$.
  Thus, $X \mathrel{\S{0}} X' \iff M', X' \not\models \PhiB(0, X)$.

  For the induction case, suppose the theorem holds for $n$, then
  \begin{align*}
      X \mathrel{\SK{n+1}} X'
    \iff
      & \text{$X \mathrel{\SK{0}} X'$ and} \\
      & \text{for all $a \in \Pi$ and $Y \in \facet{M}$ with $X \sim_a Y$,
        there exists $Y' \in \facet{M'}$} \\
      & ~~\text{such that $Y \mathrel{\SK{n}} Y'$ and $X' \sim_a Y'$.} \\
    \overset{\text{I.H.}}{\iff}
      & \text{$M', X' \not\models \PhiK(0, X)$ and} \\
      & \text{for all $a \in \Pi$ and $Y \in \facet{M}$ with $X \sim_a Y$,
        there exists $Y' \in \facet{M'}$} \\
      & ~~\text{such that $M', Y' \not\models
        \PhiK(n, Y)$ and $X' \sim_a Y'$.} \\
    \iff
      & \text{$M', X' \not\models \PhiK(0, X)$ and} \\
      & \text{for all $a \in \Pi$ and $Y \in \facet{M}$ with $X \sim_a Y$,
        $M', X' \not\models \mathrm{K}_a \PhiK(n, Y)$.} \\
    \iff
      & \text{$M', X' \not\models \PhiK(0, X)$ and
        $M', X' \not\models
          \bigvee_{a \in \Pi}
          \bigvee_{\substack{Y \in \facet{M} \\ X \sim_a Y}}
          \mathrm{K}_a \PhiK(n, Y)$.}\\
    \iff
      & M', X' \not\models \PhiK(0, X) \lor
        \bigvee_{a \in \Pi} \bigvee_{\substack{Y \in \facet{M} \\ X \sim_a Y}}
        \mathrm{K}_a \PhiK(n, Y)
  \end{align*}
  \begin{align*}
      X \mathrel{\SD{n+1}} X'
    \iff
      & \text{for all $Y \in \facet{M}$, there exists $Y' \in \facet{M'}$} \\
      & ~~\text{such that $Y \mathrel{\SD{n}} Y'$ and
        $\chi(X \cap Y) \subseteq \chi'(X' \cap Y')$.} \\
    \overset{\text{I.H.}}{\iff}
      & \text{for all $Y \in \facet{M}$, there exists $Y' \in \facet{M'}$} \\
      & ~~\text{such that $M', Y' \not\models \PhiD(n, Y)$ and
        $\chi(X \cap Y) \subseteq \chi'(X' \cap Y')$.} \\
    \iff
      & \text{for all $Y \in \facet{M}$,
        $M', X' \not\models \mathrm{D}_{\chi(X \cap Y)} \PhiD(n, Y)$.} \\
    \iff
      & M', X' \not\models
        \bigvee_{Y \in \facet{M}} \mathrm{D}_{\chi(X \cap Y)} \PhiD(n, Y).
  \end{align*}
  Hence, we get $X \mathrel{\S{n+1}} X' \iff M', X' \not\models \PhiB(n+1, X)$.
\end{proof}

\begin{corollary} \label{cor_self_not_satisfy_ulo}
  Let $X \in \facet{M}$ be a facet.
  Then, $M , X \not\models \PhiB(n, X)$ for any $n$.

\end{corollary}
\begin{proof}
  Suppose $M' = M$.
  Then, by Theorem~\ref{theorem_universal_logical_obstruction},
  it suffices to show that $X \mathrel{\S{n}} X$
  holds for any $n \in \N$.
  This is easily proved by induction on $n$.
\end{proof}

\begin{corollary}
  If there is no total $\Box$-simulation,
  there exists a positive formula $\varphi \in \pfb$
  such that $M \not\models \varphi$ and $M' \models \varphi$.

\end{corollary}
\begin{proof}
  By Theorem~\ref{th_maximam_simulation},
  there exists $n \in \N$ such that $\S{n}$ is a $\Box$-simulation.
  Since it is assumed that any simulation is not total,
  there exists $X \in \facet{M}$ such that
  for all $X' \in \facet{M'}$,
  $X \mathrel{\S{n}} X'$ does not hold.
  This implies $M' \models \PhiB(n, X)$
  by Theorem~\ref{theorem_universal_logical_obstruction}.
  On the other hand, $M \not\models \PhiB(n, X)$
  by Corollary~\ref{cor_self_not_satisfy_ulo}.
\end{proof}

As a result, we obtain a procedure that determines
whether a logical obstruction exists and
gives a concrete obstruction formula of $\pfb$, if any.
\begin{enumerate}
  \item Calculate $\SB{0}, \SB{1}, \ldots$ until it does not change,
  \item Check whether $\SB{n}$ is total or not, at the stopped $n$,
  \item If it is total, there is no logical obstruction, and
  \item If it is not total at some $X$,
    $\PhiB(n, X)$ is a logical obstruction.
\end{enumerate}

\subsection{Constructing a logical obstruction for \textsc{know-all} model}

We apply the simulation technique developed so far
to show an impossibility result for \textsc{know-all} model~\cite{DBLP:journals/tcs/CastanedaFPRRT21}.
An instance of \textsc{know-all} model specifies a distributed system
of processes communicating via a network,
where the network connection varies
at each round of communication in a way given a priori.
Formally, it is specified by a sequence of directed graphs,
each of which represents the network connection of a particular round.
Each graph $G$ consists of the set $V(G)$ of graph nodes and
the set $E(G)$ of edges,
where $V(G)$ is the set of communicating processes $\Pi$ and
each edge $(p,q) \in E(G)$ means that
the process $p$ can transmit information to the process $q$.
We assume every process can transmit its value to itself at each round.
Hence every node $p$ in a graph is assumed to have an self-loop edge $(p,p)$.

In this section, we consider \textsc{know-all} model as a protocol.
An instance of \textsc{know-all} model is
a sequence of directed graphs whose set of nodes is $\Pi$.
An edge in a graph indicates that the nodes connected by the edge
are allowed to communicate.
Every node has the self-loop,
assuming every process can always communicate with itself.

\begin{definition}
  Let $p \in \Pi$ be an agent, $P \subseteq \Pi$ be a set of agents
  and $G, H$ be graphs.
  We define some operations for graphs as follows.
  \begin{itemize}
    \item Set of in-neighbors:
      $In(p, G) = \{q \in V(G) \mid (p, q) \in E(G) \}$,
      $In(P, G) = \bigcup_{p \in P} In(p, G)$
    \item Set of out-neighbors:
      $Out(p, G) = \{q \in V(G) \mid (q, p) \in E(G) \}$,
      $Out(P, G) = \bigcup_{p \in P} Out(p, G)$
    \item Composition of graphs:
      $H \circ G = (\Pi, E)$
      where
      $E = \{(p, q) \mid Out(p, G) \cap In(q, H) \neq \emptyset\}$
    \item Domination number of a graph:
    \[
      \gamma(G) = \min \{i \in \N \mid \text{there exists $P \subseteq \Pi$
        such that $|P| = i$ and $Out(P, G) = \Pi$}\}
    \]

  \end{itemize}

  Notice that $1 \leq \gamma(G) \leq |\Pi|$, since
  graphs are assumed to have all self-loops.
  We write $G_{\leq r}$ for
  $G_r \circ (G_{r-1} \circ (\dots (G_2 \circ G_1) \dots ))$.
\end{definition}

\begin{definition}
  Let $\mathcal{G} = (G_i)_{i \in \N}$ be an instance of
  \textsc{know-all} model and
  $r \geq 1$ be a natural number.
  We define an action model $\mathcal{G}_{\leq r} = (V, S, \chi, \pre)$,
  for \emph{a protocol associated with $\mathcal{G}$ after $r$ rounds}
  as follows:
  \begin{itemize}
    \item $V = \Pi \times V^{in} \times \Views^1$,
    \item $X \in \facet{\mathcal{G}_{\leq r}}$ if there exists
      a map $input \colon \Pi \to V^{in}$ such that
      $X = \{(a, input(a), v_a^r) \mid a \in \Pi\}$
      where $v_a^r \in \Views^1$ is defined as
      $v_a^r = \{(p, input(p)) \mid p \in In(a, G_{\leq r})\}$,
    \item $\chi(a, i, v) = a$ and
    \item $\pre(\{(a, i_a, v_a) \mid a \in \Pi\})
      = \bigwedge_{a \in \Pi} \ip_a^{i_a}$.
  \end{itemize}

  We write $\{(a, i_a, G_{\leq r}) \mid a \in \Pi\}$ for a facet
  $\{((a, i_a),(a, i_a, v_a)) \mid a \in \Pi\}$ in $\IGr$,
  where $v_a = \{(p, i_p) \mid p \in In(a, G_{\leq r})\}$.

\end{definition}

The following impossibility result has been shown by
a topological argument.

\begin{theorem}[\cite{DBLP:journals/tcs/CastanedaFPRRT21}]
  Let $\mathcal{G} = (G_i)_{i \in \N}$ be an instance of
  \textsc{know-all} model and $r \geq 0$.
  If $\gamma(G_{\leq r}) > k$,
  $\mathcal{SA}_k$ is not solvable by $\mathcal{G}_{\leq r}$.

\end{theorem}

In the following, we give an alternative proof to this impossibility result
by using the logical method presented in the previous section.
In contrast to~\cite{DBLP:journals/tcs/CastanedaFPRRT21},
this logical method allows a more elementary proof,
not resorting to sophisticated topological notions or tools.

\begin{proposition}
  Let $\S{1} \subseteq \facet{\IGr} \times \facet{\ISA{k}}$
  be the binary relation defined in Definition~\ref{def_n-simulation},
  and $X = \{(a, a, G_{\leq r}) \mid a \in \Pi\} \in \IGr$
  be a facet.
  If $\gamma(G_{\leq r}) > k$,
  then $\S{1}$ is not total at $X$,
  that is, there is no $X' \in \facet{\ISA{k}}$
  satisfying $X \mathrel{\SB{1}} X'$.
\end{proposition}
\begin{proof}
  Suppose $\S{1}$ is total at $X$, by contradiction.
  By applying the totality to $X$, there exists a facet
  $X' \in \ISA{k}$ such that $X \mathrel{\S{1}} X'$.
  For all $a \in \Pi$, let us write $input_X(a)$ for an input value and
  $view_X(a)$ for a view in $\Views^1$
  such that $(a, input_X(a), view_X(a))$ is an $a$-colored vertex in $X$.
  Similarly, let us write $(a, input_{X'}(a), decide_{X'}(a))$
  for an $a$-colored vertex in $X'$.

  We have $Out(decide_{X'}(\Pi), G_{\leq r}) \subsetneq \Pi$
  because $|decide_{X'}(\Pi)| \leq k < \gamma(G_{\leq r})$.
  Fix an agent $a_0 \in \Pi \setminus Out(decide_{X'}(\Pi), G_{\leq r})$
  and define a facet $Y \in \facet{\IGr}$ as
  $Y = \{(a, input_Y(a), G_{\leq r}) \mid a \in \Pi\}$, where
  $input_Y(a) = a$ if $a \in In(a_0, G_{\leq r})$
  and $input_Y(a) = a_0$ otherwise.
  Then, we have $input_Y(\Pi) \cap decide_{X'}(\Pi) = \emptyset$
  because $a_0 \notin Out(decide_{X'}(\Pi), G_{\leq r})$
  implies $In(a_0, G_{\leq r}) \cap decide_{X'}(\Pi) = \emptyset$.
  We also have $X \sim_{a_0} Y$,
  because $input_Y(a) = a = input_X(a)$ for all $a \in In(a_0, G_{\leq r})$.

  By the definition of $\S{1}$, there exists
  $Y' \in \ISA{k}$ such that $Y \mathrel{\S{0}} Y'$ and $X' \sim_{a_0} Y'$.
  Let us also write $(a, input_{Y'}(a), decide_{Y'}(a))$ for
  an $a$-colored vertex in $Y'$.
  Since $Y \mathrel{\S{0}} Y'$ implies
  $input_{Y'}(a) = input_Y(a)$ for all $a \in \Pi$
  and $Y' \in \ISA{k}$ implies $decide_{Y'}(\Pi) \subseteq input_{Y'}(\Pi)$,
  we have $decide_{Y'}(\Pi) \subseteq input_Y(\Pi)$.
  Hence $decide_{Y'}(\Pi) \cap decide_{X'}(\Pi) = \emptyset$,
  but $X' \sim_{a_0} Y'$ implies $decide_{X'}(a_0) = decide_{Y'}(a_0)$.
  This is a contradiction.
\end{proof}

\begin{theorem}
  Let $\mathcal{G} = (G_i)_{i \in \N}$ be an infinite sequence of
  directed graphs, $r \geq 0$ and $k > \gamma(G_{\leq r})$.
  There exists a logical obstruction to the solvability of
  $\mathcal{SA}_k$ by $\mathcal{G}_{\leq r}$.
  Especially, $\mathcal{SA}_k$ is not solvable by $\mathcal{G}_{\leq r}$.

\end{theorem}
\begin{proof}
  By Theorem~\ref{th_maximam_simulation},
  there exists $n \in \N$ such that $\S{n}$ is a $\Box$-simulation.
  By the previous proposition, this simulation is not total
  at $X = \{(a, a, G_{\leq r}) \mid a \in \Pi\} \in \IGr$.
  Thus, by Theorem~\ref{theorem_universal_logical_obstruction}
  we have $M' \models \PhiB(n, X)$.
  On the other hand, $M \not\models \PhiB(n, X)$
  by Corollary~\ref{cor_self_not_satisfy_ulo}.
\end{proof}

%% file: figures/50_counterexample_for_finite.tex

\begin{tikzpicture}
  \node[circle,scale=0.6,draw,fill=white,label=below:{$0$}] (c0v0) at (-4, 0) {};
  \node[circle,scale=0.6,draw,fill=white,label=above:{$1$}] (c0v1) at (-3, 1) {};
  \node[circle,scale=0.6,draw,fill=black,label=below:{$0$}] (c1v0) at (-3, 0) {};
  \node[circle,scale=0.6,draw,fill=black,label=above:{$1$}] (c1v1) at (-4, 1) {};
  \draw (c0v0) -- (c1v0) -- (c0v1) -- (c1v1) -- (c0v0);
  \node[circle,scale=0.6,draw,fill=white,label=below:{$0 \mapsto 0$}] (c0v0d0) at (-1.5, 0) {};
  \node[circle,scale=0.6,draw,fill=black,label=below:{$0 \mapsto 0$}] (c1v0d0) at (0, 0) {};
  \node[circle,scale=0.6,draw,fill=white,label=above:{$1 \mapsto 1$}] (c0v1d1) at (0, 1) {};
  \node[circle,scale=0.6,draw,fill=black,label=above:{$1 \mapsto 1$}] (c1v1d1) at (1.5, 1) {};
  \node[circle,scale=0.6,draw,fill=white,label=below:{$0 \mapsto 2$}] (c0v0d2) at (1.5, 0) {};
  \node[circle,scale=0.6,draw,fill=black,label=below:{$0 \mapsto 2$}] (c1v0d2) at (3, 0) {};
  \node[circle,scale=0.6,draw,fill=white,label=above:{$1 \mapsto 3$}] (c0v1d3) at (3, 1) {};
  \node[circle,scale=0.6,draw,fill=black,label=above:{$1 \mapsto 3$}] (c1v1d3) at (4.5, 1) {};
  \node[circle,scale=0.6,draw,fill=white,label=below:{$0 \mapsto 4$}] (c0v0d4) at (4.5, 0) {};
  \coordinate (x) at (5.5,0);
  \draw (c0v0d0) -- (c1v0d0);
  \draw (c1v0d0) -- (c0v1d1);
  \draw (c0v1d1) -- (c1v1d1);
  \draw (c1v1d1) -- (c0v0d2);
  \draw (c0v0d2) -- (c1v0d2);
  \draw (c1v0d2) -- (c0v1d3);
  \draw (c0v1d3) -- (c1v1d3);
  \draw (c1v1d3) -- (c0v0d4);
  \draw (c0v0d4) -- (x);
  \node (dots) at (6,0.5) {$\cdots$};
\end{tikzpicture}
\caption{Simplicial models $\mathcal{I}[P]$ (on the left) and
$\mathcal{I}[T]$ (on the right)}
\label{fig_counterexample}

%% file: 70_conclusion.tex

\section{Conclusion}  \label{sec_conclusion}

We proposed a general method,
based on a simulation technique,
for showing the non-existence of logical obstruction
to distributed task.
Using this method, we showed the non-existence of logical obstruction
for two certain situations,
whose unsolvability has been topologically established:
$k$-set agreement task by (iterated) immediate snapshot protocol
in $\mathcal{L}_K$ and
$3$ process $2$-set agreement task
by multi-round iterated immediate snapshot protocol in $\mathcal{L}_D$.
Furthermore, for a finite protocol and a finite task,
we provided a procedure that determines whether
a logical obstruction exists or not,
and constructs a logical obstruction, if any.
We demonstrated that this procedure provides a logical obstruction
to the solvability of $k$-set agreement task 
for the \textsc{know-all} model.

Although this paper studied epistemic languages with
$\ip_a^v$ for atomic propositions and
$\mathrm{K}_a$ or $\mathrm{D}_A$ for modal operators,
there remain attractive languages such as
the one equipped with a modal operator $\mathrm{C}_A$
for modality of common knowledge.
The simulation technique we developed in this paper
would merit further investigation into
such other epistemic languages as well,
and it is expected that investigating epistemic languages
will lead to better understanding of the logical method.
